\documentclass[11pt,tightenlines,nofootinbib,superscriptaddress]{revtex4-2}

\usepackage{minitoc}
\usepackage[utf8]{inputenc}
\usepackage[british]{babel}
\usepackage{amsmath,amssymb,amsthm,amscd}
\usepackage{mathtools}
\usepackage{mathrsfs}
\usepackage[active]{srcltx}
\usepackage{pgf,tikz}
\usepackage{mathrsfs}
\usepackage{enumerate}
\usepackage{braket}
\usetikzlibrary{arrows}
\usepackage{dsfont}
\usepackage{paralist}
\usepackage{comment}
\usepackage{enumerate}
\usepackage{soul}

\usepackage[resetlabels,labeled]{multibib}
\newcites{Methods}{Methods-only references}

\definecolor{navyblue}{rgb}{0.0, 0.0, 0.5}
\usepackage[colorlinks=true, urlcolor=blue,citecolor=purple,anchorcolor=blue, linkcolor=black]{hyperref}
\usepackage[margin=2.54cm,bmargin=2.54cm,tmargin=2.54cm]{geometry}
\usepackage{cleveref}

\newcommand{\ketbra}[2]{|#1\rangle\langle #2|}
\def\lb{\left(}
\def\rb{\right)}
\newcommand{\tr}[1]{\operatorname{tr}\lb#1\rb}

\newcommand{\cH}{\mathcal{H}}

\newcommand{\norm}[1]{\left\lVert#1\right\rVert}

\newcommand{\cB}{\mathcal{B}}

\renewcommand{\phi}{\varphi}
\renewcommand{\epsilon}{\varepsilon}

\theoremstyle{plain}
\newtheorem{theorem}{Theorem}[section]
\newtheorem{corollary}[theorem]{Corollary}
\newtheorem{lemma}[theorem]{Lemma}
\newtheorem{proposition}[theorem]{Proposition}

\theoremstyle{remark}

\newtheorem*{claim*}{Claim}
\newtheorem*{remark*}{Remark}
\newtheorem*{example*}{Example}
\newtheorem*{notation*}{Notation}
\numberwithin{equation}{section}

\makeatletter
\@namedef{subjclassname@2020}{\textup{2020} Mathematics Subject Classification}
\makeatother

\begin{document}

\author{\begingroup
\hypersetup{urlcolor=navyblue}
\href{https://orcid.org/0000-0001-7712-6582}{Cambyse Rouz\'{e}
\endgroup}
}
\email[Cambyse Rouz\'{e} ]{rouzecambyse@gmail.com}
 \affiliation{Inria, Télécom Paris - LTCI, Institut Polytechnique de Paris, 91120 Palaiseau, France}

\author{\begingroup
\hypersetup{urlcolor=navyblue}
\href{https://orcid.org/0000-0001-9699-5994}{Daniel Stilck Fran\c{c}a}
\endgroup}
\affiliation{Univ Lyon, ENS Lyon, UCBL, CNRS, Inria, LIP, F-69342, Lyon Cedex 07, France}
\affiliation{Department of Mathematical Sciences, University of Copenhagen, Universitetsparken 5, 2100 Denmark}
\email[Daniel Stilck Fran\c ca ]{dsfranca@math.ku.dk}

\author{\begingroup
\hypersetup{urlcolor=navyblue}
\href{https://orcid.org/0000-0002-5889-4022}{\'Alvaro M. Alhambra
\endgroup}
}
\email[\'Alvaro M. Alhambra ]{alvaro.alhambra@csic.es}
 \affiliation{Instituto de F\'{i}sica T\'{e}orica UAM/CSIC, C. Nicol\'{a}s Cabrera 13-15, Cantoblanco, 28049 Madrid, Spain}

\title[]{Efficient thermalization and universal quantum computing \\
with quantum Gibbs samplers}

\begin{abstract}

The preparation of thermal states of matter is a crucial task in quantum simulation. In this work, we prove that a recently introduced, efficiently implementable dissipative evolution thermalizes to the Gibbs state in time scaling polynomially with system size at high enough temperatures for any Hamiltonian that satisfies a Lieb-Robinson bound, such as local Hamiltonians on a lattice. Furthermore, we show the efficient adiabatic preparation of the associated purifications or ``thermofield double'' states. These results establish the efficient preparation of high-temperature Gibbs states and their purifications. In the low-temperature regime, we show that implementing this family of dissipative evolutions for inverse temperatures polynomial in the system's size is computationally equivalent to polynomial time quantum computations. On a technical level, for high temperatures, our proof makes use of the mapping of the generator of the evolution into a Hamiltonian, and then connecting its convergence to that of the infinite temperature limit. For low temperature, we instead perform a perturbation at zero temperature and resort to circuit-to-Hamiltonian mappings akin to the proof of universality of quantum adiabatic computing. Taken together, our results show that a family of quasi-local dissipative evolutions efficiently prepares a large class of quantum many-body states of interest, and has the potential to mirror the success of classical Monte Carlo methods for quantum many-body systems.

\end{abstract}

\maketitle
\doparttoc

\faketableofcontents

\newpage


\section{Introduction}

Markov Chain Monte Carlo (MCMC) methods are one of the go-to tools for sampling from classical Gibbs states (GS) of spin systems~\cite{levin2017markov}. They exhibit efficiency in practice~\cite{brooks2011handbook} and are also provably efficient in some cases, such as at high temperatures~\cite{martinelli1999lectures}. These methods traditionally rely on ``local updates'', where a few spins are flipped based on local energy changes, a feature crucial for both practical implementation and theoretical analysis. However, extending such algorithms to quantum systems has been a formidable challenge. Despite significant research efforts~\cite{Temme_2011,brandao2019finite,kastoryano2016quantum,bardet2021modified,capel2020modified,bardet2024entropy,bardet2023rapid,zhang2023dissipative,shtanko2021algorithms}, general quantum algorithms capable of (quasi-)local updates that also provably converge to a quantum GS have remained elusive, until the recent breakthroughs of~\cite{Chen2025}. 

In these recent works, the authors present a quantum algorithm that generates a Lindbladian $\mathcal{L}^{(\beta)}(\cdot)$ which is at the same time $i)$ quasi-local and $ii)$ reversible in the Heisenberg picture:
\begin{align*}
\langle X,\mathcal{L}^{(\beta)\dagger}(Y)\rangle_{\sigma_\beta}=\langle \mathcal{L}^{(\beta)\dagger}(X),Y\rangle_{\sigma_\beta}
\end{align*}
where $\langle X,Y\rangle_{\sigma_\beta}:=\operatorname{tr}(X^\dagger\sigma_\beta^{\frac{1}{2}}Y\sigma_\beta^{\frac{1}{2}})$ denotes the KMS inner product. In this case, $\sigma_\beta= {\operatorname{exp}(-\beta H)}/Z_\beta$ is the Gibbs state with partition function $Z_\beta:=\tr{e^{-\beta H}}$. Crucially, the reversibility guarantees that $\sigma_\beta$ is the fixed point of the evolution generated by $\mathcal{L}^{(\beta)}$. In particular, in \cite{Chen2025}, it was shown that this Lindbladian can be efficiently implemented with a quantum circuit. Details of those Lindbladians are shown in Section \ref{sec:gibbssampling} in the Methods section and in Appendix \ref{appendixGibbssampling}.

Here, we delve into the quantum algorithm proposed in~\cite{Chen2025}, and show its ability to efficiently prepare quantum states of interest. We investigate both the high-temperature and the low-temperature regimes. In the high-temperature regime, we show that above a certain constant threshold temperature,  the Lindbladians constructed in~\cite{Chen2025} efficiently converge to the GS for all Hamiltonians obeying a Lieb-Robinson bound, which includes local Hamiltonians on a lattice~\cite{nachtergaele2006lieb}. This then rigorously establishes the efficient preparation of high-temperature quantum Gibbs states on a quantum computer for any lattice dimension~\cite{brandao2019finite,kastoryano2016quantum,bardet2023rapid,capel2020modified}. 
Interestingly, our results also imply the efficient adiabatic preparation (starting with the state at inverse temperature $\beta=0$) of the associated purifications or thermofield doubles. The preparation of such purified states is of independent interest~\cite{cottrell_how_2019,sundar_proposal_2022,brown_quantum_2023}, and this result guarantees its efficiency for a general class of models.

For low temperatures, we show that implementing the Lindbladians of~\cite{Chen2025} for $\beta=\Omega(\text{poly}(n))$ gives rise to a universal model for quantum computing equivalent to the circuit model. To this end, we show that the Lindbladians corresponding to circuit-to-Hamiltonian mappings efficiently converge to a state with a polynomial overlap with the ground state corresponding to the output of the circuit. This state can then be leveraged to efficiently obtain the ground state through a series of local measurements. These results are in analogy with those of~\cite{chen2023local}, where the authors established results at polynomially larger temperatures and polynomial rounds of mid-circuit measurements, compared to one round at the end for our construction. Our construction also has a better dependency on the temperature, and does not require mid-circuit measurements. 

On a technical level, for high temperature, our analysis establishes a spectral gap for the Lindbladians defined in~\cite{Chen2025}. This is done by viewing the high-temperature generator as a perturbation of its infinite temperature counterpart, and employing results on quasi-local perturbations of gapped Hamiltonians~\cite{michalakis2013stability}. At low temperature, we also follow a perturbative approach, but now at $\beta=+\infty$. We are inspired by the classical results of~\cite{caputo2011zero} and the recent quantum results of~\cite{chen2023local} to show that the Lindbladians of~\cite{Chen2025} have the power to simulate polynomial-size quantum circuits  already at $\beta=\Omega(\text{poly}(n))$. Taken together, our findings show that the algorithm of~\cite{Chen2025} is a powerful tool to prepare non-trivial quantum many-body states, and to replicate the success of classical MCMC in the quantum realm.

\section{Results}

\subsection{Spectral gap for high-temperature quantum Gibbs sampler}\label{sec:highT}

We focus our attention on the physically-relevant class of geometrically local and short-range Hamiltonian $H$ on the lattice $\Lambda$. More precisely, such a Hamiltonian is defined through a function that maps any non-empty finite set $X\subset \Lambda$ to a self-adjoint element $h_X$ supported in $X$ with $\max_{X} \norm{h_X}_\infty \le h$. The Hamiltonian of the system in a region $Z \subset \Lambda$ is then defined as the sum over interactions included in $Z$:
\begin{align*}
		H_Z:=\sum_{X\subseteq Z}h_X\,.
\end{align*}
In addition, we assume that the Hamiltonian is such that no interaction $h_X$ has support on more than $k$ sites, and each site $u$ appears on at most $l$ non-zero operators $h_X$. We refer to such Hamiltonian as $(k,l)$-local. As is often done, we will make the implicit assumption that $H$ is one element of a family of Hamiltonians indexed by the lattice size $n$, in which case the constants $h,k$ and $l$ are assumed to be independent of $n$. We show that for any $\beta\leq\beta^*=\mathcal{O}((h  k l)^{-1})=\mathcal{O}(1)$, the spectral gap of $-\mathcal{L}^{(\beta)\dagger}$, i.e. the smallest nonzero eigenvalue, is constant.

\begin{theorem}\label{thmgaphighT1}
In the notations of the previous paragraph, for any $(k,l)$-local Hamiltonian $H$, there exists a constant $\beta^*>0$ independent of $n$ such that for $\beta \le \beta^*$, the spectral gap of $-{\mathcal{L}}^{(\beta)\dagger}$, is lower bounded by $\frac{1}{2\sqrt{2}e^{1/4}}$. 
\end{theorem}
By a standard argument recalled in Appendix \ref{sec:highTgibbsgap}, Theorem \ref{thmgaphighT1} implies the following guarantee on the convergence of the Gibbs sampler:
\begin{corollary}\label{cor:efficient_diss_prep}
For any $\epsilon>0$, the evolution generated by $\mathcal{L}^{(\beta)}$ gets $\epsilon$-close to $\sigma_\beta$ in polynomial time, i.e. for any initial state $\rho$,
\begin{align}\label{equ:mixing_gap}
\|e^{t\mathcal{L}^{(\beta)}}(\rho)-\sigma_\beta\|_1\le \epsilon\quad \text{ for all}\quad t=\Omega(\ln(1/\epsilon)+n)\,. 
\end{align}
Furthermore, we can prepare an $\epsilon$-approximation of $\sigma_\beta$ on a quantum computer with $\widetilde{\mathcal{O}}(n^2)$ Hamiltonian simulation time and $\widetilde{\mathcal{O}}(n^3)$ two qubit gates, where $\widetilde{\mathcal{O}}$ omits terms poly-logarithmic in $\epsilon^{-1}$ and $n$.
\end{corollary}

Given Equation \eqref{equ:mixing_gap}, the result on the efficient preparation on a quantum computer immediately follows from~\cite[Theorem 1.2]{Chen2025}, where the authors show how to efficiently simulate the Lindbladians of Equation \eqref{eq:generator1} on a quantum computer. The extra factor of $n$ in the $\widetilde{\mathcal{O}}(n^2)$ cost comes from the difference in normalization with respect to \cite[Theorem 1.2]{Chen2025}, and the additional $n$ in the number of gates is due to the spatial extension of the system. 
In the next Methods section, we explain the proof steps to obtain the bound on the spectral gap, which are further detailed in Appendix \ref{sec:highTgibbsgap}.

\subsection{Adiabatic preparation of the purified Gibbs state}

The proof of Theorem \ref{thmgaphighT1} involves a lower bound on the gap of the Lindbladian, done by mapping this Lindbladian to a Hamiltonian with the same gap, but in a doubled Hilbert space, which we define as $\widetilde{\mathcal{L}}^{(\beta)}$ (this is the \emph{discriminant} as defined in \cite{Chen2025}) . By construction, the Hamiltonian $\widetilde{\mathcal{L}}^{(\beta)}$ is quasi-local, and has the purified GS as a fixed point. This means that $\widetilde{\mathcal{L}}^{(\beta)}$ is a provable realization of Hamiltonians such as those suggested in e.g.  \cite{feiguin_hermitian_2013,cottrell_how_2019}.  

Additionally, the fact that we can control the gap starting from $\beta=0$ means that we can straightforwardly devise an adiabatic scheme to prepare the purified GS. The adiabatic path consists of lowering the temperature from $\beta=0$, so that our parametrized Hamiltonian is $\mathscr{H}(s)\equiv \widetilde{\mathcal{L}}_{s\beta}$ with $s \in [0,1]$, and $\beta \le \beta^*$. This way, by initially preparing the GS of $\mathscr{H}(0)$, which consists of a product of Bell pairs $\bigotimes_{i=1}^n \ket{\Phi^+}_i$, we can adiabatically sweep towards the purified Gibbs state (or \emph{thermofield double}) $\ket{\sqrt{\sigma_\beta}}$.

\begin{theorem}\label{thm:adiabaticmain}
In the notations of the previous paragraphs, and given $\beta^*$ as introduced in Theorem \ref{thmgaphighT1}, for any $\beta\in[0,\beta^*]$, the minimum time required to adiabatically prepare a state $\epsilon$-close to the purified Gibbs state at inverse temperature $\beta$ is 
\begin{align}\label{eq:adiabatictime}
T_{\operatorname{ad}}=  \mathcal{O} \left( \frac{  (\beta n )^3 }{\epsilon^2} \right)\,.
\end{align}
\end{theorem}

The proof follows from the lower bound of the gap from Theorem \ref{thmgaphighT1}, as well as standard estimates on the runtime of the adiabatic algorithm \cite{ambainis2006elementary}. To obtain the estimate of Equation \eqref{eq:adiabatictime}, in Appendix \ref{app:adiabatic} we give upper bounds to the first and second derivatives of the Hamiltonian $\mathscr{H}(s)$ with respect to the parameter $s$. If we simulate the adiabatic evolution as a circuit, the cost is $\widetilde{\mathcal{O}} \left(  T_{\operatorname{ad}} \times n \right)$ \cite{Haah_2021}.

The generator $\widetilde{\mathcal{L}}^{(\beta)}$ thus provides a polynomial-time scheme to prepare high temperature purified GS. Being able to prepare these  states is key in quantum simulations of certain models of entangled black holes \cite{brown_quantum_2023,nezami_quantum_2023}, and their generation allows for the measurement of OTOCs \cite{sundar_proposal_2022}. Because of this, many algorithms to prepare them have been devised \cite{wu_variational_2019,martyn_product_2019,chowdhury_variational_2020,Su_2021,failde_hamiltonian_2023}, and even implemented in small experiments \cite{zhu_generation_2020}. However, these are largely limited to variational approaches, without guarantees of efficiency or scalability. In contrast, the performance of our adiabatic algorithm at high temperatures is at most cubic in system size and quadratic in the inverse precision.

\subsection{Zero-temperature Gibbs sampling and universal quantum computing}\label{sec:uniQC}

Next, we consider the near-zero temperature regime, where we prove the existence of Gibbs samplers that reach states with inverse polynomial overlap with the ground states of $\mathsf{BQP}$-hard Hamiltonians at $\beta=\Omega(\text{poly}(n))$ in polynomial time. The result is contained in the following theorem.

\begin{theorem}\label{th:gibbsQP}
 The problem of approximating expectation values of local observables on the output of efficiently simulable low-temperature Gibbs samplers is $\mathsf{BQP}$-complete.
\end{theorem}

 We give the proof in the Methods section, with further details in Appendix \ref{GibbsQPsec}. The Gibbs samplers we consider are KMS Lindbladians that are nevertheless slightly different to those considered in Sec. \ref{sec:highT}, in that they involve Metropolis-like weights as considered in \cite{Chen2025}. 
 
 This result is reminiscent of the equivalence between decision problems decidable with polynomially-sized uniform families of quantum circuits and those decidable by means of unitary quasi-adiabatic evolutions ($\mathsf{BQP}=\mathsf{AdiabQP}$) as derived in the seminal paper \cite{aharonov2008adiabatic}, and later refined in \cite{oliveira2005complexity} by exhibiting a simpler circuit-to-Hamiltonian (CTH) construction with a family of $2$-local qubit Hamiltonians over a 2D grid (see also previous refinements in \cite{ambainis2006elementary,kempe2006complexity,oliveira2005complexity}).
 By analogy, Theorem \ref{th:gibbsQP} may allow us to define a dissipative variant of the well-known equivalence between the complexity class $\mathsf{BQP}$, and the class $\mathsf{AdiabQP}$ of those decidable by means of unitary quasi-adiabatic evolutions
\cite{aharonov2008adiabatic,ambainis2006elementary,kempe2006complexity,oliveira2005complexity}. To do this, one can potentially consider a class of problems solved by families of efficiently simulable Gibbs samplers \cite{ding2024efficientquantumgibbssamplers,scandi2025thermalizationopenmanybodysystems} containing the samplers used in Theorem \ref{th:gibbsQP}.

\medskip

The previous Theorem can be used to show a separation between classical and quantum computers for approximating expected values of fast mixing observables (see also \cite{kashyap2024accuracy} for related results): a local observable $O$ with $\|O\|_\infty\le 1$ is said to be fast mixing if the time-evolved observable $O_t:=e^{t\mathcal{L}^{(\sigma_E,\beta)\dagger}}(O)$ satisfies
\begin{align}
\big|\langle 0^n|O_{t}|0^n\rangle-\tr{\sigma_\beta O}\big|\le \epsilon
\end{align}
at $t=\operatorname{poly}(\beta,\frac{1}{\epsilon})$. From this definition, we obtain the following corollary.

\begin{corollary}\label{corrBQP}
There exists no $\operatorname{poly}(\beta,\epsilon^{-1})$-runtime classical algorithm outputing an $\epsilon$-additive approximation of the expected value $\tr{\sigma_\beta O}$ for any fast-mixing observable $O$ unless $\mathsf{BPP}=\mathsf{BQP}$.
\end{corollary}

\section{Comparison to previous work}\label{sec:comparison_previous}

There is an extensive literature on quantum Gibbs sampling on quantum computers. It includes numerous schemes with exponential runtimes for general Hamiltonians and temperatures \cite{poulin_sampling_2009,Temme_2011,chowdhury2016quantum,shtanko2021algorithms,Holmes_2022,zhang2023dissipative}, including a range of variational algorithms \cite{chowdhury_variational_2020,wang_variational_2021}. Others are aimed at mimicking the natural thermalization process of open systems \cite{chen2023fast,LiWang2023,Rall_2023}, for which a fast convergence is only guaranteed in restricted cases such as commuting Hamiltonians \cite{kastoryano2016quantum,bardet2023rapid}. In this regard, Corollary \ref{cor:efficient_diss_prep} together with the algorithm of \cite{Chen2025} give efficient (polynomial-time) preparation for general non-commuting models, even restricting to high temperatures. Note that Corollary \ref{cor:efficient_diss_prep} also works with exponentially decaying and non-commuting interactions, and likely also extends to certain types of polynomially decaying interactions. It also does not require any assumptions on the Gibbs state besides the temperature, unlike previous polynomial-time algorithms limited to 1D \cite{bilgin_preparing_2010} or commuting Hamiltonians \cite{GeMolnar2016}.  A result with a comparable degree of generality is that of~\cite{brandao2019finite}. However, there the authors need to assume conditions on the decay of correlations of the GS. Although these are believed to hold at high enough temperatures, they also require the application of gates acting on $\ln^D(n)$ qubits, where $D$ is the dimension of the lattice. Except for $D=1$, standard compilation results would then give a quasi-polynomial $e^{\mathcal{O}(\ln^D(n))}$ overhead to compile such gates, so the resulting algorithm is quasi-polynomial. Another algorithm with potential for efficiency at high temperatures is the adiabatic evolution proposed in the Appendix of \cite{GeMolnar2016} based on cluster expansions, but it is not clear whether it has a strictly polynomial runtime. 

It is important to mention that there are efficient classical algorithms for high temperature Gibbs states, with provably polynomial runtimes \cite{Helmuth2020,Mann2021,Haah_2024}. As such, it is currently unclear what is a range of temperatures for which Gibbs sampling may yield a provable quantum advantage. We note, however, that our bound on the mixing times is of interest despite it potentially being in a classically simulable regime. Since Lindbladians model the natural thermalization process \cite{Mozgunov_2020,Nathan_2020,Soret_2022,scandi2025thermalizationopenmanybodysystems}, establishing a fast mixing time is of physical interest beyond algorithmic or complexity considerations. This type of bound also allows us to establish that dissipative or mixed-state phase transitions do not happen in the high-temperature setting, helping us chart the map of dissipative and mixed-state phases \cite{Coser_2019,rakovszky2024definingstablephasesopen}. Additionally, being able to efficiently produce quantum Gibbs states and their purifications opens up the realizaton of quantum experiments at finite temperature (for instance, related to 2-point or 4-point dynamical correlation functions), and motivates the realization of quantum experiments at finite temperatures. 

After the completion of this work, we became aware of the results of~\cite{tang_efficient}, where the authors show that high-temperature Gibbs states can be prepared efficiently both on a quantum and classical computer. Their results have qualitative and quantitative differences to ours, which we now discuss in detail.
There, the authors showed that at high enough temperatures, Gibbs states can be approximated arbitrarily well by completely separable states that can be sampled efficiently. Thus, their result also gives both efficient classical and quantum algorithms to sample from such Gibbs states. Unfortunately, we currently cannot give bounds with explicit constants on the temperature ranges for which the Gibbs samplers are efficient. However, the critical temperature for our bounds scales inverse linearly with the locality of the Hamiltonian. In contrast, that of~\cite{tang_efficient} scales inverse quadratically. Thus, the Gibbs sampler is guaranteed to work at lower temperatures in the high locality limit. Furthermore, their classical algorithm to sample from the distribution has a runtime that scales like $\widetilde{\mathcal{O}}(n^{7+\log(\mathfrak{d})/\log(\beta_c/\beta)})$, where we have only included the dependency in $n$ and $\mathfrak{d}$ is the degree of the interaction graph. In contrast, our algorithm has no dependency on extra parameters of the problem and runs in time $\mathcal{O}(n^2)$ \,
, thus achieving a large polynomial speed-up compared to their methods. Their methods also apply to Hamiltonians beyond ours, such as $k$-local ones not on lattices, but our results only require Lieb-Robinson bounds, so they also extend to a range of long-range models \cite{TranLR2021} that their method does not cover. Furthermore, they also obtain as a corollary of their algorithm that these high temperature states are fully separable, a strong statement which does not follow from our results. It remains an interesting open question to understand if the temperature range covered by our results contains entangled Gibbs states.

The idea of solving $\mathsf{BQP}$-hard problems by engineering Lindbladians can be traced back to \cite{Verstraete2009,Krausetal}. However, the Lindbladians constructed there did not emulate the naturally occurring physical process of thermalization. In contrast to these early works, the authors of the recent paper \cite{chen2023local} developed a thermal gradient descent algorithm inspired by natural physical processes that generate low-energy states. This algorithm involves intertwining dissipation through Gibbs sampling, followed by measurements to confirm the reduction in energy at the output of the dissipation. While this alternative model of quantum computing is polynomially equivalent to the circuit model, the measurements at each step necessitate a possibly prohibitive overhead in practice. Moreover, the algorithm is guaranteed to work only at inverse temperatures that scale as $\beta=\mathcal{O}(T^{19})$ in the size $T$ of the circuit being emulated. 

In contrast, our approach is proved to work already at $\beta =\mathcal{O}(T^{11})$ for circuits of polynomial size, and does not require an active change of the dissipation depending on intermediate energy measurements: we simply need to let the system thermalize under a Lindbladian with a Hamiltonian corresponding to a CTH encoding for a polynomial amount of time. We then measure a series of local tests corresponding to the terms of another CTH mapping. If they all accept, which we show happens with high probability, we obtain the output of the circuit. Thus, repeating this procedure a few times suffices to reach the ground state.
Our proof is arguably much more straightforward than that of \cite{chen2023local} and, additionally, our procedure is also closer in spirit to the quantum adiabatic algorithm. Implementing the Lindbladians natively could provide a natural resilience to errors, due to the robustness of both local and quasi-local dissipative dynamics \cite{Cubitt2015,Borregaard2021,kim2017robust}. In contrast, we recall that the proof of a version of the threshold theorem for adiabatic quantum computing remains a major open problem to this date \cite{ChildsAdiabatic2001,AdiabaticQEC2013}.

\section{Methods}

\subsection{Details of the Lindbladian} \label{sec:gibbssampling}
 We now recall the generator of the quantum Gibbs sampler recently introduced in \cite{Chen2025}, with further details in Appendix \ref{appendixGibbssampling}. For $\beta>0$ and a local Hamiltonian $H$ over the lattice $\Lambda\equiv [0,L]^D$ with $n:=|\Lambda|$, we denote $\gamma(\omega):=\operatorname{exp}(-{(\beta\omega+1)^2}/{2})$ and we consider the Lindbladian
\begin{align}\label{eq:generator1}
&\mathcal{L}^{(\beta)}(\rho)=-i[B,\rho]+\sum_{a\in \Lambda,\alpha\in[3]}\int_{-\infty}^{\infty}\gamma(\omega)\mathcal{D}^{a,\alpha}_\omega(\rho)\,d\omega
\end{align}
where $\mathcal{D}^{a,\alpha}_\omega$ stands for the dissipative term corresponding to the jump operators
\begin{align}
A^{a,\alpha}(\omega):=\frac{1}{\sqrt{2\pi}}\int_{-\infty}^{\infty} e^{iHt}A^{a,\alpha} e^{-iH t}e^{-i\omega t}\,f(t)\,dt\quad 
\end{align}

with  $f(t):=\operatorname{exp}(-{t^2}/{\beta^2})\,\sqrt{\beta^{-1}\sqrt{2/\pi}}\,$.
Above, $A^{a,1}\equiv \sigma_x,\,A^{a,2}\equiv \sigma_y$ and $A^{a,3}\equiv \sigma_z$ denote the $1$-local Pauli matrices on site $a$. The coherent part is defined through the self-adjoint matrix
\begin{align}
&B=\beta^{-2}\sum_{a\in \Lambda,\alpha\in [3]}\, \int_{-\infty}^\infty b_1(t/\beta) e^{-i Ht}  
\int_{-\infty}^\infty b_2(t'/\beta) e^{iHt'}A^{a,\alpha} e^{-2i Ht'}A^{a,\alpha} e^{i Ht'}dt' e^{i Ht}dt
\end{align}
for some rapidly decaying functions $b_1,b_2$ with $\|b_1\|_{L_1}, \|b_2\|_{L_1}\le 1$ whose exact expression we recall in Appendix \ref{appendixGibbssampling}.

\medskip
\subsection{Spectral gap for high-temperature quantum Gibbs sampler}\label{sec:highT}
\paragraph{Mapping the Lindbladians to Hamiltonians:} for any $\beta>0$ we define as in~\cite{Chen2025}
\begin{align}\label{eq:Lbetatilde_main}
\widetilde{\mathcal{L}}^{(\beta)}(X):=\sigma_\beta^{-\frac{1}{4}}\mathcal{L}^{(\beta)}(\sigma_\beta^{\frac{1}{4}}X\sigma_\beta^{\frac{1}{4}} )\sigma_\beta^{-\frac{1}{4}}\equiv\sum_{a,\alpha}\widetilde{\mathcal{L}}_{a,\alpha}^{(\beta)}\,.
\end{align}
Letting $\ket{\sqrt{\sigma_\beta}}$ be the vectorization of $\sqrt{\sigma_\beta}$, we see that it is in the kernel of $\widetilde{\mathcal{L}}^{(\beta)}$. Furthermore, as $\widetilde{\mathcal{L}}^{(\beta)}$ is related to the original Lindbladian by a similarity transformation, they both share the same spectrum, and $\widetilde{\mathcal{L}}^{(\beta)}$ is self-adjoint with respect to the Hilbert-Schmidt inner product for all $\beta>0$, by KMS symmetry of $\mathcal{L}^{(\beta)\dagger}$. Thus, we can interpret $\widetilde{\mathcal{L}}^{(\beta)}$ as a Hamiltonian. Furthermore, the Hilbert space on which it is defined is independent of $\sigma_\beta$.

\medskip

\paragraph{Decomposing the generator:} we can write the generator as the following telescopic sum
\begin{align}\label{eq:localgen}
  \widetilde{\mathcal{L}}_{a,\alpha}^{(\beta)}=\widetilde{\mathcal{L}}_{a,\alpha}^{\beta,0}+\sum_{r=0}^{\infty}\widetilde{\mathcal{L}}_{a,\alpha}^{\beta,r+1}-\widetilde{\mathcal{L}}_{a,\alpha}^{\beta,r},
\end{align}
where for $r>0$ and a site $a$, $\widetilde{\mathcal{L}}_{a,\alpha}^{\beta,r+1}$ is defined as the generator for the state with the reduced Hamiltonian $H_{B_a(r)}$ associated to the ball $B_a(r)$ centered in site $a$ and of radius $r$. Thus, $\widetilde{\mathcal{L}}_{a,\alpha}^{\beta,r+1}-\widetilde{\mathcal{L}}_{a,\alpha}^{\beta,r}$ is strictly supported on $B_a(r+1)$. 
\medskip

\paragraph{The $\beta\to0$ limit:} by construction, for $\beta\to0$, the $\widetilde{\mathcal{L}}^{(\beta)}$ converges to the generator of the $1-$local depolarizing channel. When performing the mapping to Hamiltonians described in Equation \eqref{eq:Lbetatilde_main}, we thus obtain a Hamiltonian $\widetilde{H}_0$ that is frustration-free, gapped, and satisfies local topological quantum order (LTQO) as in~\cite{Bravyi_2010,michalakis2013stability}. It is straightforward to show these properties, so we leave the detailed derivations to the Appendix \ref{sec:gapped}. Showing the convergence to the depolarizing semigroup requires slightly more work, and we leave the details for Appendix \ref{sec:gapped}. In essence, the choice of the functions $f,b_1$ and $b_2$ ensures that they become delta functions in the limit $\beta\to0$. It is not difficult to see that at $\beta=0$, the matrices we integrate over are just proportional to the jump operators $A^{a,i}$, which are those of the depolarizing channel.

\paragraph{Quasi-locality of the generator:} the next crucial step is to bound the difference of generators of Eq. \eqref{eq:localgen} $\widetilde{\mathcal{E}}^{\beta,r}_{a,\alpha}:=\widetilde{\mathcal{L}}_{a,\alpha}^{\beta,r+1}-\widetilde{\mathcal{L}}_{a,\alpha}^{\beta,r}$ as a function of $r$ and $\beta$. By using Lieb-Robinson bounds, we find that
\begin{equation}
    \norm{\widetilde{\mathcal{E}}^{\beta,r}_{a,\alpha}}_{2 \rightarrow 2} = e^{-\Omega(r)},
\end{equation}
where the decay rate depends on $\beta$ and the Lieb-Robinson velocity. This means that we can see the finite $\beta$ contributions to $  \widetilde{\mathcal{L}}_{a,\alpha}^{(\beta)}$ as a quasi-local perturbation whose strength vanishes with lower $\beta $ (see Appendix \ref{sec:gapped} for further details). With this, we find that $  \widetilde{\mathcal{L}}_{a,\alpha}^{(\beta)}$ satisfies all the assumptions from \cite{michalakis2013stability}, from which we can guarantee the lower bound on the gap in Theorem \ref{thmgaphighT1}. The expression of $\beta^*$ then only depends on the Lieb-Robinson velocity of $H$, and the lattice dimension. The proof extends straightforwardly to systems with exponentially decaying interactions with similar Lieb-Robinson bounds \cite{HastingsKoma2006}, and we believe it should also extend to systems with polynomially decaying interactions \cite{TranLR2021,Review_Anthony_Chen_2023}.

\subsection{Zero-temperature Gibbs sampling and universal quantum computing}\label{sec:uniQC-M}

\paragraph{Proof of Theorem \ref{th:gibbsQP}}: Similar to the quantum adiabatic algorithm, our method consists of encoding the output states of a uniform, polynomially sized family of quantum circuits $C$ into the ground states of a uniform family of local Hamiltonians $H_C$ that can usually be decomposed into a clock part $H_{\textrm{clock}}$ and circuit part. The adiabatic theorem requires a path between $H_C$ and a trivial Hamiltonian with gap closing at most polynomially fast. Here, we instead need to show the rapid convergence 
towards states with enough overlap with the ground state of $H_C$ for a Gibbs sampling algorithm with the properly selected filter function $\gamma^M(\omega)$. 
The Hamiltonians in CTH constructions can be chosen to be frustration free and have ground state energy $0$. It then suffices to measure all energy terms and eventually obtain $0$ to ensure we have prepared the ground state. For this procedure to succeed with high probability, it suffices to produce an approximation of the GS with large overlap with the ground state.

It turns out that the Gaussian filter introduced in Equation \eqref{eq:generator1} does not work as well for low temperatures, since the corresponding Markov chain does not necessarily converge to the ground state. The Gaussian profile of the filter translates into the coefficient $\gamma(\omega)$ in Eq. \eqref{eq:generator1} also being Gaussian, which is a functional form that does not effectively favor transitions to lower energies. However, it was conjectured in \cite{Chen2025} that better mixing properties might be attained by means of linear combinations of Gaussian filters without loss of KMS symmetry. In particular, with such a linear combination we can generate the following filter 
\begin{align}\label{eq:metropo}
\gamma^M(\omega):=e^{-\beta \max\big(\omega+\frac{\beta\sigma_E^2}{2},0\big)}\,,
\end{align}
where $\sigma_E$ is a free parameter quantifying the energy resolution.  This is reminiscent of Metropolis sampling, and privileges transitions to lower energy levels.

Now, we consider the generator $\mathcal{L}^{(\sigma_E,\beta)}\equiv \mathcal{L}_H^{(\sigma_E,\beta)}$ constructed through $m$ jump operators $\{A^a\}_{a\in [m]}$ given by single-qubit Paulis as in Section \ref{sec:gibbssampling}, and Gaussian filter $\gamma$ replaced by the Metropolis filter $\gamma^M$ in Eq. \eqref{eq:metropo}. The associated Hamiltonian $H$ is on $n'=n+T$ qubits, with $T$ quadratic in the number of gates in the circuit, and has $M=\textrm{poly}(n+T)$ distinct eigenvalues. To show rapid convergence to the GS, we use several continuity bounds to connect the gap of $\mathcal{L}^{(\sigma_E,\beta)}$ (denoted as $\operatorname{gap}(\mathcal{L}^{(\sigma_E,\beta)\dagger})$) to that of a simpler generator we can control. 

First, in Lemma \ref{propperturbHH0} of Appendix \ref{GibbsQPsec}, we show the following perturbation bound for $\operatorname{gap}(\mathcal{L}^{(\sigma_E,\beta)\dagger})$ under changes in the Hamiltonian. Given $H=H_0+V$, a perturbed Hamiltonian with $\|V\|_\infty$ small enough, 
\begin{align}\label{H0toHbeta}
\operatorname{gap}({\mathcal{L}}^{(\sigma_E,\beta)\dagger}_{H})>\operatorname{gap}({\mathcal{L}}^{(\sigma_E,\beta)\dagger}_{H_0})-\delta''_{\beta,\sigma_E}\,,
\end{align}
where $\delta''_{\beta,\sigma_E}$ is proportional to $\norm{V}_\infty,m$ and to $\max\{\frac{\ln{\beta \sigma_E}}{\sigma_E},\frac{1}{\beta \sigma_E^2}\}$. In this case, we choose the full CTH for $V=H_C$, and only a classical ``clock" part of it as $H_0=H_{\operatorname{clock}}$, which has a significantly simpler spectrum.

We then control the gap of the generator corresponding to the Hamiltonian $H_{\operatorname{clock}}$ by connecting it to the gap in the limit of $\sigma_E \rightarrow 0$ and $\beta \rightarrow \infty$, which is done in two steps. For these to work, we exploit the additional structure of the  spectrum of $H_{\operatorname{clock}}$, namely that it is such that $M=\text{poly}(n')$ and that its smallest energy gap is $\Omega(1)$. In Lemma \ref{le:perturbsigmaE} of Appendix \ref{GibbsQPsec}, we prove that
\begin{align}
\operatorname{gap}(\mathcal{L}^{(\sigma_E,\beta)\dagger})\ge \operatorname{gap}(\mathcal{L}^{(0,\beta)\dagger})\big(1-\delta_{\beta,\sigma_E}\big)\,
\end{align}
with $\delta_{\beta ,\sigma_E}$ polynomial in $M,m,\sigma_E$ and exponentially decaying in $\beta$. This means that at low enough temperatures, we can estimate the gap by focusing on the limit $\sigma_E \rightarrow 0$. We then connect with the gap at zero temperature. In Lemma \ref{le:lemmainftytobeta} of Appendix \ref{GibbsQPsec} we show that 
\begin{align}
\operatorname{gap}(\mathcal{L}^{(0,\beta)\dagger})\ge \operatorname{gap}(\mathcal{L}^{(0,\infty)\dagger})-\delta_{\beta}'\,,
\end{align}
with $\delta'_\beta$ scaling polynomially with $M,m$ and decays exponentially in $\beta$. Thus, the gap of the generator $\mathcal{L}^{(0,\infty)\dagger}$ provides a good approximation to that of $\mathcal{L}^{(\sigma_E,\beta)\dagger}$ as long as $\beta=\Omega\left(\text{poly}(n')\right)$ and $\sigma_E=\mathcal{O}\left(1\right)$. 

The last part of the argument is a lower bound on $\operatorname{gap}(\mathcal{L}^{(0,\infty)\dagger})$, which is possible due to the simple structure of the spectrum of $H_{\operatorname{clock}}$ and of the generator $\mathcal{L}^{(0,\infty)\dagger}$ in the limit. To do this, we consider the explicit form of $\mathcal{L}^{(0,\infty)\dagger}$, which simplifies to a ``quasi-classical" Lindbladian in which the gap can be bounded with two key ingredients: Gershgorin's circle theorem and bounds on the gap of graph Laplacians, which appear due to the structure of the jump operators. More details can be found in Proposition \ref{propgap0inftyclock} of Appendix \ref{GibbsQPsec}. Putting everything together, we find that there are regimes of large $\beta$, constant $\sigma_E$ and small $\norm{V}_\infty$ within which the continuity bounds allows us to lower bound $\operatorname{gap}({\mathcal{L}}^{(\sigma_E,\beta)\dagger}_{H_C})$.

This argument shows that one can obtain a low temperature GS $\frac{e^{-\beta H_C}}{Z_\beta}$ in polynomial time. To prove that it has high overlap with the ground state at large enough $\beta$, we control the spectral gap of $H_C$ using once again the fact that it is a perturbation of a simpler Hamiltonian $H_{\operatorname{clock}}$. This means that picking the temperature such that $\beta=\Omega \left((n')^{11}\right)$, and adjusting the choice $\norm{V}_\infty=\mathcal{O}\left(1/(n')^5 \ln{(n')}\right)$ accordingly, considering the constraints from the continuity bounds, yields the desired polynomial overlap (See more details in Appendix \ref{thmCproof}). Although the resulting Lindbladian is not necessarily quasi-local like in the high-temperature case, the results of~\cite{Chen2025} show that it can still be simulated efficiently on a quantum computer.

\medskip
\paragraph{Proof of Corollary \ref{corrBQP}}

Assume such algorithm $\mathcal{A}_{\operatorname{cl}}$ exists and take our Gibbs sampler $\mathcal{L}^{(\sigma_E,\beta)}$ encoding a poly-sized circuit $U$. By the proof of Theorem \ref{th:gibbsQP} sketched in the above paragraph, for $\beta=\Theta((n')^{11})$ and $\sigma_E=\Theta(1)$, the generator $\mathcal{L}^{(\sigma_E,\beta)\dagger}$ has spectral gap $\lambda=\Omega((n')^{-2})$. Since $\|H\|_\infty\le \beta T=\mathcal{O}((n')^{12}) $, this implies that 
\begin{align}
\big|\langle 0^n|O_{t}|0^n\rangle-\tr{\sigma_\beta O}\big|\le \epsilon \|O\|_\infty
\end{align}
for $t=\operatorname{poly}((n')^{14},\log(\epsilon^{-1}))=\operatorname{poly}(\beta,\epsilon^{-1})$, by a standard proof of fast convergence from gap estimates (see Corollary \ref{coro1} of the Appendix for a precise statement). Hence any observable $O$ is fast mixing for the evolution considered. Now choosing $O$ as a $1$-qubit POVM element encoding the solution of a the circuit $U$, if $\mathcal{A}_{\operatorname{cl}}$ could return an $\epsilon$-additive approximation of the thermal average $\tr{\sigma_\beta O}$, it would solve the problem encoded in $U$, which would imply $\mathsf{BPP}=\mathsf{BQP}$. 


\section{Conclusion}
In this work, we have established the efficient dissipative preparation of high-temperature GS and their purifications, and showed that the class of Lindbladian evolutions of~\cite{Chen2025} is BQP complete at low temperatures. While the latter should be taken more as a proof-of-principle result, we hope that further work can improve the construction. Taken together, our results show that the Lindbladians of~\cite{Chen2025} have the full potential to replicate the success of MCMC for quantum many-body states. It would also be interesting to improve the results presented here by showing a modified log-Sobolev inequality~\cite{bardet2023rapid,bardet2024entropy,capel2020modified} for the generators instead of the spectral gap, as this would provide convergence not only in polynomial but even logarithmic time.
Furthermore, implementing the Lindbladian of~\cite{Chen2025} still requires sophisticated quantum circuits. It would be interesting to see if in our regimes it is possible to implement the generators more simply, perhaps even through a microscopic model of system-environment interactions such as the one for Davies maps \cite{Davies1974}. This would be particularly interesting for near-term computation, as it is known that dissipative preparation can be more robust than simulating unitary dynamics~\cite{Cubitt2015,Borregaard2021,kim2017robust}. More generally, it would be interesting to investigate further how errors affect these time evolutions, and for instance, whether a computational model motivated by Theorem \ref{th:gibbsQP} admits a threshold theorem.

\begin{acknowledgements}
The authors acknowledge useful discussions with Anthony Chen, and thank him for finding a gap in a previous version of Theorem \ref{th:gibbsQP}.
AMA acknowledges useful discussions with Tim Hsieh and David Gosset. CR would like to thank Simone Warzel for fruitful discussions on Hamiltonian gaps, and acknowledges financial support from the ANR project QTraj (ANR-20-CE40-0024-01) of the French National Research Agency (ANR). DSF acknowledges funding from the European Union under Grant Agreement 101080142 and the project EQUALITY and from the Novo Nordisk
Foundation (Grant No. NNF20OC0059939 Quantum for Life). AMA acknowledges support from the Spanish Agencia Estatal de Investigacion through the grants ``IFT Centro de Excelencia Severo Ochoa CEX2020-001007-S", ``Proyecto de Colaboraci\'on Internacional PCI2024-153448", ``PID2023-150847NA-I00"  and ``Ram\'on y Cajal RyC2021-031610-I'', financed by MCIN/AEI/10.13039/501100011033 and the European Union NextGenerationEU/PRTR. This project was
funded within the QuantERA II Programme that has re-
ceived funding from the EU’s H2020 research and inno-
vation programme under the GA No 101017733. Part of the manuscript was previously publiched as a STOC abstract \cite{STOC}.
\end{acknowledgements}

\bibliographystyle{apsrev4-2}
\bibliography{references}

\begin{thebibliography}{75}%
\makeatletter
\providecommand \@ifxundefined [1]{%
 \@ifx{#1\undefined}
}%
\providecommand \@ifnum [1]{%
 \ifnum #1\expandafter \@firstoftwo
 \else \expandafter \@secondoftwo
 \fi
}%
\providecommand \@ifx [1]{%
 \ifx #1\expandafter \@firstoftwo
 \else \expandafter \@secondoftwo
 \fi
}%
\providecommand \natexlab [1]{#1}%
\providecommand \enquote  [1]{``#1''}%
\providecommand \bibnamefont  [1]{#1}%
\providecommand \bibfnamefont [1]{#1}%
\providecommand \citenamefont [1]{#1}%
\providecommand \href@noop [0]{\@secondoftwo}%
\providecommand \href [0]{\begingroup \@sanitize@url \@href}%
\providecommand \@href[1]{\@@startlink{#1}\@@href}%
\providecommand \@@href[1]{\endgroup#1\@@endlink}%
\providecommand \@sanitize@url [0]{\catcode `\\12\catcode `\$12\catcode
  `\&12\catcode `\#12\catcode `\^12\catcode `\_12\catcode `\%12\relax}%
\providecommand \@@startlink[1]{}%
\providecommand \@@endlink[0]{}%
\providecommand \url  [0]{\begingroup\@sanitize@url \@url }%
\providecommand \@url [1]{\endgroup\@href {#1}{\urlprefix }}%
\providecommand \urlprefix  [0]{URL }%
\providecommand \Eprint [0]{\href }%
\providecommand \doibase [0]{https://doi.org/}%
\providecommand \selectlanguage [0]{\@gobble}%
\providecommand \bibinfo  [0]{\@secondoftwo}%
\providecommand \bibfield  [0]{\@secondoftwo}%
\providecommand \translation [1]{[#1]}%
\providecommand \BibitemOpen [0]{}%
\providecommand \bibitemStop [0]{}%
\providecommand \bibitemNoStop [0]{.\EOS\space}%
\providecommand \EOS [0]{\spacefactor3000\relax}%
\providecommand \BibitemShut  [1]{\csname bibitem#1\endcsname}%
\let\auto@bib@innerbib\@empty
\bibitem [{\citenamefont {Levin}\ and\ \citenamefont
  {Peres}(2017)}]{levin2017markov}%
  \BibitemOpen
  \bibfield  {author} {\bibinfo {author} {\bibfnamefont {D.~A.}\ \bibnamefont
  {Levin}}\ and\ \bibinfo {author} {\bibfnamefont {Y.}~\bibnamefont {Peres}},\
  }\href@noop {} {\emph {\bibinfo {title} {Markov chains and mixing times}}},\
  Vol.\ \bibinfo {volume} {107}\ (\bibinfo  {publisher} {American Mathematical
  Soc.},\ \bibinfo {year} {2017})\BibitemShut {NoStop}%
\bibitem [{\citenamefont {Brooks}\ \emph {et~al.}(2011)\citenamefont {Brooks},
  \citenamefont {Gelman}, \citenamefont {Jones},\ and\ \citenamefont
  {Meng}}]{brooks2011handbook}%
  \BibitemOpen
  \bibfield  {author} {\bibinfo {author} {\bibfnamefont {S.}~\bibnamefont
  {Brooks}}, \bibinfo {author} {\bibfnamefont {A.}~\bibnamefont {Gelman}},
  \bibinfo {author} {\bibfnamefont {G.}~\bibnamefont {Jones}},\ and\ \bibinfo
  {author} {\bibfnamefont {X.-L.}\ \bibnamefont {Meng}},\ }\href@noop {} {\emph
  {\bibinfo {title} {Handbook of markov chain monte carlo}}}\ (\bibinfo
  {publisher} {CRC press},\ \bibinfo {year} {2011})\BibitemShut {NoStop}%
\bibitem [{\citenamefont {Martinelli}(1999)}]{martinelli1999lectures}%
  \BibitemOpen
  \bibfield  {author} {\bibinfo {author} {\bibfnamefont {F.}~\bibnamefont
  {Martinelli}},\ }\href@noop {} {\bibfield  {journal} {\bibinfo  {journal}
  {Lectures on probability theory and statistics (Saint-Flour, 1997)}\ }\textbf
  {\bibinfo {volume} {1717}},\ \bibinfo {pages} {93} (\bibinfo {year}
  {1999})}\BibitemShut {NoStop}%
\bibitem [{\citenamefont {Temme}\ \emph {et~al.}(2011)\citenamefont {Temme},
  \citenamefont {Osborne}, \citenamefont {Vollbrecht}, \citenamefont {Poulin},\
  and\ \citenamefont {Verstraete}}]{Temme_2011}%
  \BibitemOpen
  \bibfield  {author} {\bibinfo {author} {\bibfnamefont {K.}~\bibnamefont
  {Temme}}, \bibinfo {author} {\bibfnamefont {T.~J.}\ \bibnamefont {Osborne}},
  \bibinfo {author} {\bibfnamefont {K.~G.}\ \bibnamefont {Vollbrecht}},
  \bibinfo {author} {\bibfnamefont {D.}~\bibnamefont {Poulin}},\ and\ \bibinfo
  {author} {\bibfnamefont {F.}~\bibnamefont {Verstraete}},\ }\href
  {https://doi.org/10.1038/nature09770} {\bibfield  {journal} {\bibinfo
  {journal} {Nature}\ }\textbf {\bibinfo {volume} {471}},\ \bibinfo {pages}
  {87–90} (\bibinfo {year} {2011})}\BibitemShut {NoStop}%
\bibitem [{\citenamefont {Brand{\~a}o}\ and\ \citenamefont
  {Kastoryano}(2019)}]{brandao2019finite}%
  \BibitemOpen
  \bibfield  {author} {\bibinfo {author} {\bibfnamefont {F.~G. S.~L.}\
  \bibnamefont {Brand{\~a}o}}\ and\ \bibinfo {author} {\bibfnamefont {M.~J.}\
  \bibnamefont {Kastoryano}},\ }\href
  {https://doi.org/10.1007/s00220-018-3150-8} {\bibfield  {journal} {\bibinfo
  {journal} {Communications in Mathematical Physics}\ }\textbf {\bibinfo
  {volume} {365}},\ \bibinfo {pages} {1} (\bibinfo {year} {2019})}\BibitemShut
  {NoStop}%
\bibitem [{\citenamefont {Kastoryano}\ and\ \citenamefont
  {Brand{\~a}o}(2016)}]{kastoryano2016quantum}%
  \BibitemOpen
  \bibfield  {author} {\bibinfo {author} {\bibfnamefont {M.~J.}\ \bibnamefont
  {Kastoryano}}\ and\ \bibinfo {author} {\bibfnamefont {F.~G. S.~L.}\
  \bibnamefont {Brand{\~a}o}},\ }\href
  {https://doi.org/10.1007/s00220-016-2641-8} {\bibfield  {journal} {\bibinfo
  {journal} {Communications in Mathematical Physics}\ }\textbf {\bibinfo
  {volume} {344}},\ \bibinfo {pages} {915} (\bibinfo {year}
  {2016})}\BibitemShut {NoStop}%
\bibitem [{\citenamefont {Bardet}\ \emph {et~al.}(2021)\citenamefont {Bardet},
  \citenamefont {Capel}, \citenamefont {Lucia}, \citenamefont
  {P\'erez-Garc\'ia},\ and\ \citenamefont {Rouz\'e}}]{bardet2021modified}%
  \BibitemOpen
  \bibfield  {author} {\bibinfo {author} {\bibfnamefont {I.}~\bibnamefont
  {Bardet}}, \bibinfo {author} {\bibfnamefont {A.}~\bibnamefont {Capel}},
  \bibinfo {author} {\bibfnamefont {A.}~\bibnamefont {Lucia}}, \bibinfo
  {author} {\bibfnamefont {D.}~\bibnamefont {P\'erez-Garc\'ia}},\ and\ \bibinfo
  {author} {\bibfnamefont {C.}~\bibnamefont {Rouz\'e}},\ }\href
  {https://doi.org/10.1063/1.5142186} {\bibfield  {journal} {\bibinfo
  {journal} {Journal of Mathematical Physics}\ }\textbf {\bibinfo {volume}
  {62}},\ \bibinfo {pages} {061901} (\bibinfo {year} {2021})},\ \Eprint
  {https://arxiv.org/abs/https://doi.org/10.1063/1.5142186}
  {https://doi.org/10.1063/1.5142186} \BibitemShut {NoStop}%
\bibitem [{\citenamefont {Ángela Capel}\ \emph {et~al.}(2021)\citenamefont
  {Ángela Capel}, \citenamefont {Rouzé},\ and\ \citenamefont
  {França}}]{capel2020modified}%
  \BibitemOpen
  \bibfield  {author} {\bibinfo {author} {\bibnamefont {Ángela Capel}},
  \bibinfo {author} {\bibfnamefont {C.}~\bibnamefont {Rouzé}},\ and\ \bibinfo
  {author} {\bibfnamefont {D.~S.}\ \bibnamefont {França}},\ }\href@noop {}
  {\bibinfo {title} {The modified logarithmic sobolev inequality for quantum
  spin systems: classical and commuting nearest neighbour interactions}}
  (\bibinfo {year} {2021}),\ \Eprint {https://arxiv.org/abs/2009.11817}
  {arXiv:2009.11817 [quant-ph]} \BibitemShut {NoStop}%
\bibitem [{\citenamefont {Bardet}\ \emph {et~al.}(2024)\citenamefont {Bardet},
  \citenamefont {Capel}, \citenamefont {Gao}, \citenamefont {Lucia},
  \citenamefont {P{\'e}rez-Garc{\'i}a},\ and\ \citenamefont
  {Rouz{\'e}}}]{bardet2024entropy}%
  \BibitemOpen
  \bibfield  {author} {\bibinfo {author} {\bibfnamefont {I.}~\bibnamefont
  {Bardet}}, \bibinfo {author} {\bibfnamefont {{\'A}.}~\bibnamefont {Capel}},
  \bibinfo {author} {\bibfnamefont {L.}~\bibnamefont {Gao}}, \bibinfo {author}
  {\bibfnamefont {A.}~\bibnamefont {Lucia}}, \bibinfo {author} {\bibfnamefont
  {D.}~\bibnamefont {P{\'e}rez-Garc{\'i}a}},\ and\ \bibinfo {author}
  {\bibfnamefont {C.}~\bibnamefont {Rouz{\'e}}},\ }\href
  {https://doi.org/10.1007/s00220-023-04869-5} {\bibfield  {journal} {\bibinfo
  {journal} {Communications in Mathematical Physics}\ }\textbf {\bibinfo
  {volume} {405}},\ \bibinfo {pages} {42} (\bibinfo {year} {2024})}\BibitemShut
  {NoStop}%
\bibitem [{\citenamefont {Bardet}\ \emph {et~al.}(2023)\citenamefont {Bardet},
  \citenamefont {Capel}, \citenamefont {Gao}, \citenamefont {Lucia},
  \citenamefont {P{\'{e}}rez-Garc{\'{\i}}a},\ and\ \citenamefont
  {Rouz{\'{e}}}}]{bardet2023rapid}%
  \BibitemOpen
  \bibfield  {author} {\bibinfo {author} {\bibfnamefont {I.}~\bibnamefont
  {Bardet}}, \bibinfo {author} {\bibfnamefont {{\'{A} }.}~\bibnamefont
  {Capel}}, \bibinfo {author} {\bibfnamefont {L.}~\bibnamefont {Gao}}, \bibinfo
  {author} {\bibfnamefont {A.}~\bibnamefont {Lucia}}, \bibinfo {author}
  {\bibfnamefont {D.}~\bibnamefont {P{\'{e}}rez-Garc{\'{\i}}a}},\ and\ \bibinfo
  {author} {\bibfnamefont {C.}~\bibnamefont {Rouz{\'{e}}}},\ }\bibfield
  {journal} {\bibinfo  {journal} {Physical Review Letters}\ }\textbf {\bibinfo
  {volume} {130}},\ \href {https://doi.org/10.1103/physrevlett.130.060401}
  {10.1103/physrevlett.130.060401} (\bibinfo {year} {2023})\BibitemShut
  {NoStop}%
\bibitem [{\citenamefont {Zhang}\ \emph {et~al.}(2023)\citenamefont {Zhang},
  \citenamefont {Bosse},\ and\ \citenamefont {Cubitt}}]{zhang2023dissipative}%
  \BibitemOpen
  \bibfield  {author} {\bibinfo {author} {\bibfnamefont {D.}~\bibnamefont
  {Zhang}}, \bibinfo {author} {\bibfnamefont {J.~L.}\ \bibnamefont {Bosse}},\
  and\ \bibinfo {author} {\bibfnamefont {T.}~\bibnamefont {Cubitt}},\
  }\href@noop {} {\bibinfo {title} {Dissipative quantum gibbs sampling}}
  (\bibinfo {year} {2023}),\ \Eprint {https://arxiv.org/abs/2304.04526}
  {arXiv:2304.04526 [quant-ph]} \BibitemShut {NoStop}%
\bibitem [{\citenamefont {Shtanko}\ and\ \citenamefont
  {Movassagh}(2023)}]{shtanko2021algorithms}%
  \BibitemOpen
  \bibfield  {author} {\bibinfo {author} {\bibfnamefont {O.}~\bibnamefont
  {Shtanko}}\ and\ \bibinfo {author} {\bibfnamefont {R.}~\bibnamefont
  {Movassagh}},\ }\href@noop {} {\bibinfo {title} {Preparing thermal states on
  noiseless and noisy programmable quantum processors}} (\bibinfo {year}
  {2023}),\ \Eprint {https://arxiv.org/abs/2112.14688} {arXiv:2112.14688
  [quant-ph]} \BibitemShut {NoStop}%
\bibitem [{\citenamefont {Chen}\ \emph
  {et~al.}(2025{\natexlab{a}})\citenamefont {Chen}, \citenamefont {Kastoryano},
  \citenamefont {Brand{\~a}o},\ and\ \citenamefont {Gily{\'e}n}}]{Chen2025}%
  \BibitemOpen
  \bibfield  {author} {\bibinfo {author} {\bibfnamefont {C.-F.}\ \bibnamefont
  {Chen}}, \bibinfo {author} {\bibfnamefont {M.}~\bibnamefont {Kastoryano}},
  \bibinfo {author} {\bibfnamefont {F.~G. S.~L.}\ \bibnamefont {Brand{\~a}o}},\
  and\ \bibinfo {author} {\bibfnamefont {A.}~\bibnamefont {Gily{\'e}n}},\
  }\href {https://doi.org/10.1038/s41586-025-09583-x} {\bibfield  {journal}
  {\bibinfo  {journal} {Nature}\ }\textbf {\bibinfo {volume} {646}},\ \bibinfo
  {pages} {561} (\bibinfo {year} {2025}{\natexlab{a}})}\BibitemShut {NoStop}%
\bibitem [{\citenamefont {Nachtergaele}\ and\ \citenamefont
  {Sims}(2006)}]{nachtergaele2006lieb}%
  \BibitemOpen
  \bibfield  {author} {\bibinfo {author} {\bibfnamefont {B.}~\bibnamefont
  {Nachtergaele}}\ and\ \bibinfo {author} {\bibfnamefont {R.}~\bibnamefont
  {Sims}},\ }\href {https://doi.org/10.1007/s00220-006-1556-1} {\bibfield
  {journal} {\bibinfo  {journal} {Communications in Mathematical Physics}\
  }\textbf {\bibinfo {volume} {265}},\ \bibinfo {pages} {119–130} (\bibinfo
  {year} {2006})}\BibitemShut {NoStop}%
\bibitem [{\citenamefont {Cottrell}\ \emph {et~al.}(2019)\citenamefont
  {Cottrell}, \citenamefont {Freivogel}, \citenamefont {Hofman},\ and\
  \citenamefont {Lokhande}}]{cottrell_how_2019}%
  \BibitemOpen
  \bibfield  {author} {\bibinfo {author} {\bibfnamefont {W.}~\bibnamefont
  {Cottrell}}, \bibinfo {author} {\bibfnamefont {B.}~\bibnamefont {Freivogel}},
  \bibinfo {author} {\bibfnamefont {D.~M.}\ \bibnamefont {Hofman}},\ and\
  \bibinfo {author} {\bibfnamefont {S.~F.}\ \bibnamefont {Lokhande}},\ }\href
  {https://doi.org/10.1007/JHEP02(2019)058} {\bibfield  {journal} {\bibinfo
  {journal} {Journal of High Energy Physics}\ }\textbf {\bibinfo {volume}
  {2019}},\ \bibinfo {pages} {58} (\bibinfo {year} {2019})},\ \bibinfo {note}
  {arXiv:1811.11528 [cond-mat, physics:hep-th, physics:quant-ph]}\BibitemShut
  {NoStop}%
\bibitem [{\citenamefont {Sundar}\ \emph {et~al.}(2022)\citenamefont {Sundar},
  \citenamefont {Elben}, \citenamefont {Joshi},\ and\ \citenamefont
  {Zache}}]{sundar_proposal_2022}%
  \BibitemOpen
  \bibfield  {author} {\bibinfo {author} {\bibfnamefont {B.}~\bibnamefont
  {Sundar}}, \bibinfo {author} {\bibfnamefont {A.}~\bibnamefont {Elben}},
  \bibinfo {author} {\bibfnamefont {L.~K.}\ \bibnamefont {Joshi}},\ and\
  \bibinfo {author} {\bibfnamefont {T.~V.}\ \bibnamefont {Zache}},\ }\href
  {https://doi.org/10.1088/1367-2630/ac5002} {\bibfield  {journal} {\bibinfo
  {journal} {New Journal of Physics}\ }\textbf {\bibinfo {volume} {24}},\
  \bibinfo {pages} {023037} (\bibinfo {year} {2022})},\ \bibinfo {note}
  {arXiv:2107.02196 [cond-mat, physics:nlin, physics:quant-ph]}\BibitemShut
  {NoStop}%
\bibitem [{\citenamefont {Brown}\ \emph {et~al.}(2023)\citenamefont {Brown},
  \citenamefont {Gharibyan}, \citenamefont {Leichenauer}, \citenamefont {Lin},
  \citenamefont {Nezami}, \citenamefont {Salton}, \citenamefont {Susskind},
  \citenamefont {Swingle},\ and\ \citenamefont {Walter}}]{brown_quantum_2023}%
  \BibitemOpen
  \bibfield  {author} {\bibinfo {author} {\bibfnamefont {A.~R.}\ \bibnamefont
  {Brown}}, \bibinfo {author} {\bibfnamefont {H.}~\bibnamefont {Gharibyan}},
  \bibinfo {author} {\bibfnamefont {S.}~\bibnamefont {Leichenauer}}, \bibinfo
  {author} {\bibfnamefont {H.~W.}\ \bibnamefont {Lin}}, \bibinfo {author}
  {\bibfnamefont {S.}~\bibnamefont {Nezami}}, \bibinfo {author} {\bibfnamefont
  {G.}~\bibnamefont {Salton}}, \bibinfo {author} {\bibfnamefont
  {L.}~\bibnamefont {Susskind}}, \bibinfo {author} {\bibfnamefont
  {B.}~\bibnamefont {Swingle}},\ and\ \bibinfo {author} {\bibfnamefont
  {M.}~\bibnamefont {Walter}},\ }\href
  {https://doi.org/10.1103/PRXQuantum.4.010320} {\bibfield  {journal} {\bibinfo
   {journal} {PRX Quantum}\ }\textbf {\bibinfo {volume} {4}},\ \bibinfo {pages}
  {010320} (\bibinfo {year} {2023})},\ \bibinfo {note} {publisher: American
  Physical Society}\BibitemShut {NoStop}%
\bibitem [{\citenamefont {Chen}\ \emph
  {et~al.}(2025{\natexlab{b}})\citenamefont {Chen}, \citenamefont {Huang},
  \citenamefont {Preskill},\ and\ \citenamefont {Zhou}}]{chen2023local}%
  \BibitemOpen
  \bibfield  {author} {\bibinfo {author} {\bibfnamefont {C.-F.}\ \bibnamefont
  {Chen}}, \bibinfo {author} {\bibfnamefont {H.-Y.}\ \bibnamefont {Huang}},
  \bibinfo {author} {\bibfnamefont {J.}~\bibnamefont {Preskill}},\ and\
  \bibinfo {author} {\bibfnamefont {L.}~\bibnamefont {Zhou}},\ }\href
  {https://doi.org/10.1038/s41567-025-02781-4} {\bibfield  {journal} {\bibinfo
  {journal} {Nature Physics}\ }\textbf {\bibinfo {volume} {21}},\ \bibinfo
  {pages} {654} (\bibinfo {year} {2025}{\natexlab{b}})}\BibitemShut {NoStop}%
\bibitem [{\citenamefont {Michalakis}\ and\ \citenamefont
  {Zwolak}(2013)}]{michalakis2013stability}%
  \BibitemOpen
  \bibfield  {author} {\bibinfo {author} {\bibfnamefont {S.}~\bibnamefont
  {Michalakis}}\ and\ \bibinfo {author} {\bibfnamefont {J.~P.}\ \bibnamefont
  {Zwolak}},\ }\href {https://doi.org/10.1007/s00220-013-1762-6} {\bibfield
  {journal} {\bibinfo  {journal} {Communications in Mathematical Physics}\
  }\textbf {\bibinfo {volume} {322}},\ \bibinfo {pages} {277} (\bibinfo {year}
  {2013})}\BibitemShut {NoStop}%
\bibitem [{\citenamefont {Caputo}\ \emph {et~al.}(2011)\citenamefont {Caputo},
  \citenamefont {Martinelli}, \citenamefont {Simenhaus},\ and\ \citenamefont
  {Toninelli}}]{caputo2011zero}%
  \BibitemOpen
  \bibfield  {author} {\bibinfo {author} {\bibfnamefont {P.}~\bibnamefont
  {Caputo}}, \bibinfo {author} {\bibfnamefont {F.}~\bibnamefont {Martinelli}},
  \bibinfo {author} {\bibfnamefont {F.}~\bibnamefont {Simenhaus}},\ and\
  \bibinfo {author} {\bibfnamefont {F.~L.}\ \bibnamefont {Toninelli}},\ }\href
  {https://doi.org/https://doi.org/10.1002/cpa.20359} {\bibfield  {journal}
  {\bibinfo  {journal} {Communications on Pure and Applied Mathematics}\
  }\textbf {\bibinfo {volume} {64}},\ \bibinfo {pages} {778} (\bibinfo {year}
  {2011})}\BibitemShut {NoStop}%
\bibitem [{\citenamefont {Feiguin}\ and\ \citenamefont
  {Klich}(2013)}]{feiguin_hermitian_2013}%
  \BibitemOpen
  \bibfield  {author} {\bibinfo {author} {\bibfnamefont {A.~E.}\ \bibnamefont
  {Feiguin}}\ and\ \bibinfo {author} {\bibfnamefont {I.}~\bibnamefont
  {Klich}},\ }\href {https://doi.org/10.48550/arXiv.1308.0756} {\bibinfo
  {title} {Hermitian and non-{Hermitian} thermal {Hamiltonians}}} (\bibinfo
  {year} {2013}),\ \bibinfo {note} {arXiv:1308.0756 [cond-mat,
  physics:physics]}\BibitemShut {NoStop}%
\bibitem [{\citenamefont {Ambainis}\ and\ \citenamefont
  {Regev}(2006)}]{ambainis2006elementary}%
  \BibitemOpen
  \bibfield  {author} {\bibinfo {author} {\bibfnamefont {A.}~\bibnamefont
  {Ambainis}}\ and\ \bibinfo {author} {\bibfnamefont {O.}~\bibnamefont
  {Regev}},\ }\href@noop {} {\bibinfo {title} {An elementary proof of the
  quantum adiabatic theorem}} (\bibinfo {year} {2006}),\ \Eprint
  {https://arxiv.org/abs/quant-ph/0411152} {arXiv:quant-ph/0411152 [quant-ph]}
  \BibitemShut {NoStop}%
\bibitem [{\citenamefont {Haah}\ \emph {et~al.}(2021)\citenamefont {Haah},
  \citenamefont {Hastings}, \citenamefont {Kothari},\ and\ \citenamefont
  {Low}}]{Haah_2021}%
  \BibitemOpen
  \bibfield  {author} {\bibinfo {author} {\bibfnamefont {J.}~\bibnamefont
  {Haah}}, \bibinfo {author} {\bibfnamefont {M.~B.}\ \bibnamefont {Hastings}},
  \bibinfo {author} {\bibfnamefont {R.}~\bibnamefont {Kothari}},\ and\ \bibinfo
  {author} {\bibfnamefont {G.~H.}\ \bibnamefont {Low}},\ }\href
  {https://doi.org/10.1137/18m1231511} {\bibfield  {journal} {\bibinfo
  {journal} {SIAM Journal on Computing}\ }\textbf {\bibinfo {volume} {52}},\
  \bibinfo {pages} {FOCS18} (\bibinfo {year} {2021})}\BibitemShut {NoStop}%
\bibitem [{\citenamefont {Nezami}\ \emph {et~al.}(2023)\citenamefont {Nezami},
  \citenamefont {Lin}, \citenamefont {Brown}, \citenamefont {Gharibyan},
  \citenamefont {Leichenauer}, \citenamefont {Salton}, \citenamefont
  {Susskind}, \citenamefont {Swingle},\ and\ \citenamefont
  {Walter}}]{nezami_quantum_2023}%
  \BibitemOpen
  \bibfield  {author} {\bibinfo {author} {\bibfnamefont {S.}~\bibnamefont
  {Nezami}}, \bibinfo {author} {\bibfnamefont {H.~W.}\ \bibnamefont {Lin}},
  \bibinfo {author} {\bibfnamefont {A.~R.}\ \bibnamefont {Brown}}, \bibinfo
  {author} {\bibfnamefont {H.}~\bibnamefont {Gharibyan}}, \bibinfo {author}
  {\bibfnamefont {S.}~\bibnamefont {Leichenauer}}, \bibinfo {author}
  {\bibfnamefont {G.}~\bibnamefont {Salton}}, \bibinfo {author} {\bibfnamefont
  {L.}~\bibnamefont {Susskind}}, \bibinfo {author} {\bibfnamefont
  {B.}~\bibnamefont {Swingle}},\ and\ \bibinfo {author} {\bibfnamefont
  {M.}~\bibnamefont {Walter}},\ }\href
  {https://doi.org/10.1103/PRXQuantum.4.010321} {\bibfield  {journal} {\bibinfo
   {journal} {PRX Quantum}\ }\textbf {\bibinfo {volume} {4}},\ \bibinfo {pages}
  {010321} (\bibinfo {year} {2023})},\ \bibinfo {note} {arXiv:2102.01064
  [hep-th, physics:quant-ph]}\BibitemShut {NoStop}%
\bibitem [{\citenamefont {Wu}\ and\ \citenamefont
  {Hsieh}(2019)}]{wu_variational_2019}%
  \BibitemOpen
  \bibfield  {author} {\bibinfo {author} {\bibfnamefont {J.}~\bibnamefont
  {Wu}}\ and\ \bibinfo {author} {\bibfnamefont {T.~H.}\ \bibnamefont {Hsieh}},\
  }\href {https://doi.org/10.1103/PhysRevLett.123.220502} {\bibfield  {journal}
  {\bibinfo  {journal} {Physical Review Letters}\ }\textbf {\bibinfo {volume}
  {123}},\ \bibinfo {pages} {220502} (\bibinfo {year} {2019})}\BibitemShut
  {NoStop}%
\bibitem [{\citenamefont {Martyn}\ and\ \citenamefont
  {Swingle}(2019)}]{martyn_product_2019}%
  \BibitemOpen
  \bibfield  {author} {\bibinfo {author} {\bibfnamefont {J.}~\bibnamefont
  {Martyn}}\ and\ \bibinfo {author} {\bibfnamefont {B.}~\bibnamefont
  {Swingle}},\ }\href {https://doi.org/10.1103/PhysRevA.100.032107} {\bibfield
  {journal} {\bibinfo  {journal} {Physical Review A}\ }\textbf {\bibinfo
  {volume} {100}},\ \bibinfo {pages} {032107} (\bibinfo {year}
  {2019})}\BibitemShut {NoStop}%
\bibitem [{\citenamefont {Chowdhury}\ \emph {et~al.}(2020)\citenamefont
  {Chowdhury}, \citenamefont {Low},\ and\ \citenamefont
  {Wiebe}}]{chowdhury_variational_2020}%
  \BibitemOpen
  \bibfield  {author} {\bibinfo {author} {\bibfnamefont {A.~N.}\ \bibnamefont
  {Chowdhury}}, \bibinfo {author} {\bibfnamefont {G.~H.}\ \bibnamefont {Low}},\
  and\ \bibinfo {author} {\bibfnamefont {N.}~\bibnamefont {Wiebe}},\ }\href
  {https://doi.org/10.48550/arXiv.2002.00055} {\bibinfo {title} {A
  {Variational} {Quantum} {Algorithm} for {Preparing} {Quantum} {Gibbs}
  {States}}} (\bibinfo {year} {2020}),\ \bibinfo {note} {arXiv:2002.00055
  [quant-ph]}\BibitemShut {NoStop}%
\bibitem [{\citenamefont {Su}(2021)}]{Su_2021}%
  \BibitemOpen
  \bibfield  {author} {\bibinfo {author} {\bibfnamefont {V.~P.}\ \bibnamefont
  {Su}},\ }\bibfield  {journal} {\bibinfo  {journal} {Physical Review A}\
  }\textbf {\bibinfo {volume} {104}},\ \href
  {https://doi.org/10.1103/physreva.104.012427} {10.1103/physreva.104.012427}
  (\bibinfo {year} {2021})\BibitemShut {NoStop}%
\bibitem [{\citenamefont {Faílde}\ \emph {et~al.}(2025)\citenamefont
  {Faílde}, \citenamefont {Santos-Suárez}, \citenamefont {Herrera-Martí},\
  and\ \citenamefont {Mas}}]{failde_hamiltonian_2023}%
  \BibitemOpen
  \bibfield  {author} {\bibinfo {author} {\bibfnamefont {D.}~\bibnamefont
  {Faílde}}, \bibinfo {author} {\bibfnamefont {J.}~\bibnamefont
  {Santos-Suárez}}, \bibinfo {author} {\bibfnamefont {D.~A.}\ \bibnamefont
  {Herrera-Martí}},\ and\ \bibinfo {author} {\bibfnamefont {J.}~\bibnamefont
  {Mas}},\ }\bibfield  {journal} {\bibinfo  {journal} {Physical Review A}\
  }\textbf {\bibinfo {volume} {111}},\ \href
  {https://doi.org/10.1103/physreva.111.012432} {10.1103/physreva.111.012432}
  (\bibinfo {year} {2025})\BibitemShut {NoStop}%
\bibitem [{\citenamefont {Zhu}\ \emph {et~al.}(2020)\citenamefont {Zhu},
  \citenamefont {Johri}, \citenamefont {Linke}, \citenamefont {Landsman},
  \citenamefont {Nguyen}, \citenamefont {Alderete}, \citenamefont {Matsuura},
  \citenamefont {Hsieh},\ and\ \citenamefont {Monroe}}]{zhu_generation_2020}%
  \BibitemOpen
  \bibfield  {author} {\bibinfo {author} {\bibfnamefont {D.}~\bibnamefont
  {Zhu}}, \bibinfo {author} {\bibfnamefont {S.}~\bibnamefont {Johri}}, \bibinfo
  {author} {\bibfnamefont {N.~M.}\ \bibnamefont {Linke}}, \bibinfo {author}
  {\bibfnamefont {K.~A.}\ \bibnamefont {Landsman}}, \bibinfo {author}
  {\bibfnamefont {N.~H.}\ \bibnamefont {Nguyen}}, \bibinfo {author}
  {\bibfnamefont {C.~H.}\ \bibnamefont {Alderete}}, \bibinfo {author}
  {\bibfnamefont {A.~Y.}\ \bibnamefont {Matsuura}}, \bibinfo {author}
  {\bibfnamefont {T.~H.}\ \bibnamefont {Hsieh}},\ and\ \bibinfo {author}
  {\bibfnamefont {C.}~\bibnamefont {Monroe}},\ }\href
  {https://doi.org/10.1073/pnas.2006337117} {\bibfield  {journal} {\bibinfo
  {journal} {Proceedings of the National Academy of Sciences}\ }\textbf
  {\bibinfo {volume} {117}},\ \bibinfo {pages} {25402} (\bibinfo {year}
  {2020})}\BibitemShut {NoStop}%
\bibitem [{\citenamefont {Aharonov}\ \emph {et~al.}(2008)\citenamefont
  {Aharonov}, \citenamefont {van Dam}, \citenamefont {Kempe}, \citenamefont
  {Landau}, \citenamefont {Lloyd},\ and\ \citenamefont
  {Regev}}]{aharonov2008adiabatic}%
  \BibitemOpen
  \bibfield  {author} {\bibinfo {author} {\bibfnamefont {D.}~\bibnamefont
  {Aharonov}}, \bibinfo {author} {\bibfnamefont {W.}~\bibnamefont {van Dam}},
  \bibinfo {author} {\bibfnamefont {J.}~\bibnamefont {Kempe}}, \bibinfo
  {author} {\bibfnamefont {Z.}~\bibnamefont {Landau}}, \bibinfo {author}
  {\bibfnamefont {S.}~\bibnamefont {Lloyd}},\ and\ \bibinfo {author}
  {\bibfnamefont {O.}~\bibnamefont {Regev}},\ }\href
  {https://doi.org/10.1137/080734479} {\bibfield  {journal} {\bibinfo
  {journal} {SIAM Review}\ }\textbf {\bibinfo {volume} {50}},\ \bibinfo {pages}
  {755} (\bibinfo {year} {2008})},\ \Eprint
  {https://arxiv.org/abs/https://doi.org/10.1137/080734479}
  {https://doi.org/10.1137/080734479} \BibitemShut {NoStop}%
\bibitem [{\citenamefont {Oliveira}\ and\ \citenamefont
  {Terhal}(2008)}]{oliveira2005complexity}%
  \BibitemOpen
  \bibfield  {author} {\bibinfo {author} {\bibfnamefont {R.}~\bibnamefont
  {Oliveira}}\ and\ \bibinfo {author} {\bibfnamefont {B.~M.}\ \bibnamefont
  {Terhal}},\ }\href@noop {} {\bibfield  {journal} {\bibinfo  {journal}
  {Quantum Info. Comput.}\ }\textbf {\bibinfo {volume} {8}},\ \bibinfo {pages}
  {900–924} (\bibinfo {year} {2008})}\BibitemShut {NoStop}%
\bibitem [{\citenamefont {Kempe}\ \emph {et~al.}(2006)\citenamefont {Kempe},
  \citenamefont {Kitaev},\ and\ \citenamefont {Regev}}]{kempe2006complexity}%
  \BibitemOpen
  \bibfield  {author} {\bibinfo {author} {\bibfnamefont {J.}~\bibnamefont
  {Kempe}}, \bibinfo {author} {\bibfnamefont {A.}~\bibnamefont {Kitaev}},\ and\
  \bibinfo {author} {\bibfnamefont {O.}~\bibnamefont {Regev}},\ }\href
  {https://doi.org/10.1137/S0097539704445226} {\bibfield  {journal} {\bibinfo
  {journal} {SIAM Journal on Computing}\ }\textbf {\bibinfo {volume} {35}},\
  \bibinfo {pages} {1070} (\bibinfo {year} {2006})},\ \Eprint
  {https://arxiv.org/abs/https://doi.org/10.1137/S0097539704445226}
  {https://doi.org/10.1137/S0097539704445226} \BibitemShut {NoStop}%
\bibitem [{\citenamefont {Ding}\ \emph {et~al.}(2025)\citenamefont {Ding},
  \citenamefont {Li},\ and\ \citenamefont
  {Lin}}]{ding2024efficientquantumgibbssamplers}%
  \BibitemOpen
  \bibfield  {author} {\bibinfo {author} {\bibfnamefont {Z.}~\bibnamefont
  {Ding}}, \bibinfo {author} {\bibfnamefont {B.}~\bibnamefont {Li}},\ and\
  \bibinfo {author} {\bibfnamefont {L.}~\bibnamefont {Lin}},\ }\bibfield
  {journal} {\bibinfo  {journal} {Communications in Mathematical Physics}\
  }\textbf {\bibinfo {volume} {406}},\ \href
  {https://doi.org/10.1007/s00220-025-05235-3} {10.1007/s00220-025-05235-3}
  (\bibinfo {year} {2025})\BibitemShut {NoStop}%
\bibitem [{\citenamefont {Scandi}\ and\ \citenamefont
  {Alhambra}(2026)}]{scandi2025thermalizationopenmanybodysystems}%
  \BibitemOpen
  \bibfield  {author} {\bibinfo {author} {\bibfnamefont {M.}~\bibnamefont
  {Scandi}}\ and\ \bibinfo {author} {\bibfnamefont {A.~M.}\ \bibnamefont
  {Alhambra}},\ }\bibfield  {journal} {\bibinfo  {journal} {Physical Review X}\
  }\textbf {\bibinfo {volume} {16}},\ \href {https://doi.org/10.1103/sfp3-3sqf}
  {10.1103/sfp3-3sqf} (\bibinfo {year} {2026})\BibitemShut {NoStop}%
\bibitem [{\citenamefont {Kashyap}\ \emph {et~al.}(2025)\citenamefont
  {Kashyap}, \citenamefont {Styliaris}, \citenamefont {Mouradian},
  \citenamefont {Cirac},\ and\ \citenamefont {Trivedi}}]{kashyap2024accuracy}%
  \BibitemOpen
  \bibfield  {author} {\bibinfo {author} {\bibfnamefont {V.}~\bibnamefont
  {Kashyap}}, \bibinfo {author} {\bibfnamefont {G.}~\bibnamefont {Styliaris}},
  \bibinfo {author} {\bibfnamefont {S.}~\bibnamefont {Mouradian}}, \bibinfo
  {author} {\bibfnamefont {J.~I.}\ \bibnamefont {Cirac}},\ and\ \bibinfo
  {author} {\bibfnamefont {R.}~\bibnamefont {Trivedi}},\ }\bibfield  {journal}
  {\bibinfo  {journal} {Physical Review X}\ }\textbf {\bibinfo {volume} {15}},\
  \href {https://doi.org/10.1103/physrevx.15.021017}
  {10.1103/physrevx.15.021017} (\bibinfo {year} {2025})\BibitemShut {NoStop}%
\bibitem [{\citenamefont {Poulin}\ and\ \citenamefont
  {Wocjan}(2009)}]{poulin_sampling_2009}%
  \BibitemOpen
  \bibfield  {author} {\bibinfo {author} {\bibfnamefont {D.}~\bibnamefont
  {Poulin}}\ and\ \bibinfo {author} {\bibfnamefont {P.}~\bibnamefont
  {Wocjan}},\ }\href {https://doi.org/10.1103/PhysRevLett.103.220502}
  {\bibfield  {journal} {\bibinfo  {journal} {Physical Review Letters}\
  }\textbf {\bibinfo {volume} {103}},\ \bibinfo {pages} {220502} (\bibinfo
  {year} {2009})},\ \bibinfo {note} {arXiv:0905.2199 [quant-ph]}\BibitemShut
  {NoStop}%
\bibitem [{\citenamefont {Chowdhury}\ and\ \citenamefont
  {Somma}(2017)}]{chowdhury2016quantum}%
  \BibitemOpen
  \bibfield  {author} {\bibinfo {author} {\bibfnamefont {A.~N.}\ \bibnamefont
  {Chowdhury}}\ and\ \bibinfo {author} {\bibfnamefont {R.~D.}\ \bibnamefont
  {Somma}},\ }\bibfield  {journal} {\bibinfo  {journal} {Quantum Information
  and Computation}\ }\textbf {\bibinfo {volume} {17}},\ \href
  {https://doi.org/10.26421/QIC17.1-2} {10.26421/QIC17.1-2} (\bibinfo {year}
  {2017})\BibitemShut {NoStop}%
\bibitem [{\citenamefont {Holmes}\ \emph {et~al.}(2022)\citenamefont {Holmes},
  \citenamefont {Muraleedharan}, \citenamefont {Somma}, \citenamefont
  {Subasi},\ and\ \citenamefont {Şahinoğlu}}]{Holmes_2022}%
  \BibitemOpen
  \bibfield  {author} {\bibinfo {author} {\bibfnamefont {Z.}~\bibnamefont
  {Holmes}}, \bibinfo {author} {\bibfnamefont {G.}~\bibnamefont
  {Muraleedharan}}, \bibinfo {author} {\bibfnamefont {R.~D.}\ \bibnamefont
  {Somma}}, \bibinfo {author} {\bibfnamefont {Y.}~\bibnamefont {Subasi}},\ and\
  \bibinfo {author} {\bibfnamefont {B.}~\bibnamefont {Şahinoğlu}},\ }\href
  {https://doi.org/10.22331/q-2022-10-06-825} {\bibfield  {journal} {\bibinfo
  {journal} {Quantum}\ }\textbf {\bibinfo {volume} {6}},\ \bibinfo {pages}
  {825} (\bibinfo {year} {2022})}\BibitemShut {NoStop}%
\bibitem [{\citenamefont {Wang}\ \emph {et~al.}(2021)\citenamefont {Wang},
  \citenamefont {Li},\ and\ \citenamefont {Wang}}]{wang_variational_2021}%
  \BibitemOpen
  \bibfield  {author} {\bibinfo {author} {\bibfnamefont {Y.}~\bibnamefont
  {Wang}}, \bibinfo {author} {\bibfnamefont {G.}~\bibnamefont {Li}},\ and\
  \bibinfo {author} {\bibfnamefont {X.}~\bibnamefont {Wang}},\ }\href
  {https://doi.org/10.1103/PhysRevApplied.16.054035} {\bibfield  {journal}
  {\bibinfo  {journal} {Physical Review Applied}\ }\textbf {\bibinfo {volume}
  {16}},\ \bibinfo {pages} {054035} (\bibinfo {year} {2021})},\ \bibinfo {note}
  {publisher: American Physical Society}\BibitemShut {NoStop}%
\bibitem [{\citenamefont {Chen}\ and\ \citenamefont
  {Brandão}(2023)}]{chen2023fast}%
  \BibitemOpen
  \bibfield  {author} {\bibinfo {author} {\bibfnamefont {C.-F.}\ \bibnamefont
  {Chen}}\ and\ \bibinfo {author} {\bibfnamefont {F.~G. S.~L.}\ \bibnamefont
  {Brandão}},\ }\href@noop {} {\bibinfo {title} {Fast thermalization from the
  eigenstate thermalization hypothesis}} (\bibinfo {year} {2023}),\ \Eprint
  {https://arxiv.org/abs/2112.07646} {arXiv:2112.07646 [quant-ph]} \BibitemShut
  {NoStop}%
\bibitem [{\citenamefont {Li}\ and\ \citenamefont {Wang}(2023)}]{LiWang2023}%
  \BibitemOpen
  \bibfield  {author} {\bibinfo {author} {\bibfnamefont {X.}~\bibnamefont
  {Li}}\ and\ \bibinfo {author} {\bibfnamefont {C.}~\bibnamefont {Wang}},\ }in\
  \href {https://doi.org/10.4230/LIPIcs.ICALP.2023.87} {\emph {\bibinfo
  {booktitle} {50th International Colloquium on Automata, Languages, and
  Programming (ICALP 2023)}}},\ \bibinfo {series} {Leibniz International
  Proceedings in Informatics (LIPIcs)}, Vol.\ \bibinfo {volume} {261},\
  \bibinfo {editor} {edited by\ \bibinfo {editor} {\bibfnamefont
  {K.}~\bibnamefont {Etessami}}, \bibinfo {editor} {\bibfnamefont
  {U.}~\bibnamefont {Feige}},\ and\ \bibinfo {editor} {\bibfnamefont
  {G.}~\bibnamefont {Puppis}}}\ (\bibinfo  {publisher} {Schloss Dagstuhl --
  Leibniz-Zentrum f{\"u}r Informatik},\ \bibinfo {address} {Dagstuhl,
  Germany},\ \bibinfo {year} {2023})\ pp.\ \bibinfo {pages}
  {87:1--87:20}\BibitemShut {NoStop}%
\bibitem [{\citenamefont {Rall}\ \emph {et~al.}(2023)\citenamefont {Rall},
  \citenamefont {Wang},\ and\ \citenamefont {Wocjan}}]{Rall_2023}%
  \BibitemOpen
  \bibfield  {author} {\bibinfo {author} {\bibfnamefont {P.}~\bibnamefont
  {Rall}}, \bibinfo {author} {\bibfnamefont {C.}~\bibnamefont {Wang}},\ and\
  \bibinfo {author} {\bibfnamefont {P.}~\bibnamefont {Wocjan}},\ }\href
  {https://doi.org/10.22331/q-2023-10-10-1132} {\bibfield  {journal} {\bibinfo
  {journal} {Quantum}\ }\textbf {\bibinfo {volume} {7}},\ \bibinfo {pages}
  {1132} (\bibinfo {year} {2023})}\BibitemShut {NoStop}%
\bibitem [{\citenamefont {Bilgin}\ and\ \citenamefont
  {Boixo}(2010)}]{bilgin_preparing_2010}%
  \BibitemOpen
  \bibfield  {author} {\bibinfo {author} {\bibfnamefont {E.}~\bibnamefont
  {Bilgin}}\ and\ \bibinfo {author} {\bibfnamefont {S.}~\bibnamefont {Boixo}},\
  }\href {https://doi.org/10.1103/PhysRevLett.105.170405} {\bibfield  {journal}
  {\bibinfo  {journal} {Physical Review Letters}\ }\textbf {\bibinfo {volume}
  {105}},\ \bibinfo {pages} {170405} (\bibinfo {year} {2010})}\BibitemShut
  {NoStop}%
\bibitem [{\citenamefont {Ge}\ \emph {et~al.}(2016)\citenamefont {Ge},
  \citenamefont {Moln\'ar},\ and\ \citenamefont {Cirac}}]{GeMolnar2016}%
  \BibitemOpen
  \bibfield  {author} {\bibinfo {author} {\bibfnamefont {Y.}~\bibnamefont
  {Ge}}, \bibinfo {author} {\bibfnamefont {A.}~\bibnamefont {Moln\'ar}},\ and\
  \bibinfo {author} {\bibfnamefont {J.~I.}\ \bibnamefont {Cirac}},\ }\href
  {https://doi.org/10.1103/PhysRevLett.116.080503} {\bibfield  {journal}
  {\bibinfo  {journal} {Phys. Rev. Lett.}\ }\textbf {\bibinfo {volume} {116}},\
  \bibinfo {pages} {080503} (\bibinfo {year} {2016})}\BibitemShut {NoStop}%
\bibitem [{\citenamefont {Helmuth}\ \emph {et~al.}(2020)\citenamefont
  {Helmuth}, \citenamefont {Perkins},\ and\ \citenamefont
  {Regts}}]{Helmuth2020}%
  \BibitemOpen
  \bibfield  {author} {\bibinfo {author} {\bibfnamefont {T.}~\bibnamefont
  {Helmuth}}, \bibinfo {author} {\bibfnamefont {W.}~\bibnamefont {Perkins}},\
  and\ \bibinfo {author} {\bibfnamefont {G.}~\bibnamefont {Regts}},\ }\href
  {https://doi.org/10.1007/s00440-019-00928-y} {\bibfield  {journal} {\bibinfo
  {journal} {Probability Theory and Related Fields}\ }\textbf {\bibinfo
  {volume} {176}},\ \bibinfo {pages} {851} (\bibinfo {year}
  {2020})}\BibitemShut {NoStop}%
\bibitem [{\citenamefont {Mann}\ and\ \citenamefont
  {Helmuth}(2021)}]{Mann2021}%
  \BibitemOpen
  \bibfield  {author} {\bibinfo {author} {\bibfnamefont {R.~L.}\ \bibnamefont
  {Mann}}\ and\ \bibinfo {author} {\bibfnamefont {T.}~\bibnamefont {Helmuth}},\
  }\bibfield  {journal} {\bibinfo  {journal} {Journal of Mathematical Physics}\
  }\textbf {\bibinfo {volume} {62}},\ \href {https://doi.org/10.1063/5.0013689}
  {10.1063/5.0013689} (\bibinfo {year} {2021})\BibitemShut {NoStop}%
\bibitem [{\citenamefont {Haah}\ \emph {et~al.}(2024)\citenamefont {Haah},
  \citenamefont {Kothari},\ and\ \citenamefont {Tang}}]{Haah_2024}%
  \BibitemOpen
  \bibfield  {author} {\bibinfo {author} {\bibfnamefont {J.}~\bibnamefont
  {Haah}}, \bibinfo {author} {\bibfnamefont {R.}~\bibnamefont {Kothari}},\ and\
  \bibinfo {author} {\bibfnamefont {E.}~\bibnamefont {Tang}},\ }\href
  {https://doi.org/10.1038/s41567-023-02376-x} {\bibfield  {journal} {\bibinfo
  {journal} {Nature Physics}\ }\textbf {\bibinfo {volume} {20}},\ \bibinfo
  {pages} {1027–1031} (\bibinfo {year} {2024})}\BibitemShut {NoStop}%
\bibitem [{\citenamefont {Mozgunov}\ and\ \citenamefont
  {Lidar}(2020)}]{Mozgunov_2020}%
  \BibitemOpen
  \bibfield  {author} {\bibinfo {author} {\bibfnamefont {E.}~\bibnamefont
  {Mozgunov}}\ and\ \bibinfo {author} {\bibfnamefont {D.}~\bibnamefont
  {Lidar}},\ }\href {https://doi.org/10.22331/q-2020-02-06-227} {\bibfield
  {journal} {\bibinfo  {journal} {Quantum}\ }\textbf {\bibinfo {volume} {4}},\
  \bibinfo {pages} {227} (\bibinfo {year} {2020})}\BibitemShut {NoStop}%
\bibitem [{\citenamefont {Nathan}\ and\ \citenamefont
  {Rudner}(2020)}]{Nathan_2020}%
  \BibitemOpen
  \bibfield  {author} {\bibinfo {author} {\bibfnamefont {F.}~\bibnamefont
  {Nathan}}\ and\ \bibinfo {author} {\bibfnamefont {M.~S.}\ \bibnamefont
  {Rudner}},\ }\bibfield  {journal} {\bibinfo  {journal} {Physical Review B}\
  }\textbf {\bibinfo {volume} {102}},\ \href
  {https://doi.org/10.1103/physrevb.102.115109} {10.1103/physrevb.102.115109}
  (\bibinfo {year} {2020})\BibitemShut {NoStop}%
\bibitem [{\citenamefont {Soret}\ \emph {et~al.}(2022)\citenamefont {Soret},
  \citenamefont {Cavina},\ and\ \citenamefont {Esposito}}]{Soret_2022}%
  \BibitemOpen
  \bibfield  {author} {\bibinfo {author} {\bibfnamefont {A.}~\bibnamefont
  {Soret}}, \bibinfo {author} {\bibfnamefont {V.}~\bibnamefont {Cavina}},\ and\
  \bibinfo {author} {\bibfnamefont {M.}~\bibnamefont {Esposito}},\ }\bibfield
  {journal} {\bibinfo  {journal} {Physical Review A}\ }\textbf {\bibinfo
  {volume} {106}},\ \href {https://doi.org/10.1103/physreva.106.062209}
  {10.1103/physreva.106.062209} (\bibinfo {year} {2022})\BibitemShut {NoStop}%
\bibitem [{\citenamefont {Coser}\ and\ \citenamefont
  {Pérez-García}(2019)}]{Coser_2019}%
  \BibitemOpen
  \bibfield  {author} {\bibinfo {author} {\bibfnamefont {A.}~\bibnamefont
  {Coser}}\ and\ \bibinfo {author} {\bibfnamefont {D.}~\bibnamefont
  {Pérez-García}},\ }\href {https://doi.org/10.22331/q-2019-08-12-174}
  {\bibfield  {journal} {\bibinfo  {journal} {Quantum}\ }\textbf {\bibinfo
  {volume} {3}},\ \bibinfo {pages} {174} (\bibinfo {year} {2019})}\BibitemShut
  {NoStop}%
\bibitem [{\citenamefont {Rakovszky}\ \emph {et~al.}(2024)\citenamefont
  {Rakovszky}, \citenamefont {Gopalakrishnan},\ and\ \citenamefont {von
  Keyserlingk}}]{rakovszky2024definingstablephasesopen}%
  \BibitemOpen
  \bibfield  {author} {\bibinfo {author} {\bibfnamefont {T.}~\bibnamefont
  {Rakovszky}}, \bibinfo {author} {\bibfnamefont {S.}~\bibnamefont
  {Gopalakrishnan}},\ and\ \bibinfo {author} {\bibfnamefont {C.}~\bibnamefont
  {von Keyserlingk}},\ }\href {https://doi.org/10.1103/PhysRevX.14.041031}
  {\bibfield  {journal} {\bibinfo  {journal} {Phys. Rev. X}\ }\textbf {\bibinfo
  {volume} {14}},\ \bibinfo {pages} {041031} (\bibinfo {year}
  {2024})}\BibitemShut {NoStop}%
\bibitem [{\citenamefont {Bakshi}\ \emph {et~al.}(2024)\citenamefont {Bakshi},
  \citenamefont {Liu}, \citenamefont {Moitra},\ and\ \citenamefont
  {Tang}}]{tang_efficient}%
  \BibitemOpen
  \bibfield  {author} {\bibinfo {author} {\bibfnamefont {A.}~\bibnamefont
  {Bakshi}}, \bibinfo {author} {\bibfnamefont {A.}~\bibnamefont {Liu}},
  \bibinfo {author} {\bibfnamefont {A.}~\bibnamefont {Moitra}},\ and\ \bibinfo
  {author} {\bibfnamefont {E.}~\bibnamefont {Tang}},\ }in\ \href
  {https://doi.org/10.1109/FOCS61266.2024.00068} {\emph {\bibinfo {booktitle}
  {2024 IEEE 65th Annual Symposium on Foundations of Computer Science
  (FOCS)}}}\ (\bibinfo {year} {2024})\ pp.\ \bibinfo {pages}
  {1027--1036}\BibitemShut {NoStop}%
\bibitem [{\citenamefont {Tran}\ \emph {et~al.}(2021)\citenamefont {Tran},
  \citenamefont {Guo}, \citenamefont {Baldwin}, \citenamefont {Ehrenberg},
  \citenamefont {Gorshkov},\ and\ \citenamefont {Lucas}}]{TranLR2021}%
  \BibitemOpen
  \bibfield  {author} {\bibinfo {author} {\bibfnamefont {M.~C.}\ \bibnamefont
  {Tran}}, \bibinfo {author} {\bibfnamefont {A.~Y.}\ \bibnamefont {Guo}},
  \bibinfo {author} {\bibfnamefont {C.~L.}\ \bibnamefont {Baldwin}}, \bibinfo
  {author} {\bibfnamefont {A.}~\bibnamefont {Ehrenberg}}, \bibinfo {author}
  {\bibfnamefont {A.~V.}\ \bibnamefont {Gorshkov}},\ and\ \bibinfo {author}
  {\bibfnamefont {A.}~\bibnamefont {Lucas}},\ }\href
  {https://doi.org/10.1103/PhysRevLett.127.160401} {\bibfield  {journal}
  {\bibinfo  {journal} {Phys. Rev. Lett.}\ }\textbf {\bibinfo {volume} {127}},\
  \bibinfo {pages} {160401} (\bibinfo {year} {2021})}\BibitemShut {NoStop}%
\bibitem [{\citenamefont {Verstraete}\ \emph {et~al.}(2009)\citenamefont
  {Verstraete}, \citenamefont {Wolf},\ and\ \citenamefont
  {Ignacio~Cirac}}]{Verstraete2009}%
  \BibitemOpen
  \bibfield  {author} {\bibinfo {author} {\bibfnamefont {F.}~\bibnamefont
  {Verstraete}}, \bibinfo {author} {\bibfnamefont {M.~M.}\ \bibnamefont
  {Wolf}},\ and\ \bibinfo {author} {\bibfnamefont {J.}~\bibnamefont
  {Ignacio~Cirac}},\ }\href {https://doi.org/10.1038/nphys1342} {\bibfield
  {journal} {\bibinfo  {journal} {Nature Physics}\ }\textbf {\bibinfo {volume}
  {5}},\ \bibinfo {pages} {633–636} (\bibinfo {year} {2009})}\BibitemShut
  {NoStop}%
\bibitem [{\citenamefont {Kraus}\ \emph {et~al.}(2008)\citenamefont {Kraus},
  \citenamefont {B\"uchler}, \citenamefont {Diehl}, \citenamefont {Kantian},
  \citenamefont {Micheli},\ and\ \citenamefont {Zoller}}]{Krausetal}%
  \BibitemOpen
  \bibfield  {author} {\bibinfo {author} {\bibfnamefont {B.}~\bibnamefont
  {Kraus}}, \bibinfo {author} {\bibfnamefont {H.~P.}\ \bibnamefont
  {B\"uchler}}, \bibinfo {author} {\bibfnamefont {S.}~\bibnamefont {Diehl}},
  \bibinfo {author} {\bibfnamefont {A.}~\bibnamefont {Kantian}}, \bibinfo
  {author} {\bibfnamefont {A.}~\bibnamefont {Micheli}},\ and\ \bibinfo {author}
  {\bibfnamefont {P.}~\bibnamefont {Zoller}},\ }\href
  {https://doi.org/10.1103/PhysRevA.78.042307} {\bibfield  {journal} {\bibinfo
  {journal} {Phys. Rev. A}\ }\textbf {\bibinfo {volume} {78}},\ \bibinfo
  {pages} {042307} (\bibinfo {year} {2008})}\BibitemShut {NoStop}%
\bibitem [{\citenamefont {Cubitt}\ \emph {et~al.}(2015)\citenamefont {Cubitt},
  \citenamefont {Lucia}, \citenamefont {Michalakis},\ and\ \citenamefont
  {Perez-Garcia}}]{Cubitt2015}%
  \BibitemOpen
  \bibfield  {author} {\bibinfo {author} {\bibfnamefont {T.~S.}\ \bibnamefont
  {Cubitt}}, \bibinfo {author} {\bibfnamefont {A.}~\bibnamefont {Lucia}},
  \bibinfo {author} {\bibfnamefont {S.}~\bibnamefont {Michalakis}},\ and\
  \bibinfo {author} {\bibfnamefont {D.}~\bibnamefont {Perez-Garcia}},\ }\href
  {https://doi.org/10.1007/s00220-015-2355-3} {\bibfield  {journal} {\bibinfo
  {journal} {Communications in Mathematical Physics}\ }\textbf {\bibinfo
  {volume} {337}},\ \bibinfo {pages} {1275–1315} (\bibinfo {year}
  {2015})}\BibitemShut {NoStop}%
\bibitem [{\citenamefont {Borregaard}\ \emph {et~al.}(2021)\citenamefont
  {Borregaard}, \citenamefont {Christandl},\ and\ \citenamefont
  {Stilck~França}}]{Borregaard2021}%
  \BibitemOpen
  \bibfield  {author} {\bibinfo {author} {\bibfnamefont {J.}~\bibnamefont
  {Borregaard}}, \bibinfo {author} {\bibfnamefont {M.}~\bibnamefont
  {Christandl}},\ and\ \bibinfo {author} {\bibfnamefont {D.}~\bibnamefont
  {Stilck~França}},\ }\bibfield  {journal} {\bibinfo  {journal} {npj Quantum
  Information}\ }\textbf {\bibinfo {volume} {7}},\ \href
  {https://doi.org/10.1038/s41534-021-00363-9} {10.1038/s41534-021-00363-9}
  (\bibinfo {year} {2021})\BibitemShut {NoStop}%
\bibitem [{\citenamefont {Kim}\ and\ \citenamefont
  {Swingle}(2017)}]{kim2017robust}%
  \BibitemOpen
  \bibfield  {author} {\bibinfo {author} {\bibfnamefont {I.~H.}\ \bibnamefont
  {Kim}}\ and\ \bibinfo {author} {\bibfnamefont {B.}~\bibnamefont {Swingle}},\
  }\href@noop {} {\bibinfo {title} {Robust entanglement renormalization on a
  noisy quantum computer}} (\bibinfo {year} {2017}),\ \Eprint
  {https://arxiv.org/abs/1711.07500} {arXiv:1711.07500 [quant-ph]} \BibitemShut
  {NoStop}%
\bibitem [{\citenamefont {Childs}\ \emph {et~al.}(2001)\citenamefont {Childs},
  \citenamefont {Farhi},\ and\ \citenamefont {Preskill}}]{ChildsAdiabatic2001}%
  \BibitemOpen
  \bibfield  {author} {\bibinfo {author} {\bibfnamefont {A.~M.}\ \bibnamefont
  {Childs}}, \bibinfo {author} {\bibfnamefont {E.}~\bibnamefont {Farhi}},\ and\
  \bibinfo {author} {\bibfnamefont {J.}~\bibnamefont {Preskill}},\ }\href
  {https://doi.org/10.1103/PhysRevA.65.012322} {\bibfield  {journal} {\bibinfo
  {journal} {Phys. Rev. A}\ }\textbf {\bibinfo {volume} {65}},\ \bibinfo
  {pages} {012322} (\bibinfo {year} {2001})}\BibitemShut {NoStop}%
\bibitem [{\citenamefont {Young}\ \emph {et~al.}(2013)\citenamefont {Young},
  \citenamefont {Sarovar},\ and\ \citenamefont
  {Blume-Kohout}}]{AdiabaticQEC2013}%
  \BibitemOpen
  \bibfield  {author} {\bibinfo {author} {\bibfnamefont {K.~C.}\ \bibnamefont
  {Young}}, \bibinfo {author} {\bibfnamefont {M.}~\bibnamefont {Sarovar}},\
  and\ \bibinfo {author} {\bibfnamefont {R.}~\bibnamefont {Blume-Kohout}},\
  }\href {https://doi.org/10.1103/PhysRevX.3.041013} {\bibfield  {journal}
  {\bibinfo  {journal} {Phys. Rev. X}\ }\textbf {\bibinfo {volume} {3}},\
  \bibinfo {pages} {041013} (\bibinfo {year} {2013})}\BibitemShut {NoStop}%
\bibitem [{\citenamefont {Bravyi}\ \emph {et~al.}(2010)\citenamefont {Bravyi},
  \citenamefont {Hastings},\ and\ \citenamefont {Michalakis}}]{Bravyi_2010}%
  \BibitemOpen
  \bibfield  {author} {\bibinfo {author} {\bibfnamefont {S.}~\bibnamefont
  {Bravyi}}, \bibinfo {author} {\bibfnamefont {M.~B.}\ \bibnamefont
  {Hastings}},\ and\ \bibinfo {author} {\bibfnamefont {S.}~\bibnamefont
  {Michalakis}},\ }\bibfield  {journal} {\bibinfo  {journal} {Journal of
  Mathematical Physics}\ }\textbf {\bibinfo {volume} {51}},\ \href
  {https://doi.org/10.1063/1.3490195} {10.1063/1.3490195} (\bibinfo {year}
  {2010})\BibitemShut {NoStop}%
\bibitem [{\citenamefont {Hastings}\ and\ \citenamefont
  {Koma}(2006)}]{HastingsKoma2006}%
  \BibitemOpen
  \bibfield  {author} {\bibinfo {author} {\bibfnamefont {M.~B.}\ \bibnamefont
  {Hastings}}\ and\ \bibinfo {author} {\bibfnamefont {T.}~\bibnamefont
  {Koma}},\ }\href {https://doi.org/10.1007/s00220-006-0030-4} {\bibfield
  {journal} {\bibinfo  {journal} {Communications in Mathematical Physics}\
  }\textbf {\bibinfo {volume} {265}},\ \bibinfo {pages} {781} (\bibinfo {year}
  {2006})}\BibitemShut {NoStop}%
\bibitem [{\citenamefont {(Anthony)~Chen}\ \emph {et~al.}(2023)\citenamefont
  {(Anthony)~Chen}, \citenamefont {Lucas},\ and\ \citenamefont
  {Yin}}]{Review_Anthony_Chen_2023}%
  \BibitemOpen
  \bibfield  {author} {\bibinfo {author} {\bibfnamefont {C.-F.}\ \bibnamefont
  {(Anthony)~Chen}}, \bibinfo {author} {\bibfnamefont {A.}~\bibnamefont
  {Lucas}},\ and\ \bibinfo {author} {\bibfnamefont {C.}~\bibnamefont {Yin}},\
  }\href {https://doi.org/10.1088/1361-6633/acfaae} {\bibfield  {journal}
  {\bibinfo  {journal} {Reports on Progress in Physics}\ }\textbf {\bibinfo
  {volume} {86}},\ \bibinfo {pages} {116001} (\bibinfo {year}
  {2023})}\BibitemShut {NoStop}%
\bibitem [{\citenamefont {Davies}(1974)}]{Davies1974}%
  \BibitemOpen
  \bibfield  {author} {\bibinfo {author} {\bibfnamefont {E.~B.}\ \bibnamefont
  {Davies}},\ }\href {https://doi.org/10.1007/BF01608389} {\bibfield  {journal}
  {\bibinfo  {journal} {Communications in Mathematical Physics}\ }\textbf
  {\bibinfo {volume} {39}},\ \bibinfo {pages} {91} (\bibinfo {year}
  {1974})}\BibitemShut {NoStop}%
\bibitem [{\citenamefont {Rouz\'{e}}\ \emph {et~al.}(2025)\citenamefont
  {Rouz\'{e}}, \citenamefont {Fran\c{c}a},\ and\ \citenamefont
  {Alhambra}}]{STOC}%
  \BibitemOpen
  \bibfield  {author} {\bibinfo {author} {\bibfnamefont {C.}~\bibnamefont
  {Rouz\'{e}}}, \bibinfo {author} {\bibfnamefont {D.~S.}\ \bibnamefont
  {Fran\c{c}a}},\ and\ \bibinfo {author} {\bibfnamefont {A.~M.}\ \bibnamefont
  {Alhambra}},\ }in\ \href {https://doi.org/10.1145/3717823.3718268} {\emph
  {\bibinfo {booktitle} {Proceedings of the 57th Annual ACM Symposium on Theory
  of Computing}}},\ \bibinfo {series and number} {STOC '25}\ (\bibinfo
  {publisher} {Association for Computing Machinery},\ \bibinfo {address} {New
  York, NY, USA},\ \bibinfo {year} {2025})\ p.\ \bibinfo {pages}
  {1488–1495}\BibitemShut {NoStop}%
\bibitem [{\citenamefont {Lieb}\ and\ \citenamefont
  {Robinson}(1972)}]{lieb1972finite}%
  \BibitemOpen
  \bibfield  {author} {\bibinfo {author} {\bibfnamefont {E.~H.}\ \bibnamefont
  {Lieb}}\ and\ \bibinfo {author} {\bibfnamefont {D.~W.}\ \bibnamefont
  {Robinson}},\ }\href {https://doi.org/10.1007/BF01645779} {\bibfield
  {journal} {\bibinfo  {journal} {Communications in Mathematical Physics}\
  }\textbf {\bibinfo {volume} {28}},\ \bibinfo {pages} {251} (\bibinfo {year}
  {1972})}\BibitemShut {NoStop}%
\bibitem [{\citenamefont {Hastings}(2004)}]{hastings2004lieb}%
  \BibitemOpen
  \bibfield  {author} {\bibinfo {author} {\bibfnamefont {M.~B.}\ \bibnamefont
  {Hastings}},\ }\href {https://doi.org/10.1103/PhysRevB.69.104431} {\bibfield
  {journal} {\bibinfo  {journal} {Phys. Rev. B}\ }\textbf {\bibinfo {volume}
  {69}},\ \bibinfo {pages} {104431} (\bibinfo {year} {2004})}\BibitemShut
  {NoStop}%
\bibitem [{\citenamefont {Chen}\ \emph {et~al.}(2023)\citenamefont {Chen},
  \citenamefont {Kastoryano},\ and\ \citenamefont
  {Gilyén}}]{chen2023efficient}%
  \BibitemOpen
  \bibfield  {author} {\bibinfo {author} {\bibfnamefont {C.-F.}\ \bibnamefont
  {Chen}}, \bibinfo {author} {\bibfnamefont {M.~J.}\ \bibnamefont
  {Kastoryano}},\ and\ \bibinfo {author} {\bibfnamefont {A.}~\bibnamefont
  {Gilyén}},\ }\href@noop {} {\bibinfo {title} {An efficient and exact
  noncommutative quantum gibbs sampler}} (\bibinfo {year} {2023}),\ \Eprint
  {https://arxiv.org/abs/2311.09207} {arXiv:2311.09207 [quant-ph]} \BibitemShut
  {NoStop}%
\bibitem [{\citenamefont {Olkiewicz}\ and\ \citenamefont
  {Zegarlinski}(1999)}]{PoincareRef}%
  \BibitemOpen
  \bibfield  {author} {\bibinfo {author} {\bibfnamefont {R.}~\bibnamefont
  {Olkiewicz}}\ and\ \bibinfo {author} {\bibfnamefont {B.}~\bibnamefont
  {Zegarlinski}},\ }\href
  {https://doi.org/https://doi.org/10.1006/jfan.1998.3342} {\bibfield
  {journal} {\bibinfo  {journal} {Journal of Functional Analysis}\ }\textbf
  {\bibinfo {volume} {161}},\ \bibinfo {pages} {246} (\bibinfo {year}
  {1999})}\BibitemShut {NoStop}%
\bibitem [{\citenamefont {Kastoryano}\ and\ \citenamefont
  {Temme}(2013)}]{Kastoryano_2013}%
  \BibitemOpen
  \bibfield  {author} {\bibinfo {author} {\bibfnamefont {M.~J.}\ \bibnamefont
  {Kastoryano}}\ and\ \bibinfo {author} {\bibfnamefont {K.}~\bibnamefont
  {Temme}},\ }\bibfield  {journal} {\bibinfo  {journal} {Journal of
  Mathematical Physics}\ }\textbf {\bibinfo {volume} {54}},\ \href
  {https://doi.org/10.1063/1.4804995} {10.1063/1.4804995} (\bibinfo {year}
  {2013})\BibitemShut {NoStop}%
\bibitem [{\citenamefont {Bhatia}(1997)}]{Bhatia1997}%
  \BibitemOpen
  \bibfield  {author} {\bibinfo {author} {\bibfnamefont {R.}~\bibnamefont
  {Bhatia}},\ }\href {https://doi.org/10.1007/978-1-4612-0653-8} {\emph
  {\bibinfo {title} {Matrix Analysis}}}\ (\bibinfo  {publisher} {Springer New
  York},\ \bibinfo {year} {1997})\BibitemShut {NoStop}%
\bibitem [{\citenamefont {Kitaev}\ \emph {et~al.}(2002)\citenamefont {Kitaev},
  \citenamefont {Shen},\ and\ \citenamefont {Vyalyi}}]{kitaev2002classical}%
  \BibitemOpen
  \bibfield  {author} {\bibinfo {author} {\bibfnamefont {A.~Y.}\ \bibnamefont
  {Kitaev}}, \bibinfo {author} {\bibfnamefont {A.}~\bibnamefont {Shen}},\ and\
  \bibinfo {author} {\bibfnamefont {M.~N.}\ \bibnamefont {Vyalyi}},\
  }\href@noop {} {\emph {\bibinfo {title} {Classical and quantum
  computation}}},\ \bibinfo {number} {47}\ (\bibinfo  {publisher} {American
  Mathematical Soc.},\ \bibinfo {year} {2002})\BibitemShut {NoStop}%
\bibitem [{\citenamefont {Parlangeli}\ and\ \citenamefont
  {Notarstefano}(2011)}]{parlangeli2011reachability}%
  \BibitemOpen
  \bibfield  {author} {\bibinfo {author} {\bibfnamefont {G.}~\bibnamefont
  {Parlangeli}}\ and\ \bibinfo {author} {\bibfnamefont {G.}~\bibnamefont
  {Notarstefano}},\ }\href@noop {} {\bibfield  {journal} {\bibinfo  {journal}
  {IEEE Transactions on Automatic Control}\ }\textbf {\bibinfo {volume} {57}},\
  \bibinfo {pages} {743} (\bibinfo {year} {2011})}\BibitemShut {NoStop}%
\end{thebibliography}%


\appendix

\section{Basic notation and setup}\label{sec:setup}

Given a finite set $V$, we denote by $\cH_V=\bigotimes_{v\in V}\cH_v$ the Hilbert space of $n=|V|$ qubits (\emph{i.e.}, $\cH_v\equiv \mathbb{C}^2$ for all $v\in V$) and by $\cB_V$ the algebra of linear operators on $\cH_V$. The trace on $\cB_V$ is denoted by $\tr{\cdot}$. $\mathcal{O}_V$ corresponds to the space of self-adjoint linear operators on $\cH_V$. $\mathcal{S}_V$ denotes the set of quantum states. For any subset $Z\subseteq V$, we use the standard notations $\mathcal{O}_Z, \mathcal{S}_Z\ldots$ for the corresponding objects defined on subsystem $Z$. For systems made of $n$ qudits, we sometimes also use the notations $[n]$ when refering to the set $\{1,\dots ,n\}$. Given a state $\rho\in\mathcal{S}_V$, we denote by $\rho_Z$ its marginal in the subsystem $Z$. For any $X\in\mathcal{O}_V$, we denote by $\|X\|_1:=\tr{|X|}$ its trace norm, by $\|X\|_\infty:=\sup_{\||\psi\rangle \|\le 1}\|X|\psi\rangle\|$ its operator norm, and by $\|X\|_2:=\sqrt{\tr{|X|^2}}$ its Hilbert-Schmidt norm. Following standard quantum information-theoretic terminology, operators acting on susystems $Z\subseteq V$ are often denoted by $O_Z$. The identity map is denoted by $I$. We denote by $O^\dagger$ the adjoint of an operator $O\in \cB_V$, and by $\Phi^\dagger$ the Hilbert-Schmidt dual of a superoperator on $\cB_V$. 

\subsection{Local Hamiltonians and Lieb-Robinson bounds}

Next, we consider a quantum spin system on a $D$-dimensional finite lattice $\Lambda\equiv [0,L]^D$, with $n:=|\Lambda|=(L+1)^D$, and a local Hamiltonian $H$ on the lattice, defined through a function that maps any non-empty finite set $X\subset \Lambda$ to a self-adjoint element $h_X$ supported in $X$ with $\max_{X} \norm{h_X}_\infty \le h$. The Hamiltonian of the system in a region $Z \subset \Lambda$ is then defined as the sum over interactions included in $Z$:
\begin{align*}
		H_Z:=\sum_{X\subseteq Z}h_X\,.
\end{align*}
In addition, we assume that the Hamiltonian is such that no interaction $h_X$ has support on more than $k$ sites, and each site $u$ appears on at most $l$ non-zero operators $h_X$. We refer to such Hamiltonian as $(k,l)$-local. As often, we will make the implicit assumption that $H$ is one element of a family of Hamiltonians indexed by the lattice size, in which case the constants $h,k$ and $l$ are assumed to be independent of $n$.

\medskip

For such Hamiltonians, given the so-called Lieb-Robinson velocity
$$J := \max_{u\in \Lambda} \sum_{ X\ni u} \vert X \vert \norm{h_X}_\infty=\mathcal{O}(hkl)\,,$$
 it is well-known that, for any operator $A^u$ supported on site $u \in \Lambda$ \cite{lieb1972finite,hastings2004lieb,nachtergaele2006lieb,HastingsKoma2006,Haah_2021}:
\begin{equation}\label{eq:LiebRobinson}
    \norm{e^{-iHt}A^u e^{iHt}-e^{-iH_{B_u(r)}t}A^u e^{iH_{B_u(r)}t}}_\infty \le \norm{A^u}_\infty \frac{(2J\vert t \vert)^r}{r!}\le \norm{A^u}_\infty \frac{(2h k l\vert t \vert)^r}{r!}\,,
\end{equation}
where $B_u(r)$ denotes the ball of radius $r$ with respect to the graph distance centered at $u$, and where the second bound results from $(k,l)$-locality.

\subsection{Quantum Gibbs sampling with Gaussian filters}\label{appendixGibbssampling}

We now recall the generator of the quantum Gibbs sampler introduced in \cite{chen2023efficient}: for $\beta>0$ and a Hamiltonian $H$, we denote $\gamma(\omega):=\operatorname{exp}(-\frac{(\beta\omega+1)^2}{2})$ and we consider the Lindbladian defined for all $\rho\in \mathcal{S}_\Lambda$ as
\begin{align}\label{eq:generator}
\mathcal{L}^{(\beta)}(\rho)&=-i[B,\rho]+{\sum_{a\in \Lambda}\sum_{\alpha\in[3]}\int_{-\infty}^{\infty} \gamma(\omega)\,\Big(A^{a,\alpha}(\omega)\rho A^{a,\alpha}(\omega)^\dagger-\frac{1}{2}\big\{A^{a,\alpha}(\omega)^\dagger A^{a,\alpha}(\omega),\rho\big\}\Big)d\omega}\\
&\equiv \sum_a\sum_\alpha\mathcal{L}^{(\beta)}_{a,\alpha}(\rho)\,,
\end{align}
with jump operators
\begin{align*}
A^{a,\alpha}(\omega):=\frac{1}{\sqrt{2\pi}}\int_{-\infty}^{\infty} e^{iHt}A^{a,\alpha} e^{-iH t}e^{-i\omega t}\,f(t)\,dt\qquad \text{ where }\quad f(t):=e^{-\frac{t^2}{\beta^2}}\,\sqrt{\beta^{-1}\sqrt{2/\pi}}\,.
\end{align*}
Above, $A^{a,1}\equiv \sigma_x,A^{a,2}\equiv \sigma_y$ and $A^{a,3}\equiv \sigma_z$ denote the $1$-local Pauli matrices on site $a$. The coherent part is defined through the self-adjoint matrix
\begin{align}\label{eq:coherentterm}
B:=\sum_{a\in \Lambda}\sum_{\alpha\in [3]}\, \int_{-\infty}^\infty b_1(t) e^{-i\beta Ht} \int_{-\infty}^\infty b_2(t') e^{i\beta Ht'}A^{a,\alpha} e^{-2i\beta Ht'}A^{a,\alpha} e^{i\beta Ht'}dt' e^{i\beta Ht}dt\equiv \sum_a\sum_\alpha B^\beta_{a,\alpha}\,,
\end{align}
for some smooth, rapidly decaying functions $b_1,b_2$ with $\|b_1\|_{L_1}, \|b_2\|_{L_1}\le 1$ defined as
\begin{align}\label{eq:b1b2}
&b_1(t):= 2\sqrt{\pi}\, e^{\frac{1}{8}}\, \left(\frac{1}{\operatorname{cosh}(2\pi t)}\ast_t\,\sin(-t)e^{-2t^2}\right)\\
&b_2(t):=\frac{1}{2\pi}\,\sqrt{\frac{1}{\pi}}\,\operatorname{exp}\big(-4t^2-2it\big)\,,
\end{align}
where the symbol $\ast_t$ indicates a convolution.
The Lindbladian constructed above has the main advantage of being both quasi-local and reversible (i.e. in ``detailed balance'') in the Heisenberg picture:
\begin{align*}
\langle X,\mathcal{L}^{(\beta)\dagger}(Y)\rangle_{\sigma_\beta}=\langle \mathcal{L}^{(\beta)\dagger}(X),Y\rangle_{\sigma_\beta}\,,
\end{align*}
where $\langle X,Y\rangle_{\sigma_\beta}:=\tr{X^\dagger{\sigma_\beta}^{\frac{1}{2}}Y{\sigma_\beta}^{\frac{1}{2}}}$ denotes the KMS inner product, and where the Gibbs state
\begin{align*}
\sigma_\beta=\frac{e^{-\beta H}}{\tr{e^{-\beta H}}}
\end{align*}
is the unique invariant state of the evolution generated by $\mathcal{L}^{(\beta)}$.

\medskip

In what follows, we will make use of a well-known mapping of KMS-symmetric Lindbladians into Hamiltonians: for any $\beta>0$, we define the super-operator
\begin{align}\label{eq:Lbetatilde}
\widetilde{\mathcal{L}}^{(\beta)}(X):=\sigma_\beta^{-\frac{1}{4}}\mathcal{L}^{(\beta)}(\sigma_\beta^{\frac{1}{4}}X\sigma_\beta^{\frac{1}{4}} )\sigma_\beta^{-\frac{1}{4}}\equiv\sum_{a,\alpha}\widetilde{\mathcal{L}}_{a,\alpha}^{(\beta)}\,,
\end{align}
where $X$ is an arbitrary operator acting as input to $\widetilde{\mathcal{L}}^{(\beta)}$.
Clearly, since $\sigma_\beta$ is a steady state of the evolution, and unique due to the choice of the jumps $A^{a,\alpha}$, $\operatorname{Ker}(\widetilde{\mathcal{L}}^{(\beta)})=\sqrt{\sigma_\beta}\equiv \vert \sqrt{\sigma_\beta} \rangle$. Moreover, $\widetilde{\mathcal{L}}^{(\beta)}$ is self-adjoint with respect to the Hilbert-Schmidt inner product for all $\beta>0$, by the KMS symmetry of $\mathcal{L}^{(\beta)\dagger}$. Our strategy in what follows will be to show that $\widetilde{\mathcal{L}}^{(\beta)}$ can be expressed as a quasi-local perturbation of a certain $1$-local gapped Hamiltonian (the limit $\beta \rightarrow 0$) and to study the stability of the latter against such perturbations.

\section{Constant gap at high temperature}\label{sec:highTgibbsgap}
In this appendix, we show our first main result, namely that the Lindbladians $\mathcal{L}^{(\beta)\dagger}$ are gapped at high enough temperature:

\begin{theorem}\label{thmgaphighT}
In the notations of Appendix \ref{sec:setup}, for any $(k,l)$-local Hamiltonian $H$, there exists a constant $\beta^*>0$ independent of $n$ such that for $\beta < \beta^*$, the spectral gap of $\mathcal{L}^{(\beta)\dagger}$, is lower bounded by $\frac{1}{2\sqrt{2}e^{1/4}}$.
\end{theorem}

\noindent In our setting, the existence of a constant gap leads to the following standard consequence on the fast convergence of the sampler towards its unique fixed point:

\begin{corollary}\label{coro1}
In the setting of Theorem \ref{thmgaphighT}, for any $\beta<\beta^*$ and $\epsilon>0$, the evolution generated by $\mathcal{L}^{(\beta)}$ gets $\epsilon$-close to $\sigma_\beta$ in polynomial time. More precisely, for any $\rho\in \mathcal{S}_{\Lambda}$,
\begin{align*}
\|e^{t\mathcal{L}^{(\beta)}}(\rho)-\sigma_\beta\|_1\le \epsilon\qquad \text{ for all}\qquad t=\Omega(\ln(1/\epsilon)+n)\,. 
\end{align*}
\end{corollary}

\begin{proof}
The corollary above is a consequence of a standard result, whose proof we recall for the reader's convenience. By equivalence of s.a. in the noncommutative weighted $L_2$ norm and the KMS symmetry of $\mathcal{L}^{(\beta)\dagger}$, the trace distance can be controlled as follows: given $X:=\sigma_\beta^{-{1}/{2}}\rho \sigma_\beta^{-{1}/{2}}$,
\begin{align*}
\|e^{t\mathcal{L}^{(\beta)}}(\rho)-\sigma_\beta\|_1\le \|e^{t\mathcal{L}^{(\beta)\dagger}}(X)-I\|_{\sigma_\beta}&\le e^{-\frac{t}{2\sqrt{2}e^{1/4}}}\|X-I\|_{\sigma_\beta}\\
&\le 2\|\sigma_{\beta}^{-1}\|_\infty e^{-\frac{t}{2\sqrt{2}e^{1/4}}}\\
&\le 2^{|\Lambda|+1}\,e^{2\beta^*\|H\|_\infty-\frac{t}{2\sqrt{2}e^{1/4}}}\\
&\le 2^{|\Lambda|+1}\,e^{2\beta^*lhn-\frac{t}{2\sqrt{2}e^{1/4}}}\,,
\end{align*}
where $\norm{Y}_{\sigma_\beta}=\norm{\sigma_{\beta}^{1/4}Y\sigma_{\beta}^{1/4}}_2$.
The result follows by setting the upper bound above to $\epsilon$ and solving for $t$. 

\end{proof}

The rest of this appendix is devoted to proving Theorem \ref{thmgaphighT} and to applying it to the adiabatic algorithm. In Section \ref{stabilityFFgapped}, we recall some well-known gap stability results for the class of frustration-free gapped Hamiltonians satisfying a condition known as local topological quantum order. In Section \ref{sec:gapped}, we prove that the maps $\widetilde{\mathcal{L}}^{(\beta)}$ satisfy the conditions for stability of the gap recalled in Section \ref{stabilityFFgapped}. Finally, we analyze the efficiency of the adiabatic evolution associated to the generator $\widetilde{\mathcal{L}}^{(\beta)}$ in Section \ref{app:adiabatic}.

\subsection{Stability of frustration-free gapped Hamiltonians}\label{stabilityFFgapped}

As previously mentioned, our main result at high temperatures builds on a well-known result \cite{Bravyi_2010,michalakis2013stability} about the stability of certain gapped Hamitonians under small quasi-local perturbations: let $\widetilde{H}_0$ be a Hamiltonian defined on a $D$-dimensional lattice $\Lambda$ of side-length $L+1$ with corresponding Hilbert space $\widetilde{\mathcal{H}}_\Lambda=\bigotimes_{j\in\Lambda}\widetilde{\mathcal{H}}_j$, which satisfies the following requirements:
\begin{itemize}
\item[1.] It can be written as a sum of geometrically local terms $\widetilde{H}_0=\sum_{u\in\Lambda} \widetilde{Q}_u$, where each interaction $\widetilde{Q}_u$ is positive semidefinite and acts non-trivially on a Hilbert space supported on the ball $B_u(1)$ of radius $1$ around site $u$; 
\item[2.] It satisfies periodic boundary conditions;
\item[3.] For $\widetilde{P}_0$ the projector onto the groundspace of $\widetilde{H}_0$, we have $\widetilde{H}_0\widetilde{P}_0=0$, i.e.~the ground state energy is fixed to $0$. Moreover, $\widetilde{Q}_u\widetilde{P}_0=0$ for all $u \in \Lambda$ (frustration-freeness);
\item[4.] $\widetilde{H}_0$ has a gap $\gamma_L\ge \gamma>0$ above the groundstate subspace, for all $L\ge 2$, where $\gamma$ is a constant independent of the system size $L$.
\end{itemize}

Let us set $L^*$ acting as a cut-off parameter for topological order, which corresponds to e.g. the code distance in a stabilizer Hamiltonian.
Given a ball $A\equiv B_u(r)$ with $r\le L^*\le L$ and an observable $O_{Z}$ supported on $Z$, we set $Z(\ell):=B_u(r+\ell)$ and denote 
\begin{align*}
c_\ell(O_Z):=\frac{\tr{\widetilde{P}_{Z(\ell)}O_Z}}{\tr{\widetilde{P}_{Z(\ell)}}}\,,\qquad c(O_Z):=\frac{\tr{\widetilde{P}_0 O_Z}}{\tr{\widetilde{P}_0}}\,,
\end{align*}
where $\widetilde{P}_B(\epsilon)$ is the projection onto the subspace of eigenstates of $\widetilde{H}_B=\sum_{B_v(1)\subset B}\widetilde{Q}_v$ with energy at most $\epsilon\ge 0$ and we set $\widetilde{P}_B(0)\equiv \widetilde{P}_B$. The Hamiltonian $\widetilde{H}_0$ satisfies local topological order if, for each fixed $1\le \ell\le L-r$,
\begin{align*}
\Big\|P_{Z(\ell)}O_Z P_{Z(\ell)}-c_\ell(O_Z)P_{Z(\ell)}\Big\|\le \|O_Z\|\,\Delta_0(\ell)\,,
\end{align*}
where $\Delta_0(\ell)$ is a decaying function of $\ell$. Finally, we say that $\widetilde{H}_0$ is locally gapped with gap $\gamma(r)>0$ iff for each $u\in\Lambda$ and $r\ge 0$, $\widetilde{P}_{B_u(r)}(\gamma(r))=\widetilde{P}_{B_u(r)}$. When $\gamma(r)$ decays at most polynomially in $r$, we say that $\widetilde{H}_0$ satisfies the local gap condition. 

\medskip

Next, we consider a perturbation $\widetilde{V}$ with strength $\widetilde{J}$ and decay rate $f(r)$, with $f(r)\le 1$, $r\ge 0$, that is, $\widetilde{V}:=\sum_{u\in\Lambda}\sum_{r=0}^L\widetilde{V}_u(r)$, such that $\widetilde{V}_u(r)$ has support on $B_u(r)$ and satisfies $\|\widetilde{V}_u(r)\|\le \widetilde{J}f(r)$ for some rapidly decaying function $f(r)$ with $f(r)\le 1$. 

\begin{theorem}[Stability of gap for frustration-free gapped Hamiltonians \cite{michalakis2013stability}]\label{thm:stability} Let $\widetilde{H}_0$ be a frustration-free Hamiltonian as defined above with spectral gap $\gamma$ satisfying the local topological quantum order and local-gap conditions. For a perturbation $\widetilde{V}$ of strength $\widetilde{J}$ and decay rate $f(r)$ there exist constants $\widetilde{J}_0>0$ and $L_0\ge 2$ such that for $\widetilde{J}\le \widetilde{J}_0$ and $L\ge L_0$, the spectral gap of $\widetilde{H}_0+\widetilde{V}$ is bounded from below by $\gamma/2$.
\end{theorem}

The constants $\widetilde{J}_0$ and $\widetilde{L}_0$ depend on the unperturbed Hamiltonian $\widetilde{H}_0$ and the form of the decay $f(r)$. More specifically, they depend on the Lieb-Robinson velocity of $\widetilde{H}_0+\widetilde{V}$, the global $\gamma$ and local $\gamma(r)$ gaps, and the volume of a $D$-dimensional sphere if radius $1$ (so that $\widetilde{J}_0$ decreases for higher dimensional lattices).  Their specific form can be determined from \cite{michalakis2013stability}.

\subsection{Quantum Gibbs sampler as a gapped Hamiltonian: proof of Theorem \ref{thmgaphighT}}\label{sec:gapped}

We now show that the Hamiltonian $\widetilde{\mathcal{L}}^{(\beta)}$ defined in \eqref{eq:Lbetatilde} on the Hilbert space of Hilbert-Schmidt operators over $(\mathbb{C}^2)^{\otimes |\Lambda|}$ can be seen as the perturbation $\widetilde{H}_0+\widetilde{V}$ of a Hamiltonian $\widetilde{H}_0$ (corresponding to $\beta \rightarrow 0$) that satisfies the conditions of Theorem 
\ref{thm:stability} for $\beta$ small enough, from which we infer that it has a finite spectral gap. To do this, we study the convergence of $\widetilde{\mathcal{L}}^{(\beta)}$ as $\beta\to 0$ in norm $ \|.\|_{2\to 2}$. We first consider the following telescopic sum decomposition of $\widetilde{\mathcal{L}}_{a,\alpha}^{(\beta)}$:
	\begin{align}
		\widetilde{\mathcal{L}}_{a,\alpha}^{(\beta)}=\widetilde{\mathcal{L}}_{a,\alpha}^{\beta,0}+\sum_{r=0}^{\infty} \left(\widetilde{\mathcal{L}}_{a,\alpha}^{\beta,r+1}-\widetilde{\mathcal{L}}_{a,\alpha}^{\beta,r} \right),
	\end{align}
where for any $r\in\mathbb{N}$, $\widetilde{\mathcal{L}}_{a,\alpha}^{\beta,r}$ is defined with respect to the jump operators $A^{a,\alpha}_{r}(\omega)$ as well as the state $\sigma_r$ expressed analogously to $A^{a,\alpha}(\omega)$ and $\sigma_\beta$ up to replacing $H$ with the reduced Hamiltonian $H_{B_a(r)}$.  Therefore the perturbations $\widetilde{\mathcal{E}}_{a,\alpha}^{\beta,r}:=\widetilde{\mathcal{L}}_{a,\alpha}^{\beta,r+1}-\widetilde{\mathcal{L}}_{a,\alpha}^{\beta,r}$ are supported on regions $B_a(r+1)$. The $r=0$ case corresponds to the possible $1$-local, single qubit, terms in the Hamiltonian. 

\medskip

Considering the definition of the Lindbladian in Section \ref{sec:setup}, and the respective Hamiltonian in Equation \eqref{eq:Lbetatilde}, we can decompose the maps $\widetilde{\mathcal{L}}^{\beta,r}_{a,\alpha}$ into
	\begin{align*}
		\widetilde{\mathcal{L}}^{\beta,r}_{a,\alpha}:=\widetilde{\mathcal{C}}^{\beta,r}_{a,\alpha}+\widetilde{\Psi}^{\beta,r}_{a,\alpha}+\widetilde{\mathcal{D}}^{\beta,r}_{a,\alpha}\,,
	\end{align*}
	where 
	\begin{align*}
		&\widetilde{\mathcal{C}}^{\beta,r}_{a,\alpha}(X):=-i \Big[\Delta_{{\beta,r}}^{-\frac{1}{4}}(B_{a,\alpha}^{\beta,r})X-X\Delta_{\beta,r}^{\frac{1}{4}}(B_{a,\alpha}^{\beta,r})\Big]\,,\\
		&\widetilde{\mathcal{D}}^{\beta,r}_{a,\alpha}(X):=-\frac{1}{2}\int \Big[\Delta_{\beta,r}^{-\frac{1}{4}}\big(A^{a,\alpha}_{r}(\omega)^\dagger A^{a,\alpha}_{r}(\omega)\big)X+X\Delta_{\beta,r}^{\frac{1}{4}}\big(A^{a,\alpha}_{r}(\omega)^\dagger A^{a,\alpha}_{r}(\omega)\big)\Big]\gamma(\omega)d\omega\, ,\\
  &\widetilde{\Psi}^{\beta,r}_{a,\alpha}(X):= \int\gamma(\omega)\,\Delta_{\beta,r}^{-\frac{1}{4}}(A^{a,\alpha}_{r}(\omega))X\Delta_{\beta,r}^{\frac{1}{4}}(A^{a,\alpha}_r(\omega)^\dagger)\,d\omega\,.
	\end{align*}
	Above, we denoted by $\Delta_{\beta,r}^z(X)=\sigma_r^z X\sigma_r^{-z}$ the $z$-th power of the modular operator corresponding to the state $\sigma_r \propto {\operatorname{exp}({-\beta {H_{B_a(r)}} })}$, and the terms $B_{a,\alpha}^{\beta,r}$ are defined from Equation \eqref{eq:coherentterm}.  We aim at controlling the norm $\|\widetilde{\mathcal{E}}_{a,\alpha}^{\beta,r}\|_{2\to 2}$, showing that it decays sufficiently fast with $r$. 
	From Proposition A.1 and Corollary III.5 in \cite{chen2023efficient} we can see that
	\begin{align}\label{eq:Nrsumcode}
\widetilde{\mathcal{C}}^{\beta,r}_{a,\alpha}(X)+\widetilde{\mathcal{D}}^{\beta,r}_{a,\alpha}(X) = \frac{1}{2}\left(N^{\beta,r}_{a,\alpha} X + X N^{\beta,r}_{a,\alpha} \right),
	\end{align}
	where we define
	\begin{align}\label{eq:Nr}
		N^{\beta,r}_{a,\alpha}= \int_{-\infty}^\infty n_1(t) e^{-i\beta t H_{B_a(r)}}\left(\int_{-\infty}^\infty n_2(t')A_r^{a,\alpha \, \dagger}(\beta t') A_r^{a,\alpha}(-\beta t') dt' \right)e^{i\beta t H_{B_a(r)}} dt,
	\end{align}
with $A_r^{a,\alpha}(t)=e^{it H_{B_a(r)}}A^{a,\alpha}e^{-it H_{B_a(r)}}$ and
\begin{align} \label{eq:n1}
	&n_1(t):= \frac{\sqrt{\pi}}{2}\, \left(\frac{1}{\operatorname{cosh}(2\pi t)}\ast_t\, e^{-2t^2}\right)\\ \label{eq:n2}
&n_2(t):=8 b_2(t)=\frac{4}{\pi^{\frac{3}{2}}}\,\operatorname{exp}\big(-4t^2-2it\big)\,.
\end{align}
Notice that these are slightly different to the functions in \eqref{eq:b1b2} above due to the mapping in Equation \eqref{eq:Lbetatilde}.  We can thus bound
\begin{align}
	\norm{	\widetilde{\mathcal{C}}^{\beta,r+1}_{a,\alpha}(X)+\widetilde{\mathcal{D}}^{\beta,r+1}_{a,\alpha}(X) -	\widetilde{\mathcal{C}}^{\beta,r}_{a,\alpha}(X)-\widetilde{\mathcal{D}}^{\beta,r}_{a,\alpha}(X) }_2 \le  \norm{X}_2 \norm{	N^{\beta,r+1}_{a,\alpha}-	N^{\beta,r}_{a,\alpha}}_
	{\infty},
\end{align}
which is such that
\begin{align} \label{eq:Nbound}
	\norm{	N^{\beta,r+1}_{a,\alpha}-	N^{\beta,r}_{a,\alpha}}_
{\infty} \le  \int_{-\infty}^\infty \int_{-\infty}^\infty n_1(t) \vert n_2(t') \vert \norm{\mathcal{N}(t,t')}_{\infty} dt' dt,
\end{align}
where
$\mathcal{N}(t,t')= A_{r+1}^{a,\alpha \, \dagger}(\beta (t'-t)) A_{r+1}^{a,\alpha}(\beta( -t'-t)) - A_r^{a,\alpha \, \dagger}(\beta (t'-t)) A_r^{a,\alpha}(\beta(-t- t'))$.

Now let us divide the integrals in \eqref{eq:Nbound} as 
\begin{align*}
	\norm{	N^{\beta,r+1}_{a,\alpha}-	N^{\beta,r}_{a,\alpha}}_
{\infty} &\le  \int_{-t_0}^{t_0} \int_{-t_0}^{t_0} n_1(t) \vert n_2(t') \vert \norm{\mathcal{N}(t,t')}_{\infty} dt' dt \\ &+ 2\left( \int_{t_0}^\infty \int_{-\infty}^\infty + \int_{-\infty}^\infty \int_{t_0}^\infty + \int_{t_0}^\infty \int_{t_0}^\infty\right) n_1(t) \vert n_2(t') \vert \norm{\mathcal{N}(t,t')}_{\infty} dt' dt\,.
\end{align*}
For the first term (short times), we use the Lieb-Robinson bound recalled in Equation \eqref{eq:LiebRobinson}. Since $\|A^{a,\alpha}\|_\infty=1$, we have
\begin{align}
\norm{\mathcal{N}(t,t')}_{\infty}&=\norm{  A_{r+1}^{a,\alpha \, \dagger}(\beta (t'-t)) A_{r+1}^{a,\alpha}(\beta( -t'-t)) - A_r^{a,\alpha \, \dagger}(\beta (t'-t)) A_r^{a,\alpha}(\beta(-t- t'))}_\infty\nonumber \\
&\le 8\frac{\left(\beta J (\vert t \vert +\vert t' \vert) \right)^r}{r!}  ,\nonumber
\end{align}
which follows from the bound in Eq. \eqref{eq:LiebRobinson} and a repeated use of the triangle inequality. This means that
\begin{align}
    \int_{-t_0}^{t_0} \int_{-t_0}^{t_0} n_1(t) \vert n_2(t') \vert \norm{\mathcal{N}(t,t')}_{\infty} dt' dt & \le 8 \frac{\left(2 \beta J t_0 \right)^r}{r!} \int_{-\infty}^{\infty} \int_{-\infty}^{\infty} n_1(t) \vert n_2(t') \vert dt' dt 
    \\ & \le  2 \sqrt{2} \frac{\left(2 \beta J t_0 \right)^r}{r!} \le \frac{4}{\sqrt{2r}}  \left(\frac{2e\beta J t_0}{r} \right)^r.
\end{align}
Otherwise, using the triangle inequality, we can also bound
\begin{align}
\norm{\mathcal{N}(t,t')}_{\infty}&=\norm{  A_{r+1}^{a,\alpha \, \dagger}(\beta (t'-t)) A_{r+1}^{a,\alpha}(\beta( -t'-t)) - A_r^{a,\alpha \, \dagger}(\beta (t'-t)) A_r^{a,\alpha}(\beta(-t- t'))}_\infty\nonumber \\
&\le 8 \beta (\vert t \vert+\vert t \vert) \norm{H_{B_a(r)}-H_{B_a(r+1)}}_\infty \le 8 \beta (\vert t \vert+\vert t' \vert) r^D
\end{align}
so that we obtain 
\begin{align*}
	& 2\left( \int_{t_0}^\infty \int_{-\infty}^\infty + \int_{-\infty}^\infty \int_{t_0}^\infty + \int_{t_0}^\infty \int_{t_0}^\infty\right) n_1(t) \vert n_2(t') \vert \norm{\mathcal{N}(t,t')}_{\infty} dt' dt 
 \\
 &\le 32 (3(r+1))^D \beta h \left( \int_{t_0}^\infty \int_{-\infty}^\infty + \int_{-\infty}^\infty \int_{t_0}^\infty + \int_{t_0}^\infty \int_{t_0}^\infty\right) (\vert t \vert+\vert t' \vert) n_1(t) \vert n_2(t') \vert dt' dt
 \\
 & \le 32 (3(r+1))^D \beta h  \left( 2 e^{-2 t_0^2} + (11+48 t_0) e^{-2\pi t_0} \right),
\end{align*}
where the second line follows from a trivial bound to the volume of a $D-$sphere and the third line follows from algebraic manipulations of the integrals.
By fixing $t_0=\frac{r g(\beta J)}{2   \beta J }$ with $g(x)=\frac{\sqrt{x}}{1+\sqrt{x}}$, we thus obtain
\begin{align}
	\norm{	N^{\beta,r+1}_{a,\alpha}-	N^{\beta,r}_{a,\alpha}}_
{\infty} &\le 4  \frac{g(\beta J)^{r }}{\sqrt{2 r}} + \beta h \times e^{-\Omega \left (r\frac{ g(\beta J)} {\beta J} \right )},
\end{align}
which is an exponentially decaying interaction with the leading order term decaying with rate $\frac{ g(\beta J)} {\beta J}$.

Now we bound the transition part $\widetilde{\Psi}^{\beta,r}_{a,\alpha}(X)$. As can be seen in \cite{chen2023efficient} (Corollary A.3), we can write 
\begin{align*}
	\widetilde{\Psi}^{\beta,r}_{a,\alpha}(X)=\int_{-\infty}^\infty \int_{-\infty}^\infty h_-(t_-) h_+(t_+)A_{r}^{a,\alpha }(\beta (t_+-t_-))X A_{r}^{a,\alpha}(\beta (-t_+-t_-))^\dagger dt_+d t_-,
\end{align*}
where
\begin{align*}
	h_+(t)=e^{-1/4-4 t^2} \quad, \quad 	h_-(t)=\frac{1}{\pi}e^{-2t^2}.
\end{align*}
Considering that, again by the Lieb-Robinson bound \eqref{eq:LiebRobinson},
\begin{align*}
&\norm{	A_{r+1}^{a,\alpha }(\beta (t_+-t_-))X A_{r+1}^{a,\alpha }(\beta (-t_+-t_-))^\dagger -	A_{r}^{a,\alpha }(\beta (t_+-t_-))X A_{r}^{a,\alpha }(\beta (-t_+-t_-))^\dagger}_2\\
&~~~~~~~~~~~~~~~~~~~~~~~~~~~~~~~~~~~~~~~~~~~~~~~~~~~~~~~~~~~~~~~~~~~~~~~~~~~~~~~~~~~~~~~~\le \norm{X}_2  8\frac{\left( \beta J\, (\vert t_+ \vert +\vert t_- \vert) \right)^r}{r!},
\end{align*}
we can simply bound, by again estimating the integral over $t_+,t_-$ ,
\begin{align*}	\norm{\widetilde{\Psi}^{\beta,r+1}_{a,\alpha}(X)-	\widetilde{\Psi}^{\beta,r}_{a,\alpha}(X)}_2 \le 4 \norm{X}_2 \frac{\left(\frac{\sqrt{3}\beta J}{4}\right)^r}{\Gamma(1+\frac{r}{2})}\,.
\end{align*}
Putting all the bounds together, we can write,
\begin{align} \label{eq:decaywithr}
\|\widetilde{\mathcal{E}}^{\beta, r}_{a,\alpha}\|_{2\to 2}\le 4 \frac{g(\beta J)^{-r }}{\sqrt{2 r}} + \beta h \times e^{-\Omega \left (r\frac{ g(\beta J)} {\beta J} \right )} + 4 \frac{\left(\frac{\sqrt{3}\beta J}{4}\right)^r}{\Gamma(1+\frac{r}{2})}\,,
\end{align}
which vanishes as $\beta \rightarrow 0$.
The next step is to determine the point we perturb around $\widetilde{\mathcal{L}}_{a,\alpha}^{0,0}$. Notice that $\lim_{\beta \rightarrow 0} 	N^{\beta,r}_{a,\alpha} = \lambda {I}$ for some constant $\lambda$, and so $\lim_{\beta \rightarrow 0}	\widetilde{\mathcal{C}}^{\beta,r}_{a,\alpha}(X)+\widetilde{\mathcal{D}}^{\beta,r}_{a,\alpha}(X)=\lambda X$, which is the identity channel up to $\lambda$. Additionally, 
\begin{align*}
		\lim_{\beta \rightarrow 0} \widetilde{\Psi}^{\beta,0}_{a,\alpha}(X)=\frac{1}{2 \sqrt{2}e^{1/4}} A^{a,\alpha }X A^{a,\alpha \, \dagger}.
\end{align*}
This means that $\lambda=-\frac{1}{\sqrt{2}e^{1/4}}$.
When summing over all Pauli matrices $A^{a,\alpha }$, we obtain the generator of the local depolarizing channel. Finally, we bound the remaining contribution, which is the norm difference $\|\widetilde{\mathcal{L}}^{\beta,0}_{a,\alpha}-\widetilde{\mathcal{L}}^{0,0}_{a,\alpha}\|_{2\to 2}$ between the $1$-qubit generator at $\beta$ and that at $\beta=0$. Note that
\begin{align}\label{eq:NId}
	\norm{	N^{\beta,0}_{a,\alpha}-\lambda {I}}_\infty &\le  \int_{-\infty}^\infty n_1(t)    \int_{-\infty}^\infty  \vert n_2(t') \vert \norm{A_0^{a,\alpha \, \dagger}(\beta t') A_0^{a,\alpha}(-\beta t')- {I}}_\infty dt'dt \\ \nonumber & \le
 \int_{-\infty}^\infty n_1(t)    \int_{-\infty}^\infty  \vert n_2(t') \vert \vert \beta t' \vert 2 \norm{   [H,A_0^{a,\alpha \, \dagger}] }_\infty dt' dt \\ \nonumber &
  \le
 (4 \beta hl) \int_{-\infty}^\infty n_1(t)   \int_{-\infty}^\infty  \vert n_2(t') \vert  \vert t' \vert  dt' dt 
  \\ \nonumber  & = \frac{4 \beta hl}{\pi^{\frac{3}{2}}} \times  \int_{-\infty}^\infty n_1(t)  dt = 2 \beta hl , 
\end{align}
where we bounded $\norm{A_0^{a,\alpha \, \dagger}(t)-A_0^{a,\alpha \, \dagger}}_\infty\le \vert t \vert \norm{ [H,A_0^{a,\alpha \, \dagger}]}_\infty \le 2 h l $ . For the last part, we compare $\widetilde{\Psi}^{\beta,0}_{a,\alpha}(X)$ with the depolarizing generator $\widetilde{\Psi}^{0,0}_{a,\alpha}(X)=\frac{1}{2 \sqrt{2}e^{1/4}} A^{a,\alpha }X A^{a,\alpha \, \dagger}$ in a similar way as Equation \eqref{eq:NId}, so that

\begin{align*}
	&\norm{\widetilde{\Psi}^{\beta,0}_{a,\alpha}(X)-\widetilde{\Psi}^{0,0}_{a,\alpha}(X)}_2 \\
 &~~~\le \int_{-\infty}^\infty \int_{-\infty}^\infty h_-(t_-) h_+(t_+)\norm{ A_{0}^{a,\alpha } (\beta(t_+-t_-))X A_{0}^{a,\alpha} (\beta(-t_+-t_-))^\dagger -A_{0}^{a,\alpha } X A_{0}^{a,\alpha \, \dagger}}_2 dt_+d t_-  \\ &~~~
 \le \norm{X}_2 \int_{-\infty}^\infty \int_{-\infty}^\infty  h_-(t_-) h_+(t_+) \times  \\&~~~ \nonumber \quad \quad \quad  \left( \norm{A_{0}^{a,\alpha } (\beta(t_+-t_-))-A_{0}^{a,\alpha }}_\infty +\norm{A_{0}^{a,\alpha \, \dagger} (\beta(-t_+-t_-))-A_{0}^{a,\alpha \, \dagger }}_\infty \right) dt_+d t_- 
  \\ &~~~
 \le \norm{X}_2  (2\beta hl) \int_{-\infty}^\infty \int_{-\infty}^\infty  h_-(t_-) h_+(t_+) \left( \vert t-t' \vert +\vert t + t' \vert \right)   dt_+d t_-  \\ &~~~
 \le \beta hl\,\|X\|_2\,.
\end{align*}
This shows that $\sum_{a,\alpha}\widetilde{\mathcal{L}}^{\beta,0}_{a,\alpha}$ converges to the generator of the local depolarizing semigroup with a multiplicative constant $\lambda$ as $\beta\to 0$. Clearly the latter is gapped with $\gamma=-\lambda=\frac{1}{\sqrt{2}e^{1/4}}$, frustration-free, satisfies local topological order with sharp cutoff and is also locally gapped. Therefore, $\widetilde{H}_0:=\widetilde{\mathcal{L}}^{0,0}=\sum_{a,\alpha}\widetilde{\mathcal{L}}^{0,0}_{a,\alpha}$ satisfies the conditions of the unperturbed Hamiltonian in Theorem \ref{thm:stability}. Together with the exponential decay of Equation \eqref{eq:decaywithr}, this shows that the perturbed generator $\widetilde{\mathcal{L}}_{a,\alpha}^{(\beta)}$ obeys all of the assumptions of Theorem \ref{thm:stability}. In this case, both global and local gaps are $\gamma$, the decay $f(r)$ is given by Equation \eqref{eq:decaywithr}, and the existence of a finite Lieb-Robinson velocity follows from the quasi-locality of $\widetilde{\mathcal{L}}_{a,\alpha}^{(\beta)}$ \cite{HastingsKoma2006}.  This implies that there exists a small enough $\beta^*=\mathcal{O}(J^{-1})$ \footnote{That dependence follows from the fact that the decay in Eq. \eqref{eq:decaywithr} only depends on the products $\beta J$ and $\beta h$, and $J \ge h$. The proportionality constant should depend on the lattice dimension.}, where the exact proportionality depends on the lattice dimension and the Lieb-Robinson velocity as determined by the constants in Theorem \ref{thm:stability}, such that for $\beta < \beta^*$ the spectral gap of $\widetilde{\mathcal{L}}^{(\beta)}$, and therefore that of $\mathcal{L}^{(\beta)}$, is lower bounded by $\frac{1}{2\sqrt{2}e^{1/4}}$.

\subsection{Adiabatic algorithm} \label{app:adiabatic}

We now make use of the gap proved in Theorem \ref{thmgaphighT} to devise an adiabatic evolution on the $n$-fold tensor product of Hilbert spaces $\widetilde{\mathcal{H}}_v\simeq  \mathbb{C}^2\otimes \mathbb{C}^2$ interpolating between $\mathscr{H}(0):=\widetilde{\mathcal{L}}^{(0)}$ and $\mathscr{H}(1):=\widetilde{\mathcal{L}}^{(\beta)}$. For this, we resort to the following standard estimate on the performance of the adiabatic evolution:
\begin{lemma}\cite{ambainis2006elementary}
Let $\mathscr{H}(s)$ with $ s \in [0,1 ]$ be an adiabatic path with minimum gap $\gamma$. The minimum adiabatic time $T_{\operatorname{ad}}$ required to achieve an error $\epsilon$ in the final state is bounded as
\begin{equation}
T_{\operatorname{ad}} \le \frac{10}{\epsilon^2} \max \left\{\frac{\norm{\frac{d \mathscr{H}}{ds}}_\infty^3}{\gamma^4}, \frac{\norm{\frac{d \mathscr{H}}{ds}}_\infty \norm{\frac{d^2 \mathscr{H}}{ds^2}}_\infty}{\gamma^3} \right\}\,.
\end{equation}
\end{lemma}
In our case, the adiabatic path is $\mathscr{H}(s)=\widetilde{\mathcal{L}}^{(\beta_s)}$, with $\beta_s:=s\beta$, starting at $\beta_0=0$ and ending at the desired temperature $\beta=\beta_1$ at $s=1$.
Thus, we need to estimate the derivatives of the Hamiltonian path $\{\mathscr{H}_{ s}\}_{s\in [0,1]}$ w.r.t. parameter $s$, done in the subsections below. In particular, in \eqref{eq:firstdev} and \eqref{eq:seconddev} we show that 

\begin{align}
&\max_{s\in [0,1]}\norm{\frac{d}{ds} {\mathscr{H}}(s)}_{\infty}=\max_{\beta'\in [0,\beta]}\beta'\norm{\frac{d}{d\beta'} \widetilde{\mathcal{L}}^{(\beta')}}_{2\to 2} \le 61\, \beta \, n\,\max_{a,\alpha} \norm{\left[H,A^{a,\alpha} \right]}_\infty\label{eq:devs}\\
&\max_{s\in[0,1]}\norm{\frac{d^2}{ds^2} {\mathscr{H}}(s)}_{\infty}=\max_{\beta'\in[0,\beta]}(\beta')^2\norm {\frac{d^2}{d\beta'^2} \widetilde{\mathcal{L}}^{(\beta')}}_{2\to 2} \le C\,\beta^2\,n \left( \norm{[H,[H,A^{a,\alpha}]]}_\infty+\norm{[H,A^{a,\alpha}]}_\infty^2 \right) \,,
\end{align}
for some constant $C>0$. Morever, notice that for a $(k,l)$-local Hamiltonian, with the notation of Section \ref{sec:setup},
\begin{align*}
    \max_{a,\alpha} \norm{\left[H,A^{a,\alpha} \right]}_\infty &\le 2 hl \\
    \max_{a,\alpha} \norm{[H,[H,A^{a,\alpha}]]}_\infty & \le 4 h^2l^2 k, 
\end{align*}
which are $\mathcal{O}$(1) constants. Combining these bounds with Theorem \ref{thmgaphighT}, we obtain the bound on the runtime of the adiabatic evolution preparing the purified Gibbs state corresponding to $\sigma_\beta$:

\begin{theorem}[Run-time of the adiabatic evolution]\label{thm:adiabatic_runtime}
In the notations of the previous paragraphs, and given $\beta^*$ as introduced in Theorem \ref{thmgaphighT}, for any $\beta\in[0,\beta^*]$, the minimum time required to adiabatically prepare a state $\epsilon$-close to the purified Gibbs state at inverse temperature $\beta$ is bounded as 
\begin{align}
T_{\operatorname{ad}}\le  \frac{  C'  k l^2 (\beta h n )^3 }{\epsilon^2}\,,
\end{align}
for some universal constant $C'>0$.
\end{theorem}

\subsubsection{First derivative}	

It remains to bound the first and second derivatives of $\widetilde{\mathcal{L}}^{(\beta)}$ w.r.t.~the parameter $\beta$. This Hamiltonian can be written in terms of two contributions $\widetilde{\mathcal{L}}^{(\beta) }=\mathscr{H}_{\text{cd}}+\mathscr{H}_{\text{tr}}$. First, it is possible to bound at once the coherent term corresponding to $ \widetilde{\mathcal{C}}^\beta$, which is $	\Delta_\beta^{-\frac{1}{4}}	(B) \otimes {I} + {I} \otimes \Delta_\beta^{-\frac{1}{4}}	(B)^* $ in the purified picture, and the decay term, corresponding to $ \widetilde{\mathcal{D}}_\beta$. Similarly to Equation \eqref{eq:Nrsumcode}, in Corollary III.5 from \cite{chen2023efficient}, it is shown that their total contribution to the effective Hamiltonian is of the form $\mathscr{H}_{\text{cd}}=\frac{1}{2}(N \otimes {I}+ {I} \otimes N^*)$ in the purified picture, where $N^*$ stands for the complex conjugate of $N$ when decomposed in the computational basis, and where we recall that
	\begin{align*}
N= \sum_{a\in \Lambda}\sum_{\alpha\in [3]}\, \int_{-\infty}^\infty n_1(t) e^{-i\beta Ht} \int_{-\infty}^\infty n_2(t') e^{i\beta Ht'}A^{a,\alpha} e^{-2i\beta Ht'}A^{a,\alpha} e^{i\beta Ht'}dt' e^{i\beta Ht}dt,
	\end{align*}
with functions $n_1$ and $n_2$ defined in Equation \eqref{eq:n1} and \eqref{eq:n2} respectively. We see that
\begin{align*}
	\frac{d}{d\beta} N  = \sum_{a\in \Lambda}\sum_{\alpha\in [3]}\, \int_{-\infty}^\infty n_1(t)   \int_{-\infty}^\infty n_2(t') \mathcal{M}(t,t') dt'dt.
\end{align*}
where
\begin{align} \label{eq:hoft} \mathcal{M}(t,t')&=-it \left[H,
e^{i\beta H(t'-t)}A^{a,\alpha} e^{-2i\beta Ht'}A^{a,\alpha} e^{i\beta H(t'+t)} \right]\\&+it' e^{i\beta H(t'-t)}[H,A^{a,\alpha}] e^{-2i\beta Ht'}A^{a,\alpha} e^{i\beta H(t'+t)}+it' e^{i\beta H(t'-t)}A^{a,\alpha} e^{-2i\beta Ht'} [H,A^{a,\alpha}] e^{i\beta H(t'+t)}\,.\nonumber
\end{align}
After some algebra, we find that
\begin{align*}
    \norm{\mathcal{M}(t,t')}_\infty \le 2 \norm{[H,A^{a,\alpha}]}_\infty \left(\vert t \vert + \vert t' \vert \right),
\end{align*}
so that we can bound
\begin{align*}
\norm{	\frac{d}{d\beta} N }_\infty &\le 
 2 \max_{a,\alpha}  \norm{[H,A^{a,\alpha}]}_\infty \sum_{a\in \Lambda}\sum_{\alpha\in [3]}\, \norm{A^{a,\alpha}}_\infty^2 \int_{-\infty}^\infty \vert n_1(t) \vert   \int_{-\infty}^\infty \vert n_2(t') \vert \left(\vert t \vert + \vert t' \vert \right) dt'dt \\& = \frac{ 6 e^{\frac{\pi ^2}{2}}}{\sqrt{2} \pi ^2} \,n\, \max_{a,\alpha}  \norm{[H,A^{a,\alpha}]}_\infty.
\end{align*}
Here, we used that $\|A^{a,\alpha}\|_\infty=1,$ and evaluated the integrals explicitly.

\medskip

Similarly, we upper bound the derivative of the transition term  of the Hamiltonian, corresponding to the contribution of $\widetilde{\Psi}_\beta$, which we can write as 
\begin{align*}
 \mathscr{H}_{\text{tr}}	=  \sum_{a\in \Lambda}\sum_{\alpha\in [3]}\, \int_{-\infty}^\infty \int_{-\infty}^\infty  h_-(t_-) h_+(t_+) A^{a,\alpha}\left(\beta(t_+-t_-)\right)  \otimes  A^{a,\alpha \,*}\left(\beta(t_+-t_-) \right) dt_-dt_+,
\end{align*}
 where we have re-scaled $t,t'$ so that $h_-(t_-), h_+(t_+)$  are independent of $\beta$. We get 
\begin{align*}
\norm{\frac{d}{d\beta} \mathscr{H}_{\text{tr}}}_\infty &\le  \sum_{a\in \Lambda}\sum_{\alpha\in [3]}\, \int_{-\infty}^\infty 2 h_-(t_-) h_+(t_+) \norm{A^{a,\alpha}}_\infty \norm {	\frac{d}{d\beta} \Delta_\beta^{i(t_+-t_-)}	(A^{a,\alpha}) }_\infty
\\ & \le 6 \,n\, \max_{a,\alpha}\norm{\left[H,A^{a,\alpha} \right]}_\infty \int_{-\infty}^\infty \int_{-\infty}^\infty \vert t_+-t_-\vert h_-(t_-) h_+(t_+) dt_-dt_+, \\& \le 
\frac{3 \left(\sqrt{2}+2\right)}{4 \sqrt[4]{e} \sqrt{\pi }}  \,n\, \max_{a,\alpha} \norm{\left[H,A^{a,\alpha} \right]}_\infty,
\end{align*}
where we used once again that $\norm{A^{a,\alpha}}_\infty = 1$ and
\begin{align}\label{eq:d1beta}
\norm {	\frac{d}{d\beta} \Delta_\beta^{i(t_+-t_-)}	A^{a,\alpha} }_\infty &\le \vert t_+-t_-\vert \norm{\left[H,A^{a,\alpha} \right]}_\infty .
\end{align}
Putting everything together, we have that 
\begin{align}\label{eq:firstdev}
\norm {\frac{d}{d\beta} \widetilde{\mathcal{L}}^{(\beta)}}_{2\to 2} \le 61 \,n\, \max_{a,\alpha} \norm{\left[H,A^{a,\alpha} \right]}_\infty.
\end{align}

\subsubsection{Second derivative}

For the coherent and decaying parts, we have
\begin{align*}
\frac{d^2}{d\beta^2} N  = \sum_{a\in \Lambda}\sum_{\alpha\in [3]}\, \int_{-\infty}^\infty n_1(t)   \int_{-\infty}^\infty n_2(t') \frac{d}{d\beta} \mathcal{M}(t,t') dt'dt, 
\end{align*}
we can differentiate $\frac{d}{d\beta} \mathcal{M}(t,t')$ from Equation \eqref{eq:hoft} to obtain a somewhat lengthy expression of nested commutators. From this we obtain 
\begin{align*}
\norm{\frac{d}{d\beta} \mathcal{M}(t,t')}_\infty \le \left( \norm{[H,[H,A^{a,\alpha}]]}_\infty+\norm{[H,A^{a,\alpha}]}_\infty^2 \right)\left(c_1 t^2 + c_2 t'^2 +c_3 tt') \right),
\end{align*}
for $c_i$ some $\mathcal{O}(1)$ constants. Proceeding as above, this means that, for some constant $C_{\text{cd}}$,
\begin{align*}
\norm{\frac{d^2}{d\beta^2} N }_\infty \le \frac{C_{\text{cd}}}{2} \left( \norm{[H,[H,A^{a,\alpha}]]}_\infty+\norm{[H,A^{a,\alpha}]}_\infty^2 \right).
\end{align*}
For the transition part, using the chain rule and the triangle inequality, we have that 
\begin{align*}
\norm{\frac{d^2}{d\beta^2} \mathscr{H}_{\operatorname{tr}}}_\infty &\le  \sum_{a\in \Lambda}\sum_{\alpha\in [3]}\, \int_{-\infty}^\infty 2 h_-(t_-) h_+(t_+) \left( \norm {	\frac{d^2}{d\beta^2} \Delta_\beta^{i(t_+-t_-)}	A^{a,\alpha} }_\infty +\norm {	\frac{d}{d\beta} \Delta_\beta^{i(t_+-t_-)} A^{a,\alpha} }_\infty^2 \right).
\end{align*}
The second term is bounded by \eqref{eq:d1beta}. The first follows from the fact that
\begin{align*}
	\frac{d^2}{d\beta^2} \Delta_\beta^{i(t_+-t_-)}	A^{a,\alpha} = -(t_+-t_-)^2 [H,[H,\Delta_\beta^{i(t_+-t_-)	}A^{a,\alpha}]],
\end{align*}
so that 
$
\norm {	\frac{d^2}{d\beta^2} \Delta_\beta^{i(t_+-t_-)-\frac{1}{4}}	A^{a,\alpha} }_\infty \le (t_+-t_-)^2  \norm{[H,[H,A^{a,\alpha}]]}_\infty$.
The contribution is bounded by
\begin{align*}
\norm{\frac{d^2}{d\beta^2} \mathcal{H}_{\operatorname{tr}}}_\infty &\le 6 n\, \max_{a,\alpha} \norm{[H,[H,A^{a,\alpha}]]}_\infty   \int_{-\infty}^\infty  h_-(t_-) h_+(t_+) (t_+-t_-)^2 dt_-dt_+ \\
&\le \frac{3 \left(3 \pi  \sqrt{2}+8\right)}{16 \sqrt[4]{e} \pi } n\,  \max_{a,\alpha} \norm{[H,[H,A^{a,\alpha}]]}_\infty.
\end{align*}
The numerical constant above is close to $1$, so that finally we bound the second derivative as 
\begin{align}\label{eq:seconddev}
\norm {\frac{d^2}{d\beta^2} \widetilde{\mathcal{L}}^{(\beta)}}_{2\to 2} \le n\, \left( \norm{[H,[H,A^{a,\alpha}]]}_\infty+\norm{[H,A^{a,\alpha}]}_\infty^2 \right) \left( C_{\operatorname{cd}} +1  \right).
\end{align}


\section{Zero temperature quantum Gibbs sampling and the complexity class $\mathsf{GibbsQP}$}\label{GibbsQPsec}

In this section, we prove the final main theorem of the present paper. For this, we first need a characterization of spectral gaps of Lindbladians in terms of the so-called Poincar\'{e} inequality \cite{PoincareRef,Kastoryano_2013}: for any $X\in \mathcal{O}_V$,
\begin{align}\label{eq:vargap}
\operatorname{gap}(\mathcal{L}^{(\sigma_E,\beta)})\,\operatorname{Var}_{\sigma_\beta }(X)\le \mathcal{E}^{(\sigma_E,\beta)}(X)\,.
\end{align}
Above, $\operatorname{Var}_{\sigma_\beta }(X):=\|X-\tr{\sigma_\beta X}I\|^2_{\sigma_\beta}$ corresponds to the variance of $X$ in the Gibbs state $\sigma_\beta$, whereas $\mathcal{E}^{(\sigma_E,\beta)}(X):=-\langle X,\mathcal{L}^{(\sigma_E,\beta)\dagger}(X)\rangle_{\sigma_\beta}$ is the so-called Dirichlet form associated to the Metropolis-type generator \cite{chen2023efficient}: 
\begin{lemma} For an inverse temperature $\beta>0$, a tunable width $\sigma_E$ and a Hamiltonian $H=\sum_E EP_E\ge 0$ with Bohr frequencies $\nu\in B(H)$,
\begin{align}\label{eq:initialL}
\mathcal{L}^{(\sigma_E,\beta)\dagger}(X):=\!\!\!\sum_{a,\nu_1,\nu_2}\alpha_{\nu_1,\nu_2}\left(A^{a\dagger}_{\nu_2}XA^a_{\nu_1}-\frac{1}{2\operatorname{cosh}\big(\frac{\beta(\nu_1-\nu_2)}{4}\big)}\Big[e^{\frac{\beta}{4}(\nu_1-\nu_2)}A^{a\dagger}_{\nu_2}A^a_{\nu_1}X+e^{\frac{\beta}{4}(\nu_2-\nu_1)}XA^{a\dagger}_{\nu_2}A^a_{\nu_1}\Big]\right)
\end{align}
with coefficients $\alpha_{\nu_1,\nu_2}$ given by 
\begin{align}\label{alphaexactexpre}
\alpha_{\nu_1,\nu_2}=\frac{e^{-\beta\frac{(\nu_1+\nu_2)}{4}-\frac{(\nu_1-\nu_2)^2}{8\sigma_E^2}}}{4}\,\left(e^{-\beta\frac{\nu_1+\nu_2}{4}}\operatorname{Erfc}\Big(\frac{\beta\sigma_E^2-(\nu_1+\nu_2)}{2\sqrt{2}\sigma_E}\Big)+e^{\beta\frac{\nu_1+\nu_2}{4}}\operatorname{Erfc}\Big(\frac{\nu_1+\nu_2+\beta\sigma_E^2}{2\sqrt{2}\sigma_E}\Big)\right)\,.
\end{align}
Above, $\{A^{a}\}_{a\in[m]}$ is a set of self-adjoint jump operators to be fixed later, and $$A^a_\nu:=\sum_{E-E'=\nu}P_EA^a P_{E'}\,.$$
\end{lemma}
\begin{proof}
A general form for the generator of a KMS-symmetric Gibbs sampler was given in \cite[Corollary II.2]{chen2023efficient}: in the Heisenberg picture:
\begin{align*}
\mathcal{L}^\dagger(X)=i\sum_{\nu}[B_\nu,X]+\sum_{a,\nu_1,\nu_2}\alpha_{\nu_1,\nu_2}\left(A^{a\dagger }_{\nu_2} X A^a_{\nu_1}-\frac{1}{2}\{A^{a\dagger}_{\nu_2} A^a_{\nu_1},X\}\right)
\end{align*}
for some coefficients $\alpha_{\nu_1,\nu_2}$ satisfying certain algebraic conditions, and 
\begin{align*}
B_\nu:=\sum_{a,\nu_1-\nu_2=\nu}\frac{\operatorname{tanh}(-\beta(\nu_1-\nu_2)/4)}{2i}\alpha_{\nu_1,\nu_2}A_{\nu_2}^{a\dagger}A^a_{\nu_1}\,.
\end{align*}
We recall from \cite{chen2023efficient} the coefficients $\alpha_{\nu_1,\nu_2}$ associated to the Metropolis-type generator $\mathcal{L}^{(\sigma_E,\beta)\dagger}$:
\begin{align}\label{alphametrochen}
\alpha_{\nu_1,\nu_2}=\int_{\frac{\sigma_E^2\beta}{2}}^\infty \frac{\sqrt{\beta}}{4\sqrt{\pi x}}\, e^{-\frac{\beta(\nu_1+\nu_2+2x)^2}{16 x}}dx\,e^{-\frac{(\nu_1-\nu_2)^2}{8\sigma_E^2}}\equiv G^{(\sigma_E,\beta)}(\nu_1+\nu_2) e^{-\frac{(\nu_1-\nu_2)^2}{8\sigma_E^2}}\,,
\end{align}
with 
\begin{align*}
G^{(\sigma_E,\beta)}(\nu):=\int_{\frac{\beta\sigma_E^2}{2}}^\infty\, \frac{1}{4}\sqrt{\frac{\beta}{\pi x}}e^{-\frac{\beta \nu^2}{16 x}-\frac{\beta x}{4}}dx\,.
\end{align*}
A direct integration gives the expression given in Equation \eqref{alphaexactexpre}. Putting things together and using that $1\pm\operatorname{tanh}(\beta(\nu_1-\nu_2)/4)=e^{\pm\beta(\nu_1-\nu_2)/4}/\operatorname{cosh}(\beta(\nu_1-\nu_2)/4)$, the expression for $\mathcal{L}^{(\sigma_E,\beta)\dagger}$ follows.
\end{proof}
Just as for its Gaussian analogue used in \Cref{sec:highTgibbsgap}, the Lindbladian constructed above is reversible in the Heisenberg picture:
\begin{align*}
\langle X,\mathcal{L}^{(\sigma_E,\beta)\dagger}(Y)\rangle_{\sigma_\beta}=\langle \mathcal{L}^{(\sigma_E,\beta)\dagger}(X),Y\rangle_{\sigma_\beta}\,,
\end{align*}
 where the Gibbs state $\sigma_\beta$ is the unique invariant state of the evolution generated by $\mathcal{L}^{(\sigma_E,\beta)}$ whenever the commutant $\{A^a\}'=\mathbb{C}I$. In the limit $\sigma_E\to 0$, the above generator converges to that of a Davies dynamics:
\begin{align}\label{eq:initialL0}
\mathcal{L}^{(0,\beta)\dagger}(X)=\frac{1}{2}\sum_{a,\nu}\Big(e^{-\beta\nu}1(\nu>0)+1(\nu\le 0)\Big) \,\Big[A^{a\dagger}_\nu XA^a_\nu-\frac{1}{2}\{A^{a\dagger}_\nu A^a_\nu,X\}\Big]\,.
\end{align}

Indeed, in the limit $\sigma_E\to 0$, it is easy to see that $\alpha_{\nu_1,\nu_2}\to 2^{-1}\delta_{\nu_1,\nu_2}\big(e^{-\beta\nu_1}1(\nu_1>0)+1(\nu_1\le 0)\big)$.
In what follows, we will also make use of two other representations of the dynamics: first, we recall the mapping of KMS-symmetric Lindbladians into Hamiltonians already introduced in Section \ref{sec:setup}: for any $\beta>0$, we define the operator 
\begin{align}\label{eq:Lbetatilde2}
\widetilde{\mathcal{L}}^{(\sigma_E,\beta)}(X):=\sigma_\beta^{-\frac{1}{4}}\mathcal{L}^{(\sigma_E,\beta)}(\sigma_\beta^{\frac{1}{4}}X\sigma_\beta^{\frac{1}{4}} )\sigma_\beta^{-\frac{1}{4}}\,.
\end{align}
Using the expression for $\mathcal{L}^{(\sigma_E,\beta)}$ above, we find that (see also \cite[Proposition A.1]{chen2023efficient})
\begin{align}\label{eqLbetafourier}
\widetilde{\mathcal{L}}^{(\sigma_E,\beta)}(X):=\sum_{a,\nu_1,\nu_2}h_{\nu_1,\nu_2}\,A^a_{\nu_1}X(A^a_{\nu_2})^\dagger+\frac{1}{2}\big\{N_{\sigma_E,\beta},X\big\}
\end{align}
with $$h_{\nu_1,\nu_2}=\alpha_{\nu_1,\nu_2}e^{\beta(\nu_1+\nu_2)/4}$$ and 
\begin{align}\label{defnnu1nu2}N_{\sigma_E,\beta}:=-\sum_{a,\nu_1,\nu_2}\frac{\alpha_{\nu_1,\nu_2}}{\cosh(\beta(\nu_1-\nu_2)/4)} (A^a_{\nu_2})^\dagger A^a_{\nu_1}=N_{\sigma_E,\beta}^\dagger\equiv -\sum_{a,\nu_1,\nu_2}n_{\nu_1,\nu_2}(A^a_{\nu_2})^\dagger A^a_{\nu_1} \,.
\end{align}
Since $\sigma_\beta$ is a steady state of the evolution and we chose the jump operators such that it is also unique, $\operatorname{Ker}(\widetilde{\mathcal{L}}^{(\sigma_E,\beta)})= \vert \sqrt{\sigma_\beta} \rangle$. Moreover, $\widetilde{\mathcal{L}}^{(\sigma_E,\beta)}$ is self-adjoint with respect to the Hilbert-Schmidt inner product for all $\sigma_E,\beta>0$, by the KMS symmetry of $\mathcal{L}^{(\sigma_E,\beta)\dagger}$. In the limit $\sigma_E\to 0$, we get
\begin{align} \label{eq:L0beta}
\widetilde{\mathcal{L}}^{(0,\beta)}(X)=\frac{1}{2}\sum_{a,\nu}e^{-\frac{\beta|\nu|}{2}}A^a_\nu XA^{a\dagger}_\nu-\frac{1}{4}\sum_{a,\nu>0}e^{-\beta\nu}\{A^{a\dagger}_\nu A^{a}_\nu,X\}-\frac{1}{4}\sum_{a,\nu\le 0}\{A^{a\dagger}_\nu A^a_\nu,X\}\,.
\end{align}
Hence, in the limit $\beta\to\infty$, we find that
\begin{align}\label{eq:L0inf}
\widetilde{\mathcal{L}}^{(0,\infty)}(X)=\frac{1}{2}\sum_{a}\left(A^a_0 XA^{a}_0-\frac{1}{2}\{(A^{a}_0)^2,X\}\right)-\frac{1}{4}\sum_{a,\nu<0}\{A^{a\dagger}_\nu A^{a}_\nu,X\}\,.
\end{align}
We also recall the time representation of the generator $\widetilde{\mathcal{L}}^{(\sigma_E,\beta)}$ which we have already used in \Cref{appendixGibbssampling} in the context of a Gaussian filter. In \cite{chen2023efficient} (Section III.B, Proposition A.2, Corollary III.4, Corollary A.4 and Proposition B.1), it was proven that
\begin{align}
\widetilde{\mathcal{L}}^{(\sigma_E,\beta)}_H(X)=\mathcal{T}^{(\sigma_E,\beta)}_H(X)+\frac{1}{2}\{N^{(\sigma_E,\beta)}_H,X\}
\end{align}
where 
\begin{align}\label{eq:bigT}
\mathcal{T}^{(\sigma_E,\beta)}_H(X)=\sum_a\int_{\mathbb{R}^2}h_-(t_-)h_+(t_+) A(t_+-t_-)X A(-t_--t_+) dt_+dt_-\,,
\end{align}
with
\begin{align*}
&h_-(t):=\frac{\sigma_E}{\pi}e^{-2 \sigma_E^2 t^2} ,
&h_+(t):= \frac{1}{4\beta \sqrt{2 \pi}} \frac{e^{-2\sigma_E^2 t^2-\frac{\beta^2 \sigma_E^2}{8}}}{\frac{t^2}{{\beta^2}}+\frac{1}{16}} ,
\end{align*}
and 
\begin{align}\label{eq:NsigmabetaH}
N^{(\sigma_E,\beta)}_H=\sum_a\int_{\mathbb{R}^2}f_-(t_-)f_+(t_+)\,e^{-iHt_-}A^{a\dagger}(t_+)A^a(-t_+)e^{iHt_-}dt_+dt_-\,,
\end{align}
with
\begin{align*}
&f_-(t):=\frac{2\sigma_E}{\pi\beta}\,\left(\frac{1}{\cosh\Big(\frac{2\pi t}{\beta}\Big)}\ast_t\,e^{-2\sigma_E^2t^2}\right)\quad \text{ with }\quad \|f_-\|_{L_1}=\frac{1}{\sqrt{2\pi}}\,,\\
&f_+(t):=\lim_{\eta\to 0^+}1(|t|\ge\eta)\frac{e^{-2\sigma_E^2t^2-i\beta\sigma_E^2t}}{2\sqrt{2\pi}t\big(\frac{2t}{\beta}+i\big)}+\sqrt{\frac{\pi}{8}}\delta(t)\,.
\end{align*}
Finally, we can express the Dirichlet form associated to the above generator in the following symmetric form.

\begin{lemma}\label{symmetrizationDirichlet}
For any two operators $X,Y\in\mathcal{O}_V$,
\begin{align*}
\mathcal{E}^{(\sigma_E,\beta)}(X,Y)&:=-\langle X,\mathcal{L}^{(\sigma_E,\beta)\dagger}(Y)\rangle_{\sigma_\beta}=\sum_{a,\nu_1,\nu_2}\overline{\alpha}_{\nu_1,\nu_2}\,\langle \partial^{a}_{\nu_2}(X),\partial^a_{\nu_1}(Y)\rangle_{\sigma_\beta}=\,,
\end{align*}
where $\partial^a_\nu:=[A^a_\nu,(.)\,]$ and $$\overline{\alpha}_{\nu_1,\nu_2}:=\frac{h_{\nu_1,\nu_2}}{e^{\frac{(\nu_1-\nu_2)\beta}{4}}+e^{\frac{(\nu_2-\nu_1)\beta}{4}}}\,.$$
\end{lemma}
\begin{proof}
The result follows from simple algebra: using $\sigma_\beta A^a_\nu \sigma_\beta^{-1}=e^{-\beta\nu}A^a_\nu$,
\begin{align*}
&\sum_{a,\nu_1,\nu_2}\overline{\alpha}_{\nu_1,\nu_2}\,\langle \partial^{a}_{\nu_2}(X),\partial^a_{\nu_1}(Y)\rangle_{\sigma_\beta}\\
&\qquad=\sum_{a,\nu_1,\nu_2}\overline{\alpha}_{\nu_1,\nu_2}\tr{\sigma_\beta^{\frac{1}{2}}[A^a_{\nu_2},X]^\dagger \sigma_\beta^{\frac{1}{2}}[A^a_{\nu_1},Y]}\\
&\qquad =\sum_{a,\nu_1,\nu_2}\overline{\alpha}_{\nu_1,\nu_2}\left(
\tr{\sigma_\beta^{\frac{1}{2}}\big[X^\dagger A^{a\dagger}_{\nu_2}-A^{a\dagger}_{\nu_2}X^\dagger\big] \sigma_\beta^{\frac{1}{2}}[A^a_{\nu_1},Y]}
\right)\\
&\qquad =\sum_{a,\nu_1,\nu_2}\overline{\alpha}_{\nu_1,\nu_2}\left(
\tr{\sigma_\beta^{\frac{1}{2}}X^\dagger \sigma_\beta^{\frac{1}{2}}A^{a\dagger}_{\nu_2} [A^a_{\nu_1},Y]}e^{-\frac{\beta\nu_2}{2}}  +e^{\frac{\beta\nu_2}{2}}\tr{ \sigma_\beta^{\frac{1}{2}}X^\dagger \sigma_\beta^{\frac{1}{2}}[Y,A^a_{\nu_1}]A^{a\dagger}_{\nu_2}}
\right)\\
&\qquad =\langle X,\mathcal{K}(Y)\rangle_{\sigma_\beta}\,,
\end{align*}
where 
\begin{align*}
\mathcal{K}(X):=\sum_{a,\nu_1,\nu_2}\overline{\alpha}_{\nu_1,\nu_2} \left(A^{a\dagger}_{\nu_2}A^a_{\nu_1}Xe^{-\frac{\beta\nu_2}{2}}+e^{\frac{\beta\nu_2}{2}}XA^a_{\nu_1}A^{a\dagger}_{\nu_2}-e^{-\frac{\beta\nu_2}{2}}A^{a\dagger}_{\nu_2}XA^a_{\nu_1}-e^{\frac{\beta\nu_2}{2}}A^a_{\nu_1}XA^{a\dagger}_{\nu_2}        \right)\,.
\end{align*}
Now, we have $\overline{\alpha}_{\nu_1,\nu_2}e^{-\frac{\beta\nu_2}{2}}=\alpha_{\nu_1,\nu_2}\frac{e^{\beta(\nu_1-\nu_2)/4}}{e^{(\nu_1-\nu_2)\beta/4}+e^{(\nu_2-\nu_1)\beta/4}}$ and $\overline{\alpha}_{\nu_1,\nu_2}e^{\frac{\beta\nu_2}{2}}=\alpha_{\nu_1,\nu_2}\frac{e^{\beta(\nu_1+3\nu_2)/4}}{e^{(\nu_1-\nu_2)\beta/4}+e^{(\nu_2-\nu_1)\beta/4}}$, therefore,
\begin{align*}
\sum_{a,\nu_1,\nu_2}\overline{\alpha}_{\nu_1,\nu_2}\,e^{\frac{\beta\nu_2}{2}}A^a_{\nu_1}XA^{a\dagger}_{\nu_2}&=\sum_{a,\nu_1,\nu_2}\alpha_{\nu_1,\nu_2}\frac{e^{\beta(\nu_1+3\nu_2)/4}}{e^{(\nu_1-\nu_2)\beta/4}+e^{(\nu_2-\nu_1)\beta/4}}A^a_{\nu_1}XA^{a\dagger}_{\nu_2}\\
&=\sum_{a,\nu_1,\nu_2} \alpha_{-\nu_2,-\nu_1}\frac{e^{\beta(-\nu_2-3\nu_1)/4}}{e^{(\nu_1-\nu_2)\beta/4}+e^{(\nu_2-\nu_1)\beta/4}}A^{a\dagger}_{\nu_2}XA^{a}_{\nu_1}\\
&=\sum_{a,\nu_1,\nu_2} \alpha_{\nu_1,\nu_2}\frac{e^{\beta(\nu_2-\nu_1)/4}}{e^{(\nu_1-\nu_2)\beta/4}+e^{(\nu_2-\nu_1)\beta/4}}A^{a\dagger}_{\nu_2}XA^{a}_{\nu_1}
\end{align*}
where the second identity follows by the change of variables $\nu_1\leftarrow -\nu_2$, $\nu_2\leftarrow -\nu_1$, while the last identity arises from $\alpha_{-\nu_2,-\nu_1}=\alpha_{\nu_1,\nu_2}e^{\beta\frac{\beta(\nu_1+\nu_2)}{2}}$. Hence,
\begin{align*}
\sum_{a,\nu_1,\nu_2}\overline{\alpha}_{\nu_1,\nu_2} \left(e^{-\frac{\beta\nu_2}{2}}A^{a\dagger}_{\nu_2}XA^a_{\nu_1}+e^{\frac{\beta\nu_2}{2}}A^a_{\nu_1}XA^{a\dagger}_{\nu_2}        \right)=\sum_{a,\nu_1,\nu_2}\,\alpha_{\nu_1,\nu_2}\,A^{a\dagger}_{\nu_2}XA^a_{\nu_1}\,
\end{align*}
and $\mathcal{K}=-\mathcal{L}^{(\sigma_E,\beta)\dagger}$. The result follows.
\end{proof}
In the limit $\sigma_E\to 0$, the above simplifies to
\begin{align*}
\mathcal{E}^{(0,\beta)}(X,Y)=\frac{1}{4}\sum_{a,\nu}e^{-\frac{\beta|\nu|}{2}}\langle \partial^a_\nu(X),\partial^a_\nu(Y)\rangle_{\sigma_\beta}\,.
\end{align*}




In short, the proof of Theorem II.5 relies on a perturbative analysis of the generator $\mathcal{L}^{(\sigma_E,\beta)\dagger}$ around $\frac{1}{\beta}=0$, $\sigma_E=0$, as well as corresponding to certain perturbations of $H$. Interestingly, a similar argument was already considered in the classical literature (see e.g. \cite{caputo2011zero}). In the next subsection, we gather the relevant perturbation bounds, which we believe to be of independent interest.

\subsection{Perturbation analysis of Metropolis-type generators}

\subsubsection{Perturbation bounds in $\sigma_E$}

We begin with a perturbation bound around $\sigma_E=0$ for fixed $\beta$ and $H$. In what follows, we denote $\Delta_\nu(H):=\min_{\nu_1\ne \nu_2\in B(H)}|\nu_1-\nu_2|$ and $\Delta_E(H):=\min_{E\ne E'\in \operatorname{spec}(H)}|E-E'|$, and $M$ the total number of eigenvalues of $H$. Clearly, $\Delta_\nu(H)\le \Delta_E(H)$.
\begin{lemma}\label{le:perturbsigmaE}
Given the definition in Eq. \eqref{eq:initialL} and Eq. \eqref{eq:initialL0}, as long as $\beta\Delta_\nu(H)\ge \sqrt{3}$ and
\begin{align*}
\delta_{\beta,\sigma_E}:=m\max_{a,\nu}\|A^a_\nu\|_\infty^2M^2\operatorname{gap}(\mathcal{L}^{(0,\beta)\dagger})^{-1}\big(2M^2 e^{-\frac{\beta \Delta_\nu(H)}{8}-\frac{\Delta_\nu(H)^2}{8\sigma_E^2}}\,+ \,(\beta\sigma_E)^2e^{-\frac{\beta\Delta_E}{4}}\big)<1\,,
\end{align*}
then
\begin{align*}
\operatorname{gap}(\mathcal{L}^{(\sigma_E,\beta)\dagger})\ge \operatorname{gap}(\mathcal{L}^{(0,\beta)\dagger})\big(1-\delta_{\beta,\sigma_E}\big)\,.
\end{align*}
\end{lemma}

\begin{proof}
By the variational formulation of the gap in Eq. \eqref{eq:vargap}, and since $\mathcal{L}^{(\sigma_E,\beta)}$ and $\mathcal{L}^{(0,\beta)}$ share the same kernel $\mathbb{C}\sigma_\beta$, it is sufficient to lower bound $\frac{\mathcal{E}^{(\sigma_E,\beta)}(X)}{\|X\|_{\sigma_\beta}^2}$ in terms of $\frac{\mathcal{E}^{(0,\beta)}(X)}{\|X\|_{\sigma_\beta}^2}$ for any observable $X$ with $\tr{\sigma_\beta X}=0$. We first recall the expression for the coefficients $\alpha_{\nu_1,\nu_2}$ written in \eqref{alphametrochen}:

\begin{align}\label{alphametrochen2}
\alpha_{\nu_1,\nu_2}=\int_{\frac{\sigma_E^2\beta}{2}}^\infty \frac{\sqrt{\beta}}{4\sqrt{\pi x}}\, e^{-\frac{\beta(\nu_1+\nu_2+2x)^2}{16 x}}dx\,e^{-\frac{(\nu_1-\nu_2)^2}{8\sigma_E^2}}\equiv G^{(\sigma_E,\beta)}(\nu_1+\nu_2) e^{-\frac{\beta(\nu_1+\nu_2)}{4}-\frac{(\nu_1-\nu_2)^2}{8\sigma_E^2}}\,,
\end{align}
with 
\begin{align*}
G^{(\sigma_E,\beta)}(\nu):=\int_{\frac{\beta\sigma_E^2}{2}}^\infty\, \frac{1}{4}\sqrt{\frac{\beta}{\pi x}}e^{-\frac{\beta\nu^2}{16 x}-\frac{\beta x}{4}}dx\,.
\end{align*}
Hence 
\begin{align*}
h_{\nu_1,\nu_2}=G^{(\sigma_E,\beta)}(\nu_1+\nu_2)\,e^{-\frac{(\nu_1-\nu_2)^2}{8\sigma_E^2}}
\end{align*}
and 
\begin{align*}
\overline{\alpha}_{\nu_1,\nu_2}=\frac{G^{(\sigma_E,\beta)}(\nu_1+\nu_2)}{e^{\frac{\beta(\nu_1-\nu_2)}{4}}+e^{\frac{\beta(\nu_2-\nu_1)}{4}}}\,e^{-\frac{(\nu_1-\nu_2)^2}{8\sigma_E^2}}\,.
\end{align*}
Next, for any $\nu_1\ne \nu_2$
\begin{align*}
\left|\overline{\alpha}_{\nu_1,\nu_2}\right|&\le G^{(\sigma_E,\beta)}(\nu_1+\nu_2)e^{\frac{-\beta\Delta_\nu(H)}{4}} e^{-\frac{\Delta_\nu(H)}{8\sigma_E^2}}\\
&\le e^{\frac{-\beta\Delta_\nu(H)}{4}}e^{-\frac{\Delta_\nu(H)^2}{8\sigma_E^2}}\frac{1}{4\sqrt{\pi}}\int_0^\infty \frac{1}{\sqrt{u}}e^{-\frac{u}{4}}du\equiv \frac{e^{\frac{-\beta\Delta_\nu(H)}{4}}e^{-\frac{\Delta_\nu(H)^2}{8\sigma_E^2}}}{2}\,\,.
\end{align*}
Therefore, denoting by $\mathcal{E}_{\operatorname{diag}}^{(\sigma_E,\beta)}(X)$ the diagonal part of the Dirichlet form, i.e. 
\begin{align*}
\mathcal{E}_{\operatorname{diag}}^{(\sigma_E,\beta)}(X):=\sum_{a,\nu}\overline{\alpha}^{(\sigma_E,\beta)}_{\nu,\nu}\,\langle \partial^a_{\nu}(X),\partial^a_\nu(X)\rangle_{\sigma_\beta}\,,
\end{align*}
we have that
\begin{align*}
\Big|\mathcal{E}^{(\sigma_E,\beta)}(X)-\mathcal{E}_{\operatorname{diag}}^{(\sigma_E,\beta)}(X)\Big|\le 2M^4m e^{-\frac{\beta\Delta_\nu(H)}{8}}e^{-\frac{\Delta_\nu(H)^2}{8\sigma_E^2}}\max_{a,\nu}\|A^a_\nu\|_\infty^2 \|X\|_{\sigma_\beta}^2\equiv \epsilon_\beta\|X\|_{\sigma_\beta}^2\,.
\end{align*}
Next, we aim at controlling $\mathcal{E}_{\operatorname{diag}}^{(\sigma_E,\beta)}(X)$ in terms of $\mathcal{E}^{(0,\beta)}(X)$. For this, we consider
\begin{align*}
\big| \overline{\alpha}^{(\sigma_E,\beta)}_{\nu,\nu}- \overline{\alpha}^{(0,\beta)}_{\nu,\nu}\big|=\frac{1}{2}\Big|G^{(\sigma_E,\beta)}(2\nu)-G^{(0,\beta)}(2\nu)\Big|&\le \frac{\sqrt{\beta}}{8\sqrt{\pi}}\int_0^{\beta\sigma_E^2}\frac{1}{\sqrt{ x}}\, e^{-\frac{\beta\Delta_E^2}{4x}-\frac{\beta x}{4}}dx \\
&\le \frac{\beta^{\frac{3}{2}}\sigma_E^2}{8\sqrt{\pi}}\sup_{x>0} f(x)\,,
\end{align*}
where $f(x):=\frac{1}{\sqrt{x}}e^{-\frac{\beta\Delta_E^2}{4x}-\frac{\beta x}{4}}$. Differentiating the function, 
\begin{align*}
f'(x)=\frac{1}{4x^{\frac{5}{2}}}\left(-2x+{\beta\Delta_E^2}-{\beta x^2}\right)\,e^{-\frac{\beta\Delta_E^2}{4x}-\frac{\beta x}{4}}\,.
\end{align*}
The function $f'$ has manifestly two zeros $x_{\pm}=\frac{-1\pm \sqrt{1+(\beta\Delta_E)^2}}{\beta}$, and only $x_+\ge 0$. Clearly, $x$ corresponds to a maximum of $f$, and therefore
\begin{align*}
\big| \overline{\alpha}^{(\sigma_E,\beta)}_{\nu,\nu}- \overline{\alpha}^{(0,\beta)}_{\nu,\nu}\big|\le \frac{\beta^{\frac{3}{2}}\sigma_E^2}{8\sqrt{\pi}}\,f(x_+)\le \frac{1}{4}\,(\beta\sigma_E)^2\, e^{-\frac{\beta\Delta_E}{4}}
\end{align*}
where we further assumed that $\beta\Delta_E\ge \sqrt{3}$. Hence, 
\begin{align*}
\Big|\mathcal{E}^{(0,\beta)}(X)-\mathcal{E}_{\operatorname{diag}}^{(\sigma_E,\beta)}(X)\Big|\le  m M^2\,\max_{a,\nu}\|A^a_\nu\|_\infty^2 \|X\|^2_{\sigma_\beta}\,(\beta\sigma_E)^2
\,e^{-\frac{\beta\Delta_E}{4}}\equiv \epsilon_\beta' \|X\|^2_{\sigma_\beta}\,.
\end{align*}
Combining the bounds found above, we have that
\begin{align*}
\Big|\mathcal{E}^{(\sigma_E,\beta)}(X)-\mathcal{E}^{(0,\beta)}(X)\Big|\le (\epsilon_\beta+\epsilon_\beta')\,\|X\|_{\sigma_\beta}^2\,.
\end{align*}
Therefore, as long as $\operatorname{gap}(\mathcal{L}^{(0,\beta)\dagger})^{-1}(\epsilon_\beta+\epsilon_\beta')<1$, we have that, for all observables $X$ such that $\tr{\sigma_\beta X}=0$,
\begin{align*}
\|X\|_{\sigma_\beta}^2\le \operatorname{gap}(\mathcal{L}^{(0,\beta)\dagger})^{-1}\mathcal{E}^{(0,\beta)}(X)\le \operatorname{gap}(\mathcal{L}^{(0,\beta)\dagger})^{-1}\mathcal{E}^{(\sigma_E,\beta)}(X)+\operatorname{gap}(\mathcal{L}^{(0,\beta)\dagger})^{-1}(\epsilon_\beta+\epsilon_\beta')\|X\|_{\sigma_\beta}^2\,,
\end{align*}
which implies that 
\begin{align*}
\operatorname{gap}(\mathcal{L}^{(\sigma_E,\beta)\dagger})\ge \operatorname{gap}(\mathcal{L}^{(0,\beta)\dagger})\Big[1-\operatorname{gap}(\mathcal{L}^{(0,\beta)\dagger})^{-1}(\epsilon_\beta+\epsilon_\beta')\Big]\,.
\end{align*}

\end{proof}

\subsubsection{Perturbation bounds in $\beta$}

Next, we prove stability of the spectral gap under perturbations of the generator from infinite to finite $\beta$.

\begin{lemma}\label{le:lemmainftytobeta}
Given the definitions in Eq. \eqref{eq:L0beta} and Eq. \eqref{eq:L0inf}, as long as 
\begin{align*}
\delta_\beta':= m M^2\max_{a,\nu}\|A^a_\nu\|_{\infty}^2\, e^{-\frac{\beta\Delta_E(H)}{2}}<\operatorname{gap}(\mathcal{L}^{(0,\beta)\dagger})\,,
\end{align*}
then
\begin{align*}
\operatorname{gap}(\mathcal{L}^{(0,\beta)\dagger})\ge \operatorname{gap}(\mathcal{L}^{(0,\infty)\dagger})-\delta_{\beta}'\,.
\end{align*}
\end{lemma}

\begin{proof}
By construction, since both $\widetilde{\mathcal{L}}^{(0,\beta)}$ and $\widetilde{\mathcal{L}}^{(0,\infty)}$ have one-dimensional kernels, the result follows from standard perturbation bounds on eigenvalues, see e.g.~\cite[p. 63]{Bhatia1997}.
\end{proof}

\subsubsection{Perturbation bounds in $H$}

Next, we aim at controlling the perturbation in the Lindbladian resulting from a perturbation of its underlying Hamiltonian. For this, it will be more convenient to work in the time representation.

\begin{lemma}\label{perturbedTpart} Considering the definition in Eq. \eqref{eq:bigT}, given $H=H_0+V$, and assuming that $\beta\sigma_E>1$,

\begin{align*}
\norm{\mathcal{T}^{(\sigma_E,\beta)}_{H} - \mathcal{T}^{(\sigma_E,\beta)}_{H_0}}_{2 \rightarrow 2}\le \frac{2m\norm{V}_\infty}{\sigma_E} \,.
\end{align*}
\end{lemma}

\begin{proof}
The result follows easily from the following few lines:
\begin{align*}
&\norm{\mathcal{T}^{(\sigma_E,\beta)}_{H'} - \mathcal{T}^{(\sigma_E,\beta)}_H }_{2 \rightarrow 2} \\ &\le  \sum_a\int_{\mathbb{R}^2}h_-(t_-)h_+(t_+) \norm{ \Delta_{H'}^{-i(t_+-t_-)}A^a\otimes \Delta_{H'}^{i(t_++t_-)} A^a -\Delta_H^{-i(t_+-t_-)}A^a\otimes \Delta_H^{i(t_++t_-)} A^a }_\infty dt_+dt_- 
\\ & \le  4 \norm{V}_\infty \sum_a \norm{A^a}_\infty \int_{\mathbb{R}^2}h_-(t_-)h_+(t_+) \left( \vert t_+ \vert +\vert t_- \vert \right) dt_+dt_- 
\\ & \le   \frac{\norm{V}_\infty}{\sigma_E} \left(\frac{1}{2} + \frac{2}{\pi} \frac{e^{-\frac{\beta^2 \sigma_E^2}{8}}}{\beta \sigma_E} \right)m\,,
\end{align*}
where the last line simply follows from a direct computation of the $L_1$ norms of $h_+$ and $h_-$, respectively.
\end{proof}

Next, we aim at deriving a similar perturbation bound for the operators $N$. Due to the $1/t$ and Dirac $\delta$ function involved in each part of $f_+$, the latter is not integrable and should be understood in the sense of distributions. For this reasons. the authors of \cite{chen2023efficient} proposed a way of approximating the operator $N^{(\sigma_E,\beta)}_H$ in Eq. \eqref{eq:NsigmabetaH} that is more stable to perturbations. Since the approximation was only carried out in the case where $\sigma_E=\frac{1}{\beta}$, we first need to extend their scheme to the general case.

\begin{lemma}\label{lemmaNpart}
Given the definition in Eq. \eqref{eq:NsigmabetaH} and assuming that $\beta>1$ for any $\frac{\beta}{(\beta\sigma_E)^4} > \eta>0$, $\|N^{(\sigma_E,\beta)}_H-N^\eta\|_\infty\le \frac{7\eta m}{2\pi}\norm{H}$, with 
\begin{align*}
N^\eta:=\sum_a \int_{\mathbb{R}^2} f_-(t_-)\left(\hat{f}^\eta_+(t_+)+\sqrt{\frac{\pi}{8}}\delta(t_+)\right)\,e^{-iHt_-} A^{a\dagger}(t_+)A^a(-t_+)e^{iHt_-}\,dt_+dt_-\,,
\end{align*}
with 
\begin{align*}
\hat{f}^\eta_+(t_+):=\frac{e^{-2\sigma_E^2t_+^2-i\beta\sigma_E^2t_+}+1(|t_+|\le \eta)i\big(\frac{2t_+}{\beta}+i\big)}{2\sqrt{2\pi}t_+\big(\frac{2t_+}{\beta}+i\big)}
\end{align*}
such that, as long as $\beta \sigma_E>1$,
\begin{align}
&\|\hat{f}^\eta_+\|_{L_1}\le \ln\Big(\frac{\beta}{\eta(\beta\sigma_E)^2}\Big)+2\ln\Big(2(\beta\sigma_E)^4\Big)+320\,,\\
&\|t\mapsto t f_-(t)\|_{L_1}\le \frac{2}{\sigma_E}\,,\label{bound1}\\
&\|t\mapsto t f_+(t)\|_{L_1}\le  \frac{1}{4\sigma_E}\,,\label{bound2}\\
&\|t\mapsto t \hat{f}^\eta_+(t)\|_{L_1}\le  \frac{1}{4\sigma_E}+\frac{\eta}{\sqrt{2\pi}}\,.\label{bound3}
\end{align}
\end{lemma}

\begin{proof}
The proof consists of a slight modification of the bounds derived in \cite{chen2023efficient}: first, we have that
\begin{align*}
N^{(\sigma_E,\beta)}_H-N^\eta=\frac{1}{2\sqrt{2\pi}}\sum_a\int_{\mathbb{R}} f_-(t_-) e^{-iHt_-}Q_a e^{iHt_-}\,dt_-\,,
\end{align*}
where 
\begin{align*}
Q_a=i\int_{-\eta}^\eta\,\frac{A^{a\dagger}(t)A^a(t)}{t}\,dt\,.
\end{align*}
By developing the numerator in the integrand of $Q_a$ in terms of $\cos(Ht)$ and $\sin(Ht)$, we end up with a sum of $8$ terms, where only one of which only contains a product of $3$ $\cos(Ht)$, and all the others containing at least one $\sin(Ht)$. While the former does not contribute to the overal integral by symmetry, the latter are all at most linear in $t$ as $t\to 0$. Therefore, $\|Q_a\|_\infty\le 14\eta \norm{H}_\infty$. Hence, 
\begin{align*}
\|N^{(\sigma_E,\beta)}_H-N^\eta\|_\infty\le \frac{7\eta \,m}{2\pi} \norm{H}_\infty \,. 
\end{align*}

It remains to show the bound on the $L_1$-norm of $\hat{f}^\eta_+$. By a similar argument to that from \cite{chen2023efficient}, using that $\hat{f}^\eta_+(t)=\hat{f}^\lambda_+(t)-\frac{1}{2\sqrt{2\pi}}\frac{i}{t}\,1(\eta<|t|\le \lambda)$ whenever $\lambda\ge \eta$, we get
\begin{align*}
\|\hat{f}_+^\eta\|_{L_1}&\le \frac{1}{2\sqrt{2\pi}}\left\|t\mapsto 1(\eta<|t|\le \lambda)\frac{i}{t}\right\|_{L_1}+\frac{1}{2\sqrt{2\pi}}\,\left\|t\mapsto \frac{e^{-2\sigma_E^2t^2-i\beta\sigma_E^2t}+1(|t|\le \lambda)i\big(\frac{2t}{\beta}+i\big)}{t\big(\frac{2t}{\beta}+i\big)}\right\|_{L_1}\\
&\le \ln\left(\frac{\lambda}{\eta}\right)+\left\|t\mapsto \frac{e^{-2\sigma_E^2t^2-i\beta\sigma_E^2t}+1(|t|\le \lambda)i\big(\frac{2t}{\beta}+i\big)}{t\big(\frac{2t}{\beta}+i\big)}\right\|_{L_1}.
\end{align*}
Next, we split the $L_1$ norm above into a part $|t|>\lambda$ and $|t|\le  \lambda$. First,
\begin{align*}
\left\|t\mapsto 1(|t|\ge \lambda)\,\frac{e^{-2\sigma_E^2t^2-i\beta\sigma_E^2t}+1(|t|\le \lambda)i\big(\frac{2t}{\beta}+i\big)}{t\big(\frac{2t}{\beta}+i\big)}\right\|_{L_1}&=\int_{|t|\ge \lambda}\frac{e^{-2\sigma_E^2t^2}}{|t|\sqrt{\frac{4t^2}{\beta^2}+1}}dt\\
&=2\int_{t\ge \lambda}\frac{e^{-2\sigma_E^2t^2}}{t\sqrt{\frac{4t^2}{\beta^2}+1}}dt\\
&\le 2\int_{t\ge \lambda}\frac{1}{t\sqrt{\frac{4t^2}{\beta^2}+1}}dt\\
&=2\operatorname{arsinh}\left(\frac{\beta}{2\lambda}\right)\,.
\end{align*}
Next, we consider the integral over $|t|\le \lambda$, we get 
\begin{align*}
&\left\|t\mapsto 1(|t|\le \lambda)\,\frac{e^{-2\sigma_E^2t^2-i\beta\sigma_E^2t}+i\big(\frac{2t}{\beta}+i\big)}{t\big(\frac{2t}{\beta}+i\big)}\right\|_{L_1}\\
&\qquad\qquad\qquad\qquad\qquad\le \left\| t\mapsto \frac{1(|t|\le \lambda)}{1+\frac{4t^2}{\beta^2}} \right\|_{L_2}\,\left\| t\mapsto 1(|t|\le \lambda)\,\left(1+\frac{4t^2}{\beta^2}\right)\frac{e^{-2\sigma_E^2t^2-i\beta\sigma_E^2t}+i\big(\frac{2t}{\beta}+i\big)}{t\big(\frac{2t}{\beta}+i\big)} \right\|_{L_2}\\
&\qquad\qquad\qquad\qquad\qquad=\sqrt{\beta  \operatorname{arctan}\left(\frac{2\lambda}{\beta}\right)}\,\left\| t\mapsto 1(|t|\le \lambda)\,\frac{\Big(\frac{2t}{\beta}-i\Big)e^{-2\sigma_E^2t^2-i\beta\sigma_E^2t}+i\Big(\frac{4t^2}{\beta^2}+1\Big)}{t} \right\|_{L_2}\\
&\qquad\qquad\qquad\qquad\qquad\le \sqrt{2\lambda}\,\left\| t\mapsto 1(|t|\le \lambda)\,\frac{\Big(\frac{2t}{\beta}-i\Big)e^{-2\sigma_E^2t^2-i\beta\sigma_E^2t}+i\Big(\frac{4t^2}{\beta^2}+1\Big)}{t} \right\|_{L_2}
\end{align*}
Next, we further control the $L_2$ norm above as follows
\begin{align*}
&\left\|t\mapsto 1(|t|\le  \lambda)\frac{\Big(\frac{2t}{\beta}-i\Big)e^{-2\sigma_E^2t^2-i\beta\sigma_E^2t}+i\Big(1+\frac{4t^2}{\beta^2}\Big)}{t}\right\|_{L_2}^2
\\
&\qquad\qquad={\int_{|t|\le \lambda}\frac{\left|\Big(\frac{2t}{\beta}-i\Big)e^{-2\sigma_E^2t^2-i\beta\sigma_E^2t}+i\big(1+\frac{4t^2}{\beta^2}\big)\right|^2}{t^2}}dt\\
&\qquad\qquad=\int_{|t|\le \lambda}\frac{\left|\frac{2t}{\beta}e^{-2\sigma_E^2t^2-i\beta\sigma_E^2t} -i\left(-2\sigma_E^2t^2-i\beta\sigma_E^2t\right)\sum_{k=0}^\infty\frac{ \left(-2\sigma_E^2t^2-i\beta\sigma_E^2t\right)^k}{(k+1)!}+i\frac{4t^2}{\beta^2}\right|^2}{t^2}\,dt\\
&\qquad\qquad=\,\int_{|t|\le \lambda}{\left|\frac{2}{\beta}e^{-2\sigma_E^2t^2-i\beta\sigma_E^2t} -i\left(-2\sigma_E^2t-i\beta\sigma_E^2\right)\sum_{k=0}^\infty\frac{ \left(-2\sigma_E^2t^2-i\beta\sigma_E^2t\right)^k}{(k+1)!}+i\frac{4t}{\beta^2}\right|^2}\,dt\\
&\qquad\qquad \le 2\lambda \left(\frac{2}{\beta}+2\sigma_E^2\lambda+\beta\sigma_E^2e^{\sigma_E^2\lambda^2+\beta\sigma_E^2\lambda}+\frac{4\lambda}{\beta^2}\right)^2\,.
\end{align*}
Choosing $\lambda=\frac{\beta}{(\beta\sigma_E)^4}$, we therefore get that, assuming $\beta\sigma_E>1$,
\begin{align*}
\left\|t\mapsto 1(|t|\le \lambda)\,\frac{e^{-2\sigma_E^2t^2-i\beta\sigma_E^2t}+i\big(\frac{2t}{\beta}+i\big)}{t\big(\frac{2t}{\beta}+i\big)}\right\|_{L_1}&\le 8\lambda\,\left(\frac{1}{\beta}+\sigma_E^2\lambda+\beta\sigma_E^2e^{\sigma_E^2\lambda^2+\beta\sigma_E^2\lambda}+\frac{\lambda}{\beta^2}\right)\le {320}\,.
\end{align*}
Combining the above bounds, we get
\begin{align*}
\|\hat{f}^\eta_+\|_{L_1}\le  \ln\left(\frac{\lambda}{\eta}\right)+2\operatorname{arsinh}\left(\frac{\beta}{2\lambda}\right)+320\le \ln\left(\frac{\beta}{\eta(\beta\sigma_E)^2}\right)+2\ln\Big(2(\beta\sigma_E)^4\Big)+320
\end{align*}
where in the last line we reused that $\beta\sigma_E>1$. Next, we prove the norm bounds claimed in \eqref{bound1} and \eqref{bound2}. Starting with $f_-$, we get
\begin{align*}
\left\|t\mapsto tf_-(t)\right\|_{L_1}&=\frac{2\sigma_E}{\pi\beta}\int_{\mathbb{R}^2}|t|\, \frac{1}{\cosh\Big(\frac{2\pi (t-u)}{\beta}\Big)}e^{-2\sigma_E^2 u^2}dudt\\
&= \frac{1}{\pi\beta}\int_{\mathbb{R}^2}\frac{|t|}{\cosh\Big(\frac{2\pi t}{\beta}-\frac{v}{\sigma_E\beta}\Big)}\,e^{-\frac{v^2}{2}}dvdt\\
&\le \frac{2}{\pi\beta}\int_{(v,t)\in\mathbb{R}\times \mathbb{R}_+}{t}e^{-\frac{2\pi t}{\beta}+\frac{v}{\sigma_E\beta}}\,e^{-\frac{v^2}{2}}dvdt+\frac{2}{\pi\beta}\int_{(v,t)\in\mathbb{R}\times \mathbb{R}_+}{t}e^{-\frac{2\pi t}{\beta}-\frac{v}{\sigma_E\beta}}\,e^{-\frac{v^2}{2}}dvdt\\
&=\frac{\beta}{\pi^2}\,\int\,e^{-\frac{v^2}{2}}\left(e^{-\frac{v}{\sigma_E\beta}}+e^{\frac{v}{\sigma_E\beta}}\right) \,dv\\
&=\frac{2\sqrt{2} e^{\frac{1}{2(\sigma_E\beta)^2}}}{\sigma_E\pi\sqrt{\pi}}\,.
\end{align*}
Finally, we derive the bound \eqref{bound2},
\begin{align*}
\|t\mapsto t {f}_+(t)\|_{L_1}=\frac{1}{2\sqrt{2\pi}}\,\left\|t\mapsto \frac{e^{-2\sigma_E^2t^2}}{\big(\frac{2t}{\beta}+i\big)}\right\|_{L_1}\le \frac{1}{4\sigma_E}\,.
\end{align*}
Equation \eqref{bound3} follows directly. 
\end{proof}

Using the bounds of Lemma \ref{lemmaNpart}, we derive a perturbation bound for $N^{(\sigma_E,\beta)}_H$.
\begin{lemma}\label{perturbedNpart}
Given the definition in Eq. \eqref{eq:NsigmabetaH}, assuming that $\beta \sigma_E>1$ and given $H=H_0+H_1$,
\begin{align*}
\|N^{(\sigma_E,\beta)}_H-N^{(\sigma_E,\beta)}_{H_0}\|_\infty\le \frac{6m\|H_1\|_\infty}{\sigma_E}\, \left(649+4\kappa\ln(\beta\sigma_E)\right)+\frac{3m\beta (\norm{H_1}_\infty+2\norm{H_0}_\infty)}{(\sigma_E\beta)^\kappa}\,,
\end{align*}
for any $\kappa\ge6$.

\end{lemma}

\begin{proof}
For $\eta>0$ to be fixed later, and denoting $N^{\eta}_H$, resp. $N^\eta_{H_0}$ the operators approximating $N^{(\sigma_E,\beta)}_H$ and $N^{(\sigma_E,\beta)}_{H_0}$ respectively,  
\begin{align*}
\|N^{(\sigma_E,\beta)}_{H}-N^{(\sigma_E,\beta)}_{H_0}\|_\infty&\le \|N^{(\sigma_E,\beta)}_{H}-N^\eta_{H}\|_\infty+\|N_{H}^\eta-N_{H_0}^\eta\|_\infty+\|N_{H_0}^\eta-N^{(\sigma_E,\beta)}_{H_0}\|_\infty\\
&\le \frac{7\eta m}{\pi} \left( \norm{H}_\infty+\norm{H_0}_\infty\right)+ \|N^\eta_H-N^\eta_{H_0}\|_\infty\,.
\end{align*}
Next, we control the infinite norm on the above right-hand side. By definition we have that
\begin{align*}
N^\eta_H-N^\eta_{H_0}&=\sum_a \int_{\mathbb{R}^2} f_-(t_-)\left(\hat{f}^\eta_+(t_+)+\sqrt{\frac{\pi}{8}}\delta(t_+)\right)\,(F_H(t_+,t_-)-F_{H_0}(t_+,t_-))\,dt_+dt_-\\
&=\sum_a \int_{\mathbb{R}^2} f_-(t_-)\hat{f}^\eta_+(t_+)\,(F_H(t_+,t_-)-F_{H_0}(t_+,t_-))\,dt_+dt_-\\
&\quad +\sqrt{\frac{\pi}{8}}\sum_a\int_{\mathbb{R}}f_-(t_-)(F_H(0,t_-)-F_{H_0}(0,t_-))dt_-\,,
\end{align*}
 where $F_H(t_+,t_-)=e^{-iHt_-} A_H^{a\dagger}(t_+)A_H^a(-t_+)e^{iHt_-}$, with $A^a_H(t):=e^{itH}A^a e^{-itH}$. Now, since $\|(\Delta_H^{it}-\Delta_{H_0}^{it})(A^{a,\alpha})\|_\infty\le 2|t|\,\|H_1 \|_\infty$, we have that $\|F_H(t_+,t_-)-F_{H_0}(t_+,t_-)\|_\infty\le 6(|t_+|+|t_-|)\|H_1 \|_\infty$. Therefore
 \begin{align*}
\|N^\eta_H-N^\eta_{H_0}\|_\infty&\le 6m\|H_1 \|_\infty\,\big(\|f_-\|_{L_1}\|t\mapsto t\hat{f}^\eta_+(t)\|_{L_1}+\|t\mapsto tf_-(t)\|_{L_1}(1+\|\hat{f}^\eta_+\|_{L_1})\big)\\
&\le \frac{6m\|H_1 \|_\infty}{\sigma_E}\, \left(649+4\kappa\ln(\beta\sigma_E)\right)
 \end{align*}
where the last bound follows from those derived in Lemma \ref{lemmaNpart}, and we chose $\eta=\frac{\beta}{(\beta\sigma_E)^\kappa}$ for $\kappa\ge 6$. The result follows.

 \end{proof}

With the bounds derived above, we are ready to state our perturbation bounds for $H\mapsto \widetilde{\mathcal{L}}^{(\sigma_E,\beta)}_H$:

\begin{lemma}\label{propperturbHH0}
Given the definitions in Eq.\eqref{eqLbetafourier} and the perturbed Hamiltonian $H=H_0+H_1$, assuming that $\beta \sigma_E>1$, 
\begin{align*}
\|\widetilde{\mathcal{L}}^{(\sigma_E,\beta)}_H-\widetilde{\mathcal{L}}^{(\sigma_E,\beta)}_{H_0}\|_{2\to 2}\le  \frac{6m\|H_1 \|_\infty}{\sigma_E}\, \left(650+4\kappa\ln(\beta\sigma_E)\right)+\frac{3m\beta  \left( \norm{H_1}_\infty+2\norm{H_0}_\infty\right)}{(\sigma_E\beta)^\kappa}\equiv \delta''_{\beta,\sigma_E}\,,
\end{align*}
for any $\kappa\ge 6$. Consequently, as long as $\delta''_{\beta,\sigma_E}<\operatorname{gap}(\widetilde{\mathcal{L}}^{(\sigma_E,\beta)}_{H_0})$,
\begin{align*}
\operatorname{gap}({\mathcal{L}}^{(\sigma_E,\beta)\dagger}_{H})>\operatorname{gap}({\mathcal{L}}^{(\sigma_E,\beta)\dagger}_{H_0})-\delta''_{\beta,\sigma_E}\,.
\end{align*}

\end{lemma}
\begin{proof}
Follows directly from Lemmas \ref{perturbedNpart} and \ref{perturbedTpart}.
\end{proof}

\subsection{Gap for generators of classical Hamiltonians}

The next main ingredient of our proof is a lower bound on the gap of $\widetilde{\mathcal{L}}^{(0,\infty)}$ corresponding to a classical local Hamiltonian. From now on, we fix the set of jumps to be 
\begin{align*}
 \{A^{a,\alpha}\}_{\alpha\in[4],a\in[n]}:=\{X_a\}\cup \{Y_a\}\cup \{Z_a\}\cup \{I_a\}\,.
\end{align*}
We recall that 
\begin{align}\label{Linftyexpression}
\widetilde{\mathcal{L}}^{(0,\infty)}(X)=\frac{1}{2}\sum_{a}\left(A^a_0 XA^{a}_0-\frac{1}{2}\{(A^{a}_0)^2,X\}\right)-\frac{1}{4}\sum_{a,\nu<0}\{A^{a\dagger}_\nu A^{a}_\nu,X\}\,.
\end{align}

\begin{proposition}\label{classicalH}
Assume that $H:=\sum_{x}E(x)|x\rangle\langle x|$ is diagonal in the $n$-qubit computational basis with ground state energy $E_0=0$. For any computational basis product states $|x\rangle, |y\rangle,|u\rangle,|v\rangle$, we denote $U(x,y;u,v):=\langle u|\widetilde{\mathcal{L}}^{(0,\infty)}(|x\rangle\langle y|)|v\rangle$. Then, $U(x,y;u,v)=0$ whenever $E(u)\ne E(x) \text{ or } E(v)\ne E(y)$. Else, 
\begin{align*}
U(x,y;u,v):=\left\{ \begin{aligned}
&0;\qquad\exists a\in[n],x_{a}\ne y_{a},\,u=x^{(a)},\,v=y^{(a)}\text{ or }|x-u|\ne |y-v|\text{ or } |x-u|>1 \\
&1;\qquad\qquad\qquad\qquad\qquad\quad\qquad\qquad\quad \, \,\,  \exists a\in[n],x_{a}=y_{a},\,u=x^{(a)},\,v=y^{(a)}\\
&\big|\big\{a\in\Lambda|  x_a=y_a\big\}\big|- n-\frac{1}{2}\big|\big\{ a|E(x^{(a)})=E(x)\big\}\big|-\frac{1}{2}\big|\big\{ a|E(y^{(a)})= E(y)\big\}\big|\\
& -\frac{1}{2}\Big(\big|\big\{a|E(x)>E(x^{(a)})\big\}\big|+\big|\big\{a|E(y)>E(y^{(a)})\big\}\big|\Big);\qquad \qquad x=u,\,y=v\,.
\end{aligned}\right.
\end{align*}
Above, we denoted by $|x-y|$ the Hamming distance between bitstrings $x$ and $y$, and by $x^{(a)}$ the bitstring $x$ with bit flipped at site $a$. Moreover, the spectral gap of $\widetilde{\mathcal{L}}_\infty$ can be lower bounded as follows:
\begin{align*}
\operatorname{gap}(\widetilde{\mathcal{L}}^{(0,\infty)})\ge \min\big\{\operatorname{gap}\big(\widetilde{\mathcal{L}}^{(0,\infty)}_{\operatorname{diag}}\big),1\big\}\,,
\end{align*}
where $\widetilde{\mathcal{L}}^{(0,\infty)}_{\operatorname{diag}}$ corresponds to the restriction of $\widetilde{\mathcal{L}}^{(0,\infty)}$ to the diagonal $\operatorname{span}\big(\{|x\rangle \langle x|\}_{x\in \{0,1\}^n}\big)$.
\end{proposition}
\begin{proof}
The expressions derived for $U$ follow by direct computation starting from \eqref{Linftyexpression} and using the fact that the eigenprojections of $H$ are sums over pure product states, as well as the twirling relation $\sum_{P\in \{I,X,Y,Z\}}PAP=2\tr{A}I$. Indeed,
\begin{align*}
U(x,y;&u,v)= \frac{\delta_{E(u),E(x)}\delta_{E(v),E(y)}}{2}\\
&\quad \quad \times\sum_{a,\alpha}\!\left\{\!\langle u|A^{a,\alpha}|x\rangle\langle y|A^{a,\alpha}|v\rangle-\frac{\delta_{y,v}}{2}\langle u|A^{a,\alpha}P_{E(x)}A^{a,\alpha}|x\rangle-\frac{\delta_{u,x}}{2}\langle y|A^{a,\alpha}P_{E(y)}A^{a,\alpha}|v\rangle\!\right\}\\
&  -\frac{\delta_{E(u),E(x)}\delta_{E(v),E(y)}}{4}\,\sum_{a,\alpha,i}\delta_{y,v}\delta_{E(x)>E_i}\,\langle u|A^{a,\alpha}P_{E_i}A^{a,\alpha}|x\rangle+\delta_{u,x}\delta_{E(y)>E_i}\langle y|A^{a,\alpha}P_{E_i}A^{a,\alpha}|v\rangle
\end{align*}
where $P_E$ denotes the eigenprojection onto eigenvalue $E$ of $H$. Assuming for now that the condition $E(u)=E(x)$ and $E(v)=E(y)$ holds, we have
\begin{align*}
&U(x,y;u,v)\\
& =\sum_{a}\delta_{u_{a^c},x_{a^c}}\delta_{y_{a^c},v_{a^c}}\delta_{x_a,y_a}\delta_{u_a,v_a}-\frac{\delta_{y,v}\delta_{u_a,x_a}}{2}\tr{P_{E(x)} |x_{a^c}\rangle\langle u_{a^c}|}-\frac{\delta_{u,x}\delta_{y_a,v_a}}{2}\tr{P_{E(y)} |v_{a^c}\rangle\langle y_{a^c}|}\\
&\quad  -\frac{1}{2}\,\sum_{a,i}\,\delta_{y,v} \delta_{E(x)>E_i}\,  \delta_{u_a,x_a}\tr{P_{E_i} |x_{a^c}\rangle\langle u_{a^c}|}+\delta_{u,x}\delta_{y_a,v_a}\delta_{E(y)>E_i}\tr{P_{E_i} |v_{a^c}\rangle\langle y_{a^c}|}\\
& =\sum_a \delta_{u_{a^c},x_{a^c}}\delta_{v_{a^c},y_{a^c}}\left\{\delta_{x_a,y_a}\delta_{u_a,v_a}- \delta_{y_a,v_a}\delta_{u_a,x_a}\Big(\delta_{E(x),E(x^{(a)})}+\delta_{E(y),E(y^{(a)})}\right. \\
&\qquad\qquad\qquad\qquad\qquad\qquad\qquad \left.+\frac{\delta_{E(x)\ne E(x^{(a)})}+\delta_{E(y)\ne E(y^{(a)})}}{2}+\frac{ \delta_{E(x)>E(x^{(a)})}+\delta_{E(y)>E(y^{(a)})}}{2}\Big)\right\}\\
& =\sum_a \delta_{u_{a^c},x_{a^c}}\delta_{v_{a^c},y_{a^c}}\Big\{\delta_{x_a,y_a}\delta_{u_a,v_a}- \delta_{y_a,v_a}\delta_{u_a,x_a}\times \\
&\qquad\qquad\qquad\qquad\qquad\qquad\qquad \big(1+\frac{\delta_{E(x),E(x^{(a)})}+\delta_{E(y),E(y^{(a)})}}{2}+\frac{ \delta_{E(x)>E(x^{(a)})}+\delta_{E(y)>E(y^{(a)})}}{2}\big)\Big\}\,,
\end{align*}
where we denoted e.g.~$|x_{a^c}\rangle\langle x_{a^c}|$ in place of $I_a\otimes |x_{a^c}\rangle\langle x_{a^c}|$ for simplicity. From this one can easily read out the different cases claimed in the statement of the proposition. Now, the generator $\widetilde{\mathcal{L}}^{(0,\infty)}$ clearly decomposes into the direct sum $\widetilde{\mathcal{L}}^{(0,\infty)}=\widetilde{\mathcal{L}}_{\operatorname{diag}}^{(0,\infty)}\oplus \widetilde{\mathcal{L}}_{\operatorname{off-diag}}^{(0,\infty)}$, where $\widetilde{\mathcal{L}}_{\operatorname{off-diag}}^{(0,\infty)}$ corresponds to the restriction of $\widetilde{\mathcal{L}}^{(0,\infty)}$ to off-diagonal elements $\{|x\rangle\langle y|\}_{x\ne y}$, since $\langle x|\widetilde{\mathcal{L}}^{(0,\infty)}(|y\rangle\langle z|)|x\rangle=\langle y|\widetilde{\mathcal{L}}^{(0,\infty)}(|x\rangle\langle x|)|z\rangle=0$ for any $x,y,z$ with $y\ne z$. Therefore, the claim about the gap would follow directly from showing that $\widetilde{\mathcal{L}}_{\operatorname{off-diag}}^{(0,\infty)}$ is bounded above by $-1$. Now, it is easy to see that for any $x\ne y$,
\begin{align}\label{eqgersh}
\sum_{(u,v)\ne (x,y)}U(x,y;u,v)&=\big|\big\{a|x_a=y_a,\,E(x)=E(x^{(a)}),E(y)=E(y^{(a)})\big\}\big|\nonumber \\
&\le \frac{1}{2}\big|\big\{a|E(x^{(a)})=E(x)\big\}\big|+\frac{1}{2}\big|\big\{a|E(y^{(a)})=E(y)\big\}\big|\nonumber\\
&\le n-1-|\{a|x_a=y_a\}|+\frac{1}{2}\big|\big\{a|E(x^{(a)})=E(x)\big\}\big|+\frac{1}{2}\big|\big\{a|E(y^{(a)})=E(y)\big\}\big|\nonumber\\
&\le -U(x,y;x,y)-1\,.
\end{align}
Therefore, by Gershgorin circle theorem, $\operatorname{spec}(\widetilde{\mathcal{L}}_{\operatorname{off-diag}}^{(0,\infty)})\subseteq (-\infty,-1]$. The result follows.
\end{proof}

\subsection{Kitaev circuit-to-Hamiltonian mappings and universality}\label{app:kitaev}

We now show how to use the results of the last Section to conclude that the class of dissipative evolutions we considered so far, when enlarged with the ability to perform arbitrary few-qubit measurements, provides a model of quantum computation that is polynomially equivalent to the circuit model for $\beta=\Omega(\operatorname{poly}(n))$. To achieve this, we resort to the circuit-to-Hamiltonian mapping of \cite{chen2023local}, similar to the ones  in~\cite{aharonov2008adiabatic,kitaev2002classical}. Consider an $n$-qubit circuit $C=U_T\dots U_1$ made of $T$ geometrically local unitary gates $U_1,\dots ,U_T$, each operating on one or two neighbouring qubits on a $2D$ lattice. For simplicity and without loss of generality, we will assume that $T$ is divisible by $2$ in what follows. For the clock construction of~\cite{chen2023local}, it will be important that for some $L=\textrm{poly}(n)$, the first and last $t_0=cL^2$ gates are identity and only  the intermediary $L$ gates act nontrivially. Clearly, we can map any polynomial-sized circuit into such a circuit with only polynomial overheads.
Next, we define the following Hamiltonian on $n+T$ qubits

\begin{align}\label{equ:definition_history_Hamiltonian}
H_{C}:=H_{\operatorname{clock}}+ H_{\operatorname{in}}+ H_{\operatorname{prop}}\,,
\end{align}
where 
\begin{align*}
&H_{\operatorname{clock}}:=J_{\operatorname{clock}}\sum_{t=1}^{T-1}f_t\,I\otimes |01\rangle\langle 01|_{t,t+1}\equiv H_0\\
&H_{\operatorname{in}}:=J_{\operatorname{in}}\sum_{j=1}^ng_j\,|1\rangle\langle 1|_j\otimes |10\rangle\langle 10|_{t_j-1,t_j}\\
&H_{\operatorname{prop}}:=\frac{1}{2}J_{\operatorname{prop}}\,\sum_{t=1}^T H_{\operatorname{prop}}(t)\,.
\end{align*}
The couplings $J_{\operatorname{clock}}$, $J_{\operatorname{in}}$ and $J_{\operatorname{prop}}$ are to be fixed later, and we have that
\begin{align*}
&H_{\operatorname{prop}}(t):=I-h_t(U_t\otimes \ketbra{110}{100}_{t-1,t,t+1}+U_t^\dagger\otimes |100\rangle\langle 110|_{t-1,t,t+1})\,, 
\end{align*}
for $1< t< T$ and
\begin{align*}
&H_{\operatorname{prop}}(1):=I-h_1(U_1\otimes |10\rangle\langle 00|_{1,2}+U_1^\dagger\otimes |00\rangle\langle 10|_{1,2}),\\
&H_{\operatorname{prop}}(T):=I-h_T(U_T\otimes |11\rangle\langle 10|_{T-1,T}+U_T^\dagger \otimes |10\rangle\langle 11|_{T-1,T})\,.
\end{align*}
Above, each of the times $t_j$ in the definition of $H_{\operatorname{in}}$ corresponds to the first time qubit $j$ is acted on. We also set
\begin{align*}
J_{\operatorname{clock}}=1,\qquad f_t=(T-t)/T,\qquad g_j=1/\xi_{t_j-1}:=2^T\binom{T}{t_j-1}^{-1},\qquad h_t=\sqrt{t(T-t+1)}\,.
\end{align*}
The construction of the circuit in \cite{chen2023local}, particularly the fact that only intermediary gates act nontrivially, ensures that $g_j=\mathcal{O}(T)$ for all $j$. We recall a few important properties of this Hamiltonian, as proven in \cite{chen2023local}.
\begin{proposition}\label{prop:properties_hc}
The Hamiltonian $H_{C}$ defined in Equation \eqref{equ:definition_history_Hamiltonian} has a unique ground state of energy $0$ given by 
\begin{align*}
|\eta_0\rangle =\sum_{t=0}^T\,\sqrt{\xi_t}\,(U_t\dots U_1)|0^n\rangle \otimes |1^t0^{T-t}\rangle\,, \text{ where }\,\xi_t:=\frac{1}{2^T}\binom{T}{t}\,.
\end{align*}
$H_{\operatorname{clock}}$ has a spectral gap of $1/T$ and $\Delta_E(H_{\operatorname{clock}})\ge \Delta_\nu(H_{\operatorname{clock}})\ge 1/T$. Moreover, the orthogonal complement $H_C|_{\operatorname{Im}(H_C)}$  can be block decomposed as follows: denoting by $P'$ the projection onto the kernel of $H_{\operatorname{clock}}+H_{\operatorname{prop}}$, and since $P'|\eta_0\rangle=|\eta_0\rangle$, we can decompose $\operatorname{Im}(H_C)$ into the direct sum $\operatorname{Im}(H_C):=\cH_1\oplus\cH_2$, where $\cH_1:=(P'-|\eta_0\rangle\langle \eta_0|)\operatorname{Im}(H_C)$ and $\cH_2=(I-P')\operatorname{Im}(H_C)$, such that
\begin{align*}
H_C|_{\operatorname{Im}(H_C)}=\begin{pmatrix}
A&B\\
B^\dagger &C
\end{pmatrix}
\end{align*}
with $A\ge J_{\operatorname{in}}I$, $C\ge (J_{\operatorname{prop}}-2\|H_{\operatorname{in}}\|_\infty)I$, and $\|B\|_\infty\le \|H_{\operatorname{in}}\|_\infty$.
\end{proposition}

From the above proposition, one can easily deduce the following properties of $H_C$:

\begin{corollary}\label{overlapgibbsground}
Assume that $\frac{3}{2}J_{\operatorname{in}}\ge \|H_{\operatorname{in}}\|_\infty^2(J_{\operatorname{prop}}-2\|H_{\operatorname{in}}\|_\infty)^{-1}$. Then, the Hamiltonian $H_C$ has a spectral gap $\operatorname{gap}(H_C)\ge \min\{J_{\operatorname{in}}/2,J_{\operatorname{prop}}-2\|H_{\operatorname{in}}\|_\infty\}$. Moreover, the Gibbs state $\sigma_{\beta,H_C}$ at inverse temperature $\beta$ satisfies $\|\sigma_{\beta}-|\eta_0\rangle\langle \eta_0|\|_1\le \epsilon$ for $\beta\ge (\min\{J_{\operatorname{in}}/2,J_{\operatorname{prop}}-2\|H_{\operatorname{in}}\|_\infty\})^{-1}\, \ln(2^{n+T}/\epsilon)$.
\end{corollary}

\begin{proof}
We start by proving the claimed lower bound on the spectral gap. For this, we use the standard observation that $H_C|_{\operatorname{Im}(H_C)}$ is related to the diagonal matrix 
\begin{align*}
\begin{pmatrix}
A-BC^{-1}B^\dagger&0\\0& C
\end{pmatrix}
\end{align*}
up to a similarity transformation. Therefore, $\operatorname{gap}(H_C)\ge \min\{\lambda_{\min}(A-BC^{-1}B^\dagger),\lambda_{\min}(C)\}$.  Moreover, by Proposition \ref{prop:properties_hc}, we have that $\lambda_{\min}(C)\ge (J_{\operatorname{prop}}-2\|H_{\operatorname{in}}\|_\infty)I$.
\begin{align*}
A-BC^{-1}B^\dagger\ge \big(J_{\operatorname{in}}-\|H_{\operatorname{in}}\|_\infty^2(J_{\operatorname{prop}}-2\|H_{\operatorname{in}}\|_\infty)^{-1}\big)I\,.
\end{align*}
For the latter lower bound to be larger than $J_{\operatorname{in}}/2$, we need to impose that
\begin{align*}
\frac{3}{2}J_{\operatorname{in}}\ge \|H_{\operatorname{in}}\|_\infty^2(J_{\operatorname{prop}}-2\|H_{\operatorname{in}}\|_\infty)^{-1}\,.
\end{align*}
The bound on the spectral gap of $H_C$ follows. Next, denoting by $E_i$ the eigenvalues of $H_C$ in increasing order, with $E_0=0$, we can control the trace distance as
\begin{align*}
\|\sigma_{\beta,H_C}-|\eta_0\rangle\langle \eta_0|\|_1=1-\frac{1}{\tr{e^{-\beta H_C}}}=1-\frac{1}{1+\sum_{i\ge 1}e^{-\beta E_i}}&\le 1-\frac{1}{1+(2^{n+T}-1)e^{-\beta \operatorname{gap}(H_C)}}\\
&\le 2^{n+T}e^{-\beta\operatorname{gap}(H_C)}\,.
\end{align*}
The result follows. 
\end{proof}
Next, we provide a lower bound on the spectral gap of $\widetilde{\mathcal{L}}^{(0,\infty)}_{H_{\operatorname{clock}}}$. The proof is inspired by \cite{caputo2011zero}.

\begin{proposition}\label{propgap0inftyclock}
Given the definition in Eq. \eqref{eq:L0inf},
\begin{align*}
\operatorname{gap}\big(\widetilde{\mathcal{L}}^{(0,\infty)}_{{H}_{\operatorname{clock}}}\big)\ge \min\left\{2\Big(1-\cos\Big(\frac{\pi}{T+1}\Big)\Big),1\right\}\,.
\end{align*}
\end{proposition}
\begin{proof}

By Proposition \ref{classicalH}, we have that $\operatorname{gap}\big(\widetilde{\mathcal{L}}^{(0,\infty)}_{{H}_{\operatorname{clock}}}\big)=\min\big\{\operatorname{gap}\big(\widetilde{\mathcal{L}}^{(0,\infty)}_{{H}_{\operatorname{clock}},\operatorname{diag}}\big),1\big\}$, where $\widetilde{\mathcal{L}}^{(0,\infty)}_{{H}_{\operatorname{clock}},\operatorname{diag}}$ corresponds to the restriction of $\widetilde{\mathcal{L}}^{(0,\infty)}_{{H}_{\operatorname{clock}}}$ to the diagonal $\operatorname{span}\big(\{|x\rangle \langle x|\}_{x}\big)$.
Thus, it only remains to control the spectral gap of $\widetilde{\mathcal{L}}^{(0,\infty)}_{{H}_{\operatorname{clock}},\operatorname{diag}}$, or equivalently that of the matrix $\{U_{\operatorname{diag}}(x;u):=U(x,x;u,u)\}_{x,u}$ as defined in Proposition \ref{classicalH} in order to conclude. The coefficients of $U_{\operatorname{diag}}$ cancel whenever $E(u)\ne E(x)$. Else,

\begin{align*}
U_{\operatorname{diag}}(x;u):=\left\{ \begin{aligned}
&0\qquad\,\,\, \qquad\qquad\qquad\qquad\qquad\qquad\qquad\qquad\qquad\qquad\,\,\,\,\quad\,\qquad|x-u|>1 \\
&1\qquad \,\,\qquad\qquad\qquad\qquad\qquad\qquad\qquad\qquad\qquad\qquad\quad \, \,\,  \exists a\in[n],\,u=x^{(a)}\\
&-\big|\big\{ a|E(x^{(a)})\le E(x)\big\}\big|\qquad\qquad \qquad\qquad\qquad\,\quad  \qquad \qquad \, \, \, \,\,\,\qquad x=u\,.
\end{aligned}\right.
\end{align*}
In words, the matrix $U_{\operatorname{diag}}$ decomposes into a direct sum $U_{\operatorname{diag}}=\bigoplus_{E\in\operatorname{spec}({H}_{\operatorname{clock}})}U_E$ of blocks associated to each energy $E=E(u)=E(x)$. For the block corresponding to the lowest energy $E_0=0$, 
\begin{align*}
U_0(x;u):=\left\{ \begin{aligned}
&0\qquad\,\,\, \qquad\qquad\qquad\qquad\qquad\qquad\qquad\qquad\qquad\qquad\,\,\,\,\quad\,\qquad|x-u|>1 \\
&1\qquad \,\,\qquad\qquad\qquad\qquad\qquad\qquad\qquad\qquad\qquad\qquad\quad \, \,\,  \exists a\in[n],\,u=x^{(a)}\\
&-\big|\big\{ a|E(x^{(a)})=0\big\}\big|\qquad\qquad\qquad\qquad\qquad\qquad
\quad \qquad\quad\quad \, \, \, \, \,\qquad\, x=u\,.
\end{aligned}\right.
\end{align*}
This turns out to be the generator of a graph Laplacian corresponding to the graph whose vertices are the bitstrings of energy $E_0=0$, and whose edges connect vertices of Hamming distance $1$ from each other. In the specific case where $H={H}_{\operatorname{clock}}$, the ground space is spanned by the $T+1$ vectors $|0\dots 0\rangle,\,|1,0,\dots,0\rangle,\,\dots ,\, |1,\dots ,1\rangle$. In this ordered basis, $U_0$ takes the form of the $(T+1)\times (T+1)$ matrix
\begin{align*}
U_0:=\begin{pmatrix}
-1&1&0&0&\dots &\dots &0\\
1&-2&1&0&\dots&\dots &0\\
0&1&-2&1&0&\dots &0\\
\vdots&\ddots&\ddots &\ddots&1 &-2 &1\\
\vdots&\ddots&\ddots &\ddots&0 &1 &-1
\end{pmatrix}\,.
\end{align*}
The spectrum of this matrix was computed in \cite{parlangeli2011reachability}. In particular, it is easy to see that the uniform distribution $(T+1)^{-1}(1,\dots,1)$ is the unique invariant measure of $U_0$. Moreover, the gap is equal to $2(1-\cos\big(\frac{\pi}{T+1}\big))$. Next, we consider the higher blocks $U_E$, $E>0$. It is easy to see that for any bitstring $x$ of energy $E$, and any bit $a\in[T]$, the energy of $x^{(a)}$ is always different from that of $x$ by at least $1/T$. Therefore, $U_E$ is diagonal. Moreover, one can always find a bit $a$ such that $E(x^{(a)})<E(x)$: indeed, since $E(x)>0$, there must be a $01$ transition in $x$. Then, updating $01$ to $11$, the energy decreases by $1/T$. Therefore, $U_E\le -1$. Thus, we have proved that

\begin{align*}
 \operatorname{gap}\big(\widetilde{\mathcal{L}}^{(0,\infty)}_{{H}_{\operatorname{clock}}}\big)=\min\big\{\operatorname{gap}\big(\widetilde{\mathcal{L}}^{(0,\infty)}_{{H}_{\operatorname{clock}},\operatorname{diag}}\big),1\big\}=\min\left\{2\Big(1-\cos\Big(\frac{\pi}{T+1}\Big)\Big),1\right\}\,.
\end{align*}

\end{proof}

\subsection{Proof of Theorem \ref{th:gibbsQP}} \label{thmCproof}
Let us now put all the technical ingredients together. We show that there exists a choice of parameters that ensures that the Lindbladian $\mathcal{L}^{(\sigma_E,\beta)\dagger}_{H_C}$ has a gap that is at least polynomially small and that the Gibbs state at temperature $\beta$ has constant overlap with the ground state of $H_C$. This ensures that we can obtain the output of the circuit $C$ from a constant number of rounds of local measurements on the Gibbs state. This concludes the proof of the inclusion of $\textrm{BQP}$ in $\textrm{GibbsQP}$.

Without loss of generality, we assume that $\cos\Big(\frac{\pi}{T+1}\Big)>1/2$, i.e. $T>3$. In Proposition \ref{propgap0inftyclock}, we have proven that
\begin{align*}
\operatorname{gap}\big(\widetilde{\mathcal{L}}^{(0,\infty)}_{{H}_{\operatorname{clock}}}\big)\ge \min\left\{2\Big(1-\cos\Big(\frac{\pi}{T+1}\Big)\Big),1\right\}\ge 2\Big(1-\cos\Big(\frac{\pi}{T+1}\Big)\Big)\ge \frac{1}{T^2}\,.
\end{align*}
With Lemma \ref{le:lemmainftytobeta}, and using that $\max_{a,\nu}\|A^a_\nu\|_\infty\le 1$, $m=4(T+n)$, $M= T^3$, as well as $\Delta_E(H_{\operatorname{clock}})\ge \Delta_\nu(H_{\operatorname{clock}})\ge 1/T$ (cf. Proposition \ref{prop:properties_hc}), we get that  
\begin{align*}
\beta>2(n+T)\ln\left(2(n+T)^9\right)\Rightarrow \operatorname{gap}\big(\mathcal{{L}}^{(0,\beta)\dagger}_{{H}_{\operatorname{clock}}}\big)\ge \frac{1}{2T^2}\,.
\end{align*}
Combining this bound with Lemma \ref{le:perturbsigmaE}, we find that for $\beta,\sigma_E$  such that $\sigma_E<1$ and
\begin{align*}
\beta>16(n+T)\ln(64(n+T)^{15})\Rightarrow \operatorname{gap}\big(\mathcal{{L}}^{(\sigma_E,\beta)\dagger}_{{H}_{\operatorname{clock}}}\big)\ge \frac{1}{4T^2}\,.
\end{align*}
Let us assume that $\norm{H_1}_\infty \le \norm{H_0}_\infty \le J_\text{clock} T= T$. 
 With Corollary \ref{propperturbHH0}, taking $\kappa=6$ and letting $\beta=T^\delta/\sigma_E$ for some $\delta>1$ to be fixed later, 
 \begin{align*}
\frac{\|H_1\|_\infty}{\sigma_E}\le \frac{1}{600000\delta (n+T)^3\ln(n+T)} \quad \text{and} \quad  \beta \le \frac{1}{192}(n+T)^{6\delta-4} 
\Rightarrow \operatorname{gap}(\mathcal{L}_{H_C}^{(\sigma_E,\beta)\dagger})\ge \frac{1}{8T^2}\,.
 \end{align*}
 By construction, we have that $\|H_1\|_\infty\le \|H_{\operatorname{prop}}\|_\infty+\|H_{\operatorname{in}}\|_\infty=\mathcal{O}(J_{\operatorname{prop}}T^2+J_{\operatorname{in}}nT)$. Therefore, the above conditions  on $\|H_1\|_\infty$ can be enforced by choosing both
\begin{align*}
	J_{\operatorname{prop}},J_{\operatorname{in}}\le\frac{\sigma_E}{1200000(n+T)^5\delta\ln(n+T)}\,.
\end{align*}
Moreover, in order for the evolution to converge to a Gibbs state $\sigma_\beta$ with large enough overlap with the ground state $\ket{\eta_0}$ of $\sigma_{\beta,H_C}$, we saw in Corollary \ref{overlapgibbsground} that it is sufficient to require that $J_{\operatorname{prop}}=\Omega(J_{\operatorname{in}}(nT)^2)$ and
\begin{align*}
	\beta=\Omega\left(\frac{n+T}{J_{\operatorname{in}}}\right)\,.
\end{align*}
To reconcile all these conditions, we hence take, for $\sigma_E<1$,
\begin{align*}
	J_{\operatorname{prop}}=\mathcal{O}\left(\frac{\sigma_E}{(n+T)^5\delta\ln(n+T)}\right),\quad J_{\operatorname{in}}=\mathcal{O}\left(\frac{\sigma_E}{(n+T)^{9}\delta\ln(n+T)}\right)\,,
\end{align*}
so that 
\begin{align*}
 \beta=\Omega\left(\frac{(n+T)^{10
 }\delta\ln(n+T)}{\sigma_E}\right)\,.
\end{align*}
Since we had chosen $\beta\sigma_E=T^\delta$, this imposes that $\delta\leq 11$. Finally, $\sigma_E$ is chosen so that $\beta=\mathcal{O}\big((n+T)^{6\delta-4}\big)$, so it can even be chosen to be a constant $\sigma_E=\mathcal{O}(1)$. The result follows.

\end{document}